\documentclass[10pt,authoryear]{article}
\usepackage{appendix}
\usepackage{setspace}
\usepackage[font=small,  margin=2cm]{caption}
\usepackage[tbtags]{amsmath}
\usepackage{amsthm,amssymb,amsfonts}
\usepackage{natbib}
\usepackage{multirow}
\usepackage{lscape}
\usepackage{subfigure}
\usepackage{makecell}
\usepackage{exscale}
\usepackage{booktabs}
\usepackage{array}
\usepackage{fullpage}
\usepackage{url}
\usepackage{algorithm}
\usepackage{algpseudocode}
\usepackage{bm}
\usepackage{smile}
\usepackage{mathtools}
\usepackage{wrapfig}
\usepackage{lipsum}
\usepackage{mathrsfs}
\usepackage{relsize}
\usepackage{dsfont}
\usepackage{multirow}
\usepackage{apalike}
\usepackage{chngcntr}
\usepackage{graphicx}
\usepackage{algorithmicx}
\usepackage{algpseudocode}
\usepackage[usenames,dvipsnames,svgnames,table]{xcolor}
\usepackage[colorlinks=true,
linkcolor=blue,
urlcolor=blue,
citecolor=blue]{hyperref}

\numberwithin{equation}{section}
\numberwithin{asmp}{section}
\numberwithin{theorem}{section}
\numberwithin{corollary}{section}
\numberwithin{definition}{section}
\renewcommand{\baselinestretch}{1.5}
\begin{document}

\title{\LARGE An Efficient Iterative Least Squares Algorithm for Large-dimensional Matrix Factor Model via Random Projection}

\author{
	Yong He\footnotemark[1],~~Ran Zhao\footnotemark[1],~~and~Wen-Xin Zhou\footnotemark[2]
	}
\renewcommand{\thefootnote}{\fnsymbol{footnote}}
\footnotetext[1]{Institute of Financial Studies, Shandong University, China. E-mail: {\tt heyong@sdu.edu.cn, Zhaoran@mail.sdu.edu.cn}}
\footnotetext[2]{Department of Mathematical Sciences, University of California, San Diego, USA. E-mail: {\tt wez243@ucsd.edu}}
\maketitle

The matrix factor model has drawn growing attention for its advantage in achieving two-directional dimension reduction simultaneously for matrix-structured observations.	In this paper, we propose a simple iterative least squares algorithm for matrix factor models, in contrast to the Principal Component Analysis (PCA)-based methods in the literature. In detail, we first propose to estimate the latent factor matrices by projecting the observations with two deterministic weight matrices, which are chosen to diversify away the idiosyncratic components. We show that the inferences on factors are still asymptotically valid even if we overestimate both the row/column factor numbers.
We then estimate the row/column loading matrices by minimizing the squared loss function under certain identifiability conditions. The resultant estimators of the loading matrices are treated as the new weight/projection matrices and thus the above update procedure can be iteratively performed until convergence. Theoretically, given the true dimensions of the factor matrices, we derive the convergence rates of the estimators for loading matrices and common components at any  $s$-th step iteration. Additionally, we propose an eigenvalue-ratio method to estimate the pair of factor numbers consistently. Thorough numerical simulations are conducted to investigate the finite-sample performance of the proposed methods and two real datasets associated with financial portfolios and multinational macroeconomic indices are used to illustrate our algorithm's  practical usefulness.
\vspace{2em}

	
	\section{Introduction}

Factor modeling is an extremely popular approach for  dimension reduction in large-dimensional time series analysis, which has been successfully applied to large panels of time series for forecasting macroeconomic variables \citep{stockwatson02JASA}, building low-dimensional indicators of the whole economic activity \citep{stock2002macroeconomic}. In the last two decades, there has been a flourish of literature on large-dimensional factor models; see, for example, \cite{bai2003inferential},\cite{onatski09}, \cite{ahn2013eigenvalue}, \cite{fan2013large},  \cite{Trapani2018A}, \cite{Sahalia2017Using}, \cite{kong2017number}, \cite{Barigozzi2018Simultaneous}, \cite{yu2019robust}, \cite{barigozzi2020consistent}, \cite{Chen2021Quantile}, \cite{He2020large} and  \cite{fan2022learning}.

In economics and finance,  observations are usually well structured to be an array/matrix, such as a time list of tables recording several macroeconomic variables across a number of countries or a series of customers' ratings on a large number of items in an online platform.
In the last few years, the literature has paid increasing attention to factor analysis for matrix time series.  \cite{wang2019factor} for the first time proposed the following factor model for matrix time series:
\begin{equation}\label{MFM}
		\Xb_{t}=\Rb\Fb_{t}\Cb^{\top}+\Eb_{t}, \ t=1,\dots,T,
	\end{equation}
    where $\Rb$ is the $p_{1} \times k_{1}$ row factor loading matrix exploiting the variations of $\Xb_{t}$ across the rows, $\Cb$ is the $p_{2} \times k_{2}$ column factor loading matrix reflecting the differences in the columns of $\Xb_{t}$, $\Fb_{t}$ is the $k_{1} \times k_{2}$ common factor matrix and $\Eb_{t}$ is the idiosyncratic component.
\cite{wang2019factor} proposed
estimators of the factor loading matrices by an eigen-analysis of the auto-cross-covariance matrix;  \cite{fan2021}  proposed an $\alpha$-PCA method by exploiting an eigen-analysis
of a weighted average of the mean and the column (row) covariance matrix of the data;
\citet{Yu2021Projected} proposed a Projection Estimation (PE) method that further improved the estimation efficiency of the factor loading matrices. \cite{He2021Matrix} established the equivalence between minimizing the squared loss and the PE method by \citet{Yu2021Projected} and further proposed a robust method  by replacing the squared loss with the Huber loss. The resultant estimators of factor loading matrices can be simply obtained by an eigen-analysis of weighted sample covariance matrices of the projected data. \cite{he2022matrix} proposed to recover the loading
spaces of the matrix elliptical factor model by an eigen-analysis of the generalized row/column matrix Kendall's tau, which generalizes the multivariate Kendall's tau to the random matrix setting. However, to our knowledge, all the theoretical studies of the matrix factor model in the literature crucially rely on the assumption that the pair of factor numbers $k_1$ and $k_2$ is consistently estimated, which typically requires that the factors are relatively strong, data have weak serial correlation or the number of observations is large. In practical applications, these  requirements may fail to hold due to weak signal-to-noise ratio or  non-stationarity, making the first top eigenvalues of the row/column covariance matrix less separated from the remaining ones; see also the discussions in \cite{fan2022learning} for vector factor models. Over-estimating the number of factors would be a promising remedy as discussed in \cite{moon2015linear,westerlund2015cross,barigozzi2020consistent} for classical vector factor models. The impact of over-estimating the pair of factor numbers for matrix factors models remains unknown.


In this article, we propose a simple iterative least squares algorithm for matrix factor models, in contrast to the Principal Component Analysis (PCA)-based methods in the literature.
In detail, in the first step, we  propose to estimate the latent factor matrices by projecting the matrix observations with two deterministic weight matrices, which are chosen to diversify away the idiosyncratic components. This idea is similar in spirit to that in \cite{fan2022learning} and the estimator does not rely on eigenvectors.
In the second step,  we update the row/column loading matrices by minimizing the squared loss function under the identifiability condition. Then the estimators of the loading matrices are treated as the new weight matrices and we iteratively proceed with the above two steps  until a convergence criterion is reached. The contributions of the current work lie in the following aspects. Firstly,
to our knowledge, this is the first work on matrix factor analysis that does not involve any eigen-decomposition of large matrices. The proposed iterative least squares algorithm is quite simple and computationally efficient, with computational complexity $O\left(Tp_{1}p_{2}\right)$ in contrast to the typical $O\left(Tp_{1}^{2}+Tp_{2}^{2}\right)$ complexity of the eigen-decomposition based methods. Secondly, we show that even if both numbers of row and column factors are over-estimated,  the inferences on factor scores are still asymptotically valid, which is new to the literature on matrix factor models. Thirdly, given the factor numbers are correctly specified, we establish the convergence rates of the estimators for loading matrices and common components at the $s$-th iteration for any $s\geq 1$ under some strong factor identifiability condition. Compared to one-step estimation \citep{Yu2021Projected}, the multi-step estimator is proven to be less sensitive to the initial estimator.
 In addition, the iterative least squares algorithm also reduces the magnitudes of the idiosyncratic error components in each step, thereby increasing the signal-to-noise ratio and enjoying the same advantage as the projection estimation method by \cite{Yu2021Projected}. At last, we introduce an eigenvalue-ratio method to determine the number of factors, which is shown to be consistent.

 The rest of the article is organized as follows. In Section 2, we introduce the factor estimation method via two-directional diversified projections. We derive the consistency of the vectorized factor space even when the factor numbers are over-estimated. In Section 3, we propose the iterative least squares estimators for the loading spaces and derive the convergence rates of the estimators at any  $s$-th step iteration.  We also present the consistency of our model selection criterion. In Section 4, we conduct thorough numerical studies to illustrate the advantages of the proposed methods  over the state-of-the-art methods. In Section 5, we analyze a financial dataset and a multinational macroeconomic indices dataset to illustrate the empirical usefulness of the proposed methods. We discuss possible future research directions and conclude the paper in Section 6. The proofs of the main theorems and additional details are collected in the supplementary materials.

We end this section by introducing some notations that will be used throughout the paper. For any vector $\bmu=(\mu_1,\ldots,\mu_p)^\top \in \RR^p$, let $\|\bmu\|_2=(\sum_{i=1}^p\mu_i^2)^{1/2}$, $\|\bmu\|_\infty=\max_i|\mu_i|$. For a real number $a$, denote  $[a]$ as the largest integer smaller than or equal to $a$. $\otimes$ denotes the Kronecker product. For a matrix $\Ab$, let $\mathrm{A}_{ij}$ (or $\mathrm{A}_{i,j}$) be the $(i,j)$-th entry of $\Ab$, $\Ab^\top$ the transpose of $\Ab$, ${\rm tr}(\Ab)$ the trace of $\Ab$, $\text{rank}(\Ab)$ the rank of $\Ab$, $\Ab^+$ the  Moore-Penrose generalized inverse of $\Ab$ and $\text{diag}(\Ab)$ a vector composed of the diagonal elements of $\Ab$. Let $\text{Vec}(\Ab)$ be the vector obtained by stacking the columns of $\Ab$. Denote $\Pb_{\Ab}$ as the projection matrix $\Pb_{\Ab}=\Ab(\Ab^\top\Ab)^{-1}\Ab^\top$ and $\text{span}(\Ab)$  as the space spanned by the columns of $\Ab$. Denote $\lambda_j(\Ab)$ as the $j$-th largest eigenvalue of a nonnegative definitive matrix $\Ab$, and let $\|\Ab\|_{2}$ be the spectral norm of matrix $\Ab$ and $\|\Ab\|_F$ be the Frobenius norm of $\Ab$. For two series of random variables, $X_n \lesssim Y_n$ means that $X_n=O_p\left(Y_n\right)$, and $X_n \gtrsim Y_n$ means that $Y_n=O_p\left(X_n\right)$. The notation $X_n$ and $Y_n$, $X_n\asymp Y_n$ means $X_n=O_p(Y_n)$ and $Y_n=O_p(X_n)$. The constants $c, C_1, C_2$ in different lines can be different.

\section{Factor Estimation via Two-directional Diversified Projections}\label{sec:factormatrix}
In this section, we introduce a simple two-directional diversified projection method to estimate the factor score matrices. We also investigate the theoretical properties of the estimators under  cases of finite samples and overestimation of the factor numbers.
	\subsection{The estimation of factors}
	In this section, we propose a simple   two-directional diversified projection method to estimate the factor score matrices, which does not involve eigen-decomposition. Let $\Wb_{i}=(\bw_{i,.1},\bw_{i,.2},\dots,\bw_{i,.m_{i}})$ be a given exogenous (or deterministic) $p_{i} \times m_{i}$ matrix $i=1,2$, where $\bw_{i,.j} $, the $j$-th column of the matrix $\Wb_{i}$, is  a vector of ``diversified weights" in the sense that its strength should be approximately equally distributed across most of its components.
We call $\Wb_{1}$ and $\Wb_{2}$ the ``left  projection matrix'' and ``right  projection matrix'', respectively.
We propose to estimate $\Fb_{t}$  simply by
	
\begin{equation}\label{equ:fhat}
\hat{\Fb}_{t}=\dfrac{1}{p_{1}p_{2}}\Wb_{1}^{\top}\Xb_{t}\Wb_{2}.
\end{equation}
	By the matrix factor model (\ref{MFM}), we have
\begin{equation}\label{equ:fhat}
\hat{\Fb}_{t}=\dfrac{1}{p_{1}p_{2}}\Wb_{1}^{\top}\Rb\Fb_{t}\Cb^{\top}\Wb_{2}+\dfrac{1}{p_{1}p_{2}}\Wb_{1}^{\top}\Eb_{t}\Wb_{2}:=\Hb_{1}\Fb_{t}\Hb_{2}^{\top}+\bm{\cE}_{t},
\end{equation}
	where $\Hb_{1}=\Wb_{1}^{\top}\Rb/{p_{1}}$, $\Hb_{2}=\Wb_{2}^{\top}\Cb/{p_{2}}$ and $\bm{\cE}_{t}=\Wb_{1}^{\top}\Eb_{t}\Wb_{2}/{(p_{1}p_{2})}$. Thus $\hat{\Fb}_{t}$ estimates $\Fb_{t}$ up to two affine transformation with $\bm{\cE}_{t}$ as the estimation error. The assumption that $\Wb_{1},\Wb_{2}$ should be diversified guarantees that as $\min(p_{1},p_{2})\to \infty$, $\bm{\cE}_{t}$ is diversified away (converging to zero in probability). We call the new factor matrix estimator ``bi-diversified factors", which  reduces the dimensions of $\Xb_{t}$ from $p_1\times p_2$ to $m_1\times m_2$ by two-directional diversified projections. Due to the clean expansion (\ref{equ:fhat}), the mathematics for theoretical analysis is much simpler than the most benchmark
estimators. Intuitively, $\hat{\Fb}_{t}$ would lead to valid inferences in
factor-augmented models so long as $m_i\geq k_i, i=1,2$, in the same spirit as \cite{fan2022learning} for vector factor model,
and we leave this to future work as the matrix factor-augmented model is still in its infancy.

	\subsection{Theoretical analysis for the estimators of Factors}
In this section, we investigate the theoretical properties of the estimators  of factors under finite sample cases and  the overestimation of the factor numbers. We assume that
   the predetermined constants $m_{i},i=1,2$  do not grow with $p_{i}$, which are named as ``the working (pseudo) numbers of row (column) factors''. Since in practice we do not know the true number of factors, we often take slightly large numbers $m_{i}$ such that $m_{i} \geq k_{i}$ are likely to hold. Let $\Wb_{i}$  be either deterministic or random but independent of the $\sigma$-algebra generated by $\{\Eb_{t}:t\in [T]\}$. We further assume that the projection matrices $\Wb_{i}, i=1,2$ satisfy the following:
	
	\begin{asmp} \label{Assumption 1}
		There are positive constants $c_{1}$ and $c_{2}$, such that (almost surely if $\Wb_{1}, \Wb_{2}$ are random) as $p_{1}, p_{2} \to \infty$,
		
		(1) $\max \limits_{i \leq p_{1}}|w_{1,ij}| < c_{1}$, $\max \limits_{i \leq p_{2}}|w_{2,ij}| < c_{2}$;
		
		(2) the $m_{1} \times m_{1}$ matrix $\Wb_{1}^{\top}\Wb_{1}/p_{1}$ and $m_{2} \times m_{2}$ matrix $\Wb_{2}^{\top}\Wb_{2}/p_{2}$ satisfy $\lambda_{\min}(\Wb_{1}^{\top}\Wb_{1}/p_{1}) \geq c_{1}$,  $\lambda_{\min}(\Wb_{2}^{\top}\Wb_{2}/{p_{2}}) \geq c_{2}$;
		
		(3) $\Wb_{1},\Wb_{2}$ are independent of $\Eb_{t},t \in [T]$.
	\end{asmp}
	
	Vectorizing the matrix $\hat{\Fb}_{t}$ in (\ref{equ:fhat}), we have $$\text{Vec}(\hat{\Fb}_{t})=(\Hb_{2}\otimes\Hb_{1})\text{Vec}(\Fb_{t})+\dfrac{1}{p_{1}p_{2}}(\Wb_{2}\otimes\Wb_{1})^{\top}\text{Vec}(\Eb_{t}):=\Hb\text{Vec}(\Fb_{t})+\dfrac{1}{p_{1}p_{2}}\Wb^{\top}\text{Vec}(\Eb_{t}),$$ where $\Hb=\Hb_{2}\otimes\Hb_{1}$ and $\Wb=\Wb_{2}\otimes\Wb_{1}$.
Therefore,  $\text{Vec}(\hat{\Fb}_{t})$ can be treated as an estimate of $\text{Vec}(\Fb_{t})$ up to  a transformation matrix, where $\Hb \in \RR^{m_{1}m_{2} \times k_{1}k_{2}}$, with estimation error equal to $\Wb^{\top}\text{Vec}(\Eb_{t})/(p_{1}p_{2})$. When each element of $\Eb_{t}$ is cross-sectional weakly dependent, Assumption \ref{Assumption 1} guarantees the cross-sectional central limit theorem of $\Wb^{\top}\text{Vec}(\Eb_{t})/(p_{1}p_{2})$. For example, assume each element of $\Eb_{t}$ is cross-sectional independent, under Assumption \ref{Assumption 1}, as $p_{1}p_{2} \to \infty$, we get

 \begin{equation}\label{equ:ane}
 	\dfrac{1}{\sqrt{p_{1}p_{2}}}\Wb^{\top}\text{Vec}(\Eb_{t}) \stackrel{d}{\rightarrow} N(\bm{0},\Vb),
\end{equation}
where $\Vb=\lim_{p_{1}p_{2}\to \infty}\Wb^{\top}\text{Var}(\text{Vec}(\Eb_{t}))\Wb/(p_{1}p_{2})$, assuming it exists. The convergence (\ref{equ:ane}) shows that for each $t \leq T$, $\sqrt{p_{1}p_{2}}(\text{Vec}(\hat{\Fb}_{t})-\Hb\text{Vec}(\Fb_{t}))$ is asymptotically normal regardless of whether $T$ goes to infinity,  $m_{1}=k_{1}, m_{2}=k_{2}$ or not. It only requires $p_{1}p_{2}\to \infty$, which is particularly useful for analyzing short matrix time series.
	
	In addition, the factor components should not be diversified away, which entails the following conditions on transformation matrix $\Hb_{i}, i=1,2$. Let $\nu_{\min}(\Hb_{i}), \nu_{\max}(\Hb_{i})$ denote the minimum and maximum nonzero singular values of $\Hb_{i}$, respectively.
	
	\begin{asmp} \label{Assumption 2}
		Suppose $m_{1} \geq k_{1}, m_{2} \geq k_{2}$. Almost surely
		
		(1) $\text{rank}(\Hb_{1})=k_{1}$,  $\text{rank}(\Hb_{2})=k_{2}$;
		
		(2) there exist constants $c_{1}$ and $c_{2}$ such that $$\nu_{\min}^{2}(\Hb_{i}) \gg \dfrac{1}{p_{i}}, \, \nu_{\max}(\Hb_{i}) \leq c_{1}\nu_{\min}(\Hb_{i}), \ \ i=1,2.$$
	\end{asmp}
	
	Assumption \ref{Assumption 2} (1) requires $\Wb_{1}$ to have at least $k_{1}$ columns that are not orthogonal to $\Rb$ and $\Wb_{2}$ to have at least $k_{2}$ columns that are not orthogonal to $\Cb$ so that $\Rb$ and $\Cb$ are not diversified away. It also ensures that the space spanned by $\text{Vec}(\hat{\Fb}_{t})$ is asymptotically equal to the space spanned by $\text{Vec}(\Fb_{t})$. This is the key assumption imposed, but not stringent in the context of over-estimating factors \citep{barigozzi2020consistent,fan2022learning}. Assumption \ref{Assumption 2} (2) determines the rate of convergence in recovering the space spanned by the factors and ensures that the weight matrix and the loading matrix are not orthogonal. Assume $\nu_{\min}(\Wb_{1}^{\top}\Rb)=p_{1}^{\alpha_{1}},  \nu_{\min}(\Wb_{2}^{\top}\Cb)=p_{2}^{\alpha_{2}}$,  then  Assumption \ref{Assumption 2} (2) entails that $\alpha_{1}, \alpha_{2} \geq 1/2$.	
	Assumption \ref{Assumption 1} and Assumption \ref{Assumption 2} are direct generalizations of Assumptions  2.1 and 2.2 in \cite{fan2022learning} to the matrix factor models.

	For the matrix factor model, once we vectorize the observations, the model reduces to a vector factor model
	$$\Xb=\Fb(\Cb \otimes \Rb)^{\top}+\Eb,$$
	where $\Fb=
		(\text{Vec}(\Fb_{1}), \text{Vec}(\Fb_{2}), \ldots, \text{Vec}(\Fb_{T}))^{\top}$,
	$\Xb=(\text{Vec}(\Xb_{1}), \text{Vec}(\Xb_{2}), \ldots, \text{Vec}(\Xb_{T}))^{\top}$.
	$
	\Eb=(\text{Vec}(\Eb_{1}), \text{Vec}(\Eb_{2}),$ $ \ldots, \text{Vec}(\Eb_{T}))^{\top}$. By (\ref{equ:fhat}), we have
	$$\hat{\Fb}=\Fb\Hb^{\top}+\bm{\cE},
	$$
	where
	$\hat{\Fb}=(\text{Vec}(\hat{\Fb}_{1}), \text{Vec}(\hat{\Fb}_{2}), \ldots, \text{Vec}(\hat{\Fb}_{T}))^{\top}$, $\bm{\cE}=(\text{Vec}(\bm{\cE}_{1}), \text{Vec}(\bm{\cE}_{2}),\ldots, \text{Vec}(\bm{\cE}_{T}))^{\top}$.
	To derive the theoretical properties of the factor estimators, we further impose the following assumptions.
	
	\begin{asmp} \label{Assumption 3}
	There are constants $c,C$ such that\\
	(1) $p_2^{-1}\|\EE(\Eb_{t}\Eb_{t}^{\top})\|_{2}\leq c$; \\
	(2) $ 1/{T}\sum_{s=1}^{T}\sum_{T=1}^{T}\EE\|\Fb_{t}\|_{F}\|\Fb_{s}\|_{F}\text{tr}(\EE(\Eb_{t}\Eb_{s}^{\top}|\Fb)) < c$;\\	
	(3) $\lambda_{\min}\left(\EE\left(\Eb^{\top}\Eb/{T}\right) \right) \geq c_{0}, \EE\|\EE\left(\text{Vec}(\Eb_{t})\text{Vec}(\Eb_{t})^{\top}|\Fb\right)\|_{2} <c$; \\
	(4) $c < \lambda_{\min}(\sum_{t=1}^{T}\text{Vec}(\Fb_{t})\text{Vec}(\Fb_{t})^{\top}/{T}) \leq \lambda_{\max}(\sum_{t=1}^{T}\text{Vec}(\Fb_{t})\text{Vec}(\Fb_{t})^{\top}/{T}) < C$, a.s.;\\	
	(5) for any $t\in[T], i_{1},i_{1} \in [p_{1}], j_{1},j_{2} \in [p_{2}]$,  $\sum_{s=1}^{T}\sum_{u_{1},v_{1}=1}^{p_{1}}\sum_{u_{2},v_{2}=1}^{p_{2}}|\text{Cov}(e_{t,i_{1}i_{2}}e_{t,j_{1}j_{2}},e_{s,u_{1}u_{2}}$ $e_{s,v_{1}v_{2}})| \leq c$.
\end{asmp}

Assumption \ref{Assumption 3} requires weak dependence among idiosyncratic errors, which is common in the literature; see, for example, \cite{He2021Vector}, \cite{fan2021} and \cite{Yu2021Projected}. The following theorem establishes the convergence of $\hat{\Fb}_{t}$ to the true ${\Fb}_{t}$ up to transformation matrices, regardless of whether $T\rightarrow\infty$ or overestimating the factor numbers.
    \begin{theorem} \label{Convergence}
   	Suppose Assumptions \ref{Assumption 1} (1), \ref{Assumption 1} (3),  \ref{Assumption 2} and \ref{Assumption 3} (1) hold. As  $\min\{p_{1},p_{2}\} \to \infty$ and $T$ is finite or $T \to \infty$. Then for any bounded $m_{1} \geq k_{1}$, $m_{2} \geq k_{2}$, for any $t \in [T]$,
   	$$\big\|\Mb_{1}^{\top}\hat{\Fb}_{t}\Mb_{2}-\Fb_{t}\big\|_{2}=O_{p}\left(\dfrac{1}{\sqrt{p_{1}}}\nu_{\min}^{-1}(\Hb_{1})\nu_{\min}^{-1}(\Hb_{2})\right),$$
   	where $\Mb_{1}=(\Hb_{1}\Hb_{1}^{\top})^{+}\Hb_{1}\in \RR^{m_{1} \times k_{1}}$, $\Mb_{2}=(\Hb_{2}\Hb_{2}^{\top})^{+}\Hb_{2}\in \RR^{m_{2} \times k_{2}}$.
   \end{theorem}

The following theorem shows that the linear space spanned by $\hat{\Fb}$  equals to the linear space spanned by $\Fb$ asymptotically.
	
	\begin{theorem} \label{Vec}
		Suppose Assumptions \ref{Assumption 1}-\ref{Assumption 3} hold. For any bounded $m_{1} \geq k_{1}$, $m_{2} \geq k_{2}$, we have as $\min\{p_{1},p_{2}\} \to \infty$ that
		$$\|\Pb_{\hat{\Fb}}\Pb_{\Fb}-\Pb_{\Fb}\|_{2}=O_{p}\left(\dfrac{1}{\sqrt{p_{1}p_{2}}}\nu_{\min}^{-1}(\Hb_{1})\nu_{\min}^{-1}(\Hb_{2})\right),$$
		$$\|\Pb_{\hat{\Fb}\Mb}-\Pb_{\Fb}\|_{2}=O_{p}\left(\dfrac{1}{\sqrt{p_{1}p_{2}}}\nu_{\min}^{-1}(\Hb_{1})\nu_{\min}^{-1}(\Hb_{2})\right),$$
		where  $\Mb=(\Hb\Hb^{\top})^{+}\Hb$ and $\Hb=\Hb_{2}\otimes\Hb_{1}$.
		
	\end{theorem}
	
	Theorem \ref{Vec} establishes that when $m_{1} \geq k_{1}$, $m_{2} \geq k_{2}$, the linear space spanned by $\hat{\Fb}$ is asymptotically the same as the linear space spanned by $\Fb$. Theorem \ref{Vec} also shows that a particular subspace of $\text{span}{(\hat\Fb)}$ is consistent for $\text{span}{(\Fb)}$. In particular, when $m_{1}=k_{1}, m_{2}=k_{2}$, we have $\Pb_{\hat{\Fb}\Mb}=\Pb_{\hat{\Fb}}$ as $\Mb$ is invertible. It then degenerates to the usual space consistency.
	
Consistent estimation of the number of factors $k_{1}, k_{2}$ typically requires strong conditions, which are difficult to fulfill in finite samples case. One  advantage of the proposed method for estimating the factor matrices is that it is still robust against overestimating the number of factors in many statistical inference problems.
	As for the choices of weight matrices, one can select $\Wb_{1}$ and $\Wb_{2}$ following the same strategies by \cite{fan2022learning}, such as the Hadamard Projection. By Theorem \ref{Vec}, $\{\hat{\Fb}_{t}, t=1,\ldots,T\}$ would lead to valid inferences in
factor-augmented models so long as $m_i\geq k_i, i=1,2$
and we leave this to our future work as the matrix factor-augmented model is still in its infancy.

\section{Iterative Least Squares Estimators for Loading Spaces}
In this section, we introduce the iterative least squares estimators for the column/row loading spaces. We also derive the  convergence rates of the estimators for loading matrices at the $s$-th iteration (for any $s\geq 1$) provided that the pair of factor numbers are correctly specified. In case that the pair of factor numbers are unknown, we propose a novel eigenvalue-ratio method  to estimate $k_1$ and $k_2$.

\subsection{The Estimation of Loading Spaces}

In Section \ref{sec:factormatrix}, we introduce the way to estimate the factor matrices with two diversified projection matrices $\Wb_1,\Wb_2$ and denote the estimators in (\ref{equ:fhat}) as  $\{\hat{\Fb}_{t}^{(1)}, t=1,\ldots,T\}$, which are of dimension $m_1\times m_2$ with $m_1\geq k_1,m_2\geq k_2$.
Given $\{\hat{\Fb}_{t}^{(1)}\}$, it is straightforward to estimate the row factor loading matrix $\Rb$ by minimizing the squared Frobenius loss under the identifiability condition:

\begin{equation}\label{M2.1}
	\begin{aligned}
		& \min \limits_{\Rb} L_{1}(\Rb)=\min \limits_{\Rb} \dfrac{1}{T}\sum_{t=1}^{T}\|\Xb_{t}-\Rb\hat{\Fb}_{t}^{(1)}\Wb_{2}^{\top}\|_{F}^{2},\\
		& \text{s.t}~~	\dfrac{1}{p_{1}}\Rb^{\top}\Rb=\Ib_{m_{1}}, m_{1} \geq k_{1},	
	\end{aligned}
\end{equation}
where $\Wb_2$ is the column projection matrix.
The objective function  in (\ref{M2.1}) can be simplified as $$L_{1}(\Rb)= \dfrac{1}{T}\sum_{t=1}^{T}\Big[\text{tr}\big(\Xb_{t}^{\top}\Xb_{t}\big)-2\text{tr}\big(\Xb_{t}^{\top}\Rb\hat{\Fb}_{t}^{(1)}\Wb_{2}^{\top}\big)+p_{1}\text{tr}\big(\Wb_{2}\hat{\Fb}_{t}^{(1)\top}\hat{\Fb}_{t}^{(1)}\Wb_{2}^{\top}\big)\Big],$$
and the Lagrangian function is
$$\min \limits_{\Rb,\bTheta}\cL_{1}(\Rb,\bTheta)=\min \limits_{\Rb,\bTheta}\left\{L_{1}(\Rb)+ \text{tr}\left
[\bTheta\left(\dfrac{1}{p_{1}}\Rb^{\top}\Rb-\Ib_{m_{1}}\right)\right]\right\},$$
where the Lagrangian multipliers $\bTheta$ is a symmetric matrix.
Taking  $\partial \cL_{1}(\Rb,\bTheta)/ \partial \Rb =0$ and $\partial \cL_{1}(\Rb,\bTheta) / \partial \bTheta =0$, we obtain
\begin{equation}\label{equ:derivative}
	\dfrac{\partial \cL_{1}(\Rb,\bTheta)}{\partial \Rb}= \dfrac{1}{T}\sum_{t=1}^{T}\big(-2\Xb_{t}\Wb_{2}\hat{\Fb}_{t}^{(1)\top}+\dfrac{2}{p_{1}}\Rb\bTheta\big)=0,\ \ \dfrac{\partial \cL_{1}(\Rb,\bTheta)}{\partial \bTheta}=\dfrac{1}{p_{1}}\Rb^{\top}\Rb-\Ib_{m_{1}}=0.
\end{equation}
Further, we can derive the explicit expression for $\hat\bTheta$ and $\hat{\Rb}$ satisfying (\ref{equ:derivative}), that is,
$$\hat \bTheta=\dfrac{\sqrt{p_{1}}}{T}\left[\left(\sum_{t=1}^{T}\hat{\Fb}_{t}^{(1)}\Wb_{2}^{\top}\Xb_{t}^{\top}\right)\left(\sum_{t=1}^{T}\Xb_{t}\Wb_{2}\hat{\Fb}_{t}^{(1)\top}\right)\right]^{1/2},$$
\begin{equation}\label{Rhat1}
	\hat{\Rb}^{(1)}=\sqrt{p_{1}}\left(\sum_{t=1}^{T}\Xb_{t}\Wb_{2}\hat{\Fb}_{t}^{(1)\top}\right)\left[\left(\sum_{t=1}^{T}\hat{\Fb}_{t}^{(1)}\Wb_{2}^{\top}\Xb_{t}^{\top}\right)\left(\sum_{t=1}^{T}\Xb_{t}\Wb_{2}\hat{\Fb}_{t}^{(1)\top}\right)\right]^{-1/2},
\end{equation}
and we set $\hat{\Rb}^{(1)}$ as the one-step estimator of  the row loading matrix. Similarly, once we get the estimator $\hat{\Rb}^{(1)}$, we can obtain $\hat{\Cb}^{(1)}$ by minimizing the following loss function under the identifiability condition:
$$\begin{aligned}
	& \min \limits_{\Cb} L_{2}(\Cb)=\min \limits_{\Cb} \dfrac{1}{T}\sum_{t=1}^{T}\big\|\Xb_{t}-\hat{\Rb}^{(1)}\hat{\Fb}_{t}^{(1)}\Cb^{\top}\big\|_{F}^{2},\\
	& \text{s.t}	\dfrac{1}{p_{2}}\Cb^{\top}\Cb=\Ib_{m_{2}}, m_{2} \geq k_{2},
\end{aligned}$$
with the Lagrangian function
$$\min \limits_{\Cb,\bLambda} \cL_{2}(\Cb,\bLambda)=\min \limits_{\Cb,\bLambda}\left\{\dfrac{1}{T}\sum_{t=1}^{T}\big\|\Xb_{t}-\hat{\Rb}^{(1)}\hat{\Fb}_{t}^{(1)}\Cb^{\top}\big\|_{F}^{2}+\text{tr}\left[\bLambda\left(\dfrac{1}{p_{2}}\Cb^{\top}\Cb-\Ib_{m_{2}}\right)\right]\right\},$$
we get the following estimator of the column factor loading matrix
\begin{equation}\label{Chat1} \hat{\Cb}^{(1)}=\sqrt{p_{2}}\left(\sum_{t=1}^{T}\Xb_{t}^{\top}\hat{\Rb}^{(1)}\hat{\Fb}_{t}^{(1)}\right)\left[\left(\sum_{t=1}^{T}\hat{\Fb}_{t}^{(1)\top}\hat{\Rb}^{(1)\top}\Xb_{t}\right)\left(\sum_{t=1}^{T}\Xb_{t}^{\top}\hat{\Rb}^{(1)}\hat{\Fb}_{t}^{(1)}\right)\right]^{-1/2}.
\end{equation}
Given $\hat{\Cb}^{(1)}$,  we update the estimator of $\Fb_{t}$ as $\hat{\Fb}_{t}^{(2)}=\hat{\Rb}^{(1)\top}\Xb_{t}\hat{\Cb}^{(1)}/(p_{1}p_{2})$, thereby  updating the estimators of $\Rb$ and $\Cb$. In detail,  update the estimator of $\Rb$ as $\hat{\Rb}^{(2)}$ by replacing $\hat{\Fb}_{t}^{(1)}, \Wb_{2}$ in  (\ref{Rhat1}) with $\hat{\Fb}_{t}^{(2)}, \hat{\Cb}^{(1)}$, respectively; update the estimator of  $\Cb$ as $\hat{\Cb}^{(2)}$ by replacing $\hat{\Fb}_{t}^{(1)}, \hat{\Rb}^{(1)}$ in  (\ref{Chat1}) with $\hat{\Fb}_{t}^{(2)}, \hat{\Rb}^{(2)}$, respectively;
we repeat the above steps until a convergence criterion is met.
At the $(s+1)$-th iteration, the estimators $\hat{\Fb}_{t}^{(s+1)}$, $\hat{\Rb}^{(s+1)}$ and $\hat{\Cb}^{(s+1)}$ have the following expressions:
$$\hat{\Fb}_{t}^{(s+1)}=\dfrac{1}{p_{1}p_{2}}\hat{\Rb}^{(s)\top}\Xb_{t}\hat{\Cb}^{(s)},$$
$$\hat{\Rb}^{(s+1)}=\sqrt{p_{1}}\left(\sum_{t=1}^{T}\Xb_{t}\hat{\Cb}^{(s)}\hat{\Fb}_{t}^{(s+1)\top}\right)\left[\left(\sum_{t=1}^{T}\hat{\Fb}_{t}^{(s+1)}\hat{\Cb}^{(s)\top}\Xb_{t}^{\top}\right)\left(\sum_{t=1}^{T}\Xb_{t}\hat{\Cb}^{(s)}\hat{\Fb}_{t}^{(s+1)\top}\right)\right]^{-1/2},$$
$$\hat{\Cb}^{(s+1)}=\sqrt{p_{2}}\left(\sum_{t=1}^{T}\Xb_{t}^{\top}\hat{\Rb}^{(s+1)}\hat{\Fb}_{t}^{(s+1)}\right)\left[\left(\sum_{t=1}^{T}\hat{\Fb}_{t}^{(s+1)\top}\hat{\Rb}^{(s+1)\top}\Xb_{t}\right)\left(\sum_{t=1}^{T}\Xb_{t}^{\top}\hat{\Rb}^{(s+1)}\hat{\Fb}_{t}^{(s+1)}\right)\right]^{-1/2},$$
where $\hat{\Rb}^{(s)}$ and $\hat{\Cb}^{(s)}$ are the estimators from the $s$-th iteration. The Random Projection-based Iterative Least Squares (RPILS) procedure  for the matrix factor model is summarized in Algorithm \ref{alg:1} below and the theoretical analysis is presented in the following section.

As for the convergence criterion, denote the common component matrix at the $(s+1)$-th step as $\hat{\Sbb}^{(s+1)}_{t}=\hat{\Rb}^{(s)}\hat{\Fb}_{t}^{(s+1)}\hat{\Cb}^{(s)\top}$.
In our simulation studies, the iterative procedure is terminated either when a pre-specified maximum iteration number ($\text{maxiter}=100$) is reached or  when $$\|\hat{\Sbb}^{(s+1)}_{t}-\hat{\Sbb}^{(s)}_{t}\|_{F} \leq \epsilon,$$
where $\epsilon$ is a small constant ($10^{-6}$) given in advance.
\begin{algorithm}
	\renewcommand{\algorithmicrequire}{\textbf{Input:}}
	\renewcommand{\algorithmicensure}{\textbf{Output:}}
	\caption{Random Projection based Iterative Least Squares (RPILS) procedure for matrix factor model}
	\label{alg:1}
	\begin{algorithmic}[1]
		\Require Data matrices $\{\Xb_{t}\},t \leq T$, the pair of pseudo row and column factor numbers $m_{1} (m_{1} \geq k_{1}$) and $m_{2} ( m_{2} \geq k_{2})$, the diversified projection matrices $\Wb_{1},\Wb_{2}$
		\Ensure Factor loading matrices $\tilde{\Rb}\in \RR^{p_1\times m_1}, \tilde{\Cb}\in\RR^{p_2\times m_2}$ and factor matrix $\tilde{\Fb}_{t}\in\RR^{m_1\times m_2},t \leq T$
		\State obtain the initial estimator $\hat{\Fb}_{t}^{(1)}$ by  $\hat{\Fb}_{t}^{(1)}=\Wb_{1}^{\top}\Xb_{t}\Wb_{2}/(p_{1}p_{2})$;
		\State given $\hat{\Fb}_{t}^{(1)}$ and $\Wb_{2}$, get an estimator $\hat{\Rb}^{(1)}$ by Equation (\ref{Rhat1}); given $\hat{\Fb}_{t}^{(1)}$ and $\hat{\Rb}^{(1)}$, get an estimator $\hat{\Cb}^{(1)}$ by Equation (\ref{Chat1});
		\State update $\hat{\Fb}_{t}^{(2)}$ by $\hat{\Fb}_{t}^{(2)}=\hat{\Rb}^{(1)\top}\Xb_{t}\hat{\Cb}^{(1)}/(p_{1}p_{2})$;
		\State update $\hat{\Rb}^{(2)}$ by replacing $\hat{\Fb}_{t}^{(1)}, \Wb_{2}$ in the Equation (\ref{Rhat1}) with $\hat{\Fb}_{t}^{(2)}, \hat{\Cb}^{(1)}$, respectively; update $\hat{\Cb}^{(2)}$ by replacing $\hat{\Fb}_{t}^{(1)}, \hat{\Rb}^{(1)}$ in the Equation (\ref{Chat1}) with $\hat{\Fb}_{t}^{(2)}, \hat{\Rb}^{(2)}$, respectively;
		\State repeat steps 3-4 until convergence and output the estimators from the last step  denoted as $\tilde{\Rb}$, $\tilde{\Cb}$ and $\{\tilde{\Fb}_{t},t \leq T\}$, respectively.
		
	\end{algorithmic}
\end{algorithm}

   \subsection{Theoretical Properties}\label{sec:3.2}
	
	In this section, we establish the convergence rates of the estimators of loading matrices at  $s$-th iteration for any $s \geq 1$. To this end, we first impose some additional  conditions that are common in the literature.
	
    \begin{asmp} \label{Assumption 4}

    The factor matrix satisfies $\EE(\Fb_{t})=\bm{0}, \EE\|\Fb_{t}\|^{4} \leq c <\infty$ for some constant $c>0$, and $$\dfrac{1}{T}\sum_{t=1}^{T}\Fb_{t}\Fb_{t}^{\top}\stackrel{p}{\rightarrow}\bm{\Sigma}_{1}\ ~~\text{and}~~ \ \dfrac{1}{T}\sum_{t=1}^{T}\Fb_{t}^{\top}\Fb_{t}\stackrel{p}{\rightarrow}\bm{\Sigma}_{2},$$ where $\bm{\Sigma}_{i}, i=1,2$ is a $k_{i}\times k_{i}$ positive definite matrix with bounded eigenvalues.
    \end{asmp}

    \begin{asmp} \label{Assumption 5}
    There exist positive constants $c,c_{0} <\infty$ such that
    (1) $\EE(e_{t,ij})=0.$
    (2) for any $t \in [T], i \in [p_{1}], j \in [p_{2}]$,
    $\sum_{s=1}^{T} \sum_{l=1}^{p_{1}} \sum_{h=1}^{p_{2}}|\EE(e_{t,ij}e_{s,lh})| \leq c \  \text{and} \ \sum_{l=1}^{p_{1}}\sum_{h=1}^{p_{2}}|\EE(e_{t,lj}e_{t,ih})| \leq c.$
    (3) for any $T\in[T], i_{1},i_{1} \in [p_{1}], j_{1},j_{2} \in [p_{2}]$,  $\sum_{s=1}^{T}\sum_{u_{1},v_{1}=1}^{p_{1}}\sum_{u_{2},v_{2}=1}^{p_{2}}|\text{Cov}(e_{t,i_{1}i_{2}}e_{t,j_{1}j_{2}},e_{s,u_{1}u_{2}}e_{s,v_{1}v_{2}})| \leq c$.
    	
    \end{asmp}

    \begin{asmp}\label{Assumption 6}
    There exist positive constants $c_{1}, c_{2}$ such that $\|\Rb\|_{\max} \leq c_{1}, \|\Cb\|_{\max} \leq c_{2}$.
    \end{asmp}

    \begin{asmp}\label{Assumption 7}
    There exists a constant $c>0$ such that
    (1) for any deterministic vectors $\bv$ and $\bw$ satisfying $\|\bv\|=1$ and $\|\bw\|=1$, $\EE\big\|\dfrac{1}{\sqrt{T}}\sum_{t=1}^{T}(\Fb_{t}\bv^{\top}\Eb_{t}\bw)\big\|^{2}\leq c;$
    (2) for any $i\in [p_{1}], j \in [p_{2}]$,  $\|\sum_{i^{\prime}=1}^{p_{1}}\sum_{j^{\prime}=1}^{p_{2}}\EE(\bxi_{i,j}\otimes\bxi_{i^{\prime},j^{\prime}})\|_{\max} \leq c$, where $\bxi_{i,j}=\text{Vec}(T^{-1/2}\sum_{t=1}^{T}\Fb_{t}e_{t,ij})$;
    (3) for any $i_{1}, i_{2} \in [p_{1}], j_{1}, j_{2} \in [p_{2}]$, $\|\sum_{i_{1}^{\prime},i_{1}^{\prime}=1}^{p_{1}}\sum_{j_{1}^{\prime},j_{2}^{\prime}=1}^{p_{2}}$ $\text{Cov}(\bxi_{i_{1},j_{1}}\otimes\bxi_{i_{2},j_{2}},\bxi_{i_{1}^{\prime},j_{1}^{\prime}}\otimes\bxi_{i_{2}^{\prime},j_{2}^{\prime}})\|_{\max} \leq c.$
    \end{asmp}

  Assumptions \ref{Assumption 4}--\ref{Assumption 7}  are standard in the literature on matrix factor models; see for example \cite{fan2021}, \cite{Yu2021Projected} and \cite{He2021Vector}. Assumption \ref{Assumption 5} ensures the (cross-sectional and time series) summability of the idiosyncratic terms $\Eb_t$, which allows for (weak) dependence in both space and time domains.
 Assumption \ref{Assumption 6} requires that the common factors are pervasive.  Assumption \ref{Assumption 7} allows the common factors $\Fb_t$ and errors $\Eb_t$ to be weakly correlated, which is satisfied, e.g., when  $\{\Fb_t\}$ and $\{\Eb_t\}$  are two mutually independent groups. One may refer to the detailed discussions on these assumptions by \cite{He2021Vector}. In the following theorem, we first establish the convergence rates of the  one-step estimator $\hat{\Rb}^{(1)}$ and  $\hat{\Cb}^{(1)}$ defined in (\ref{Rhat1}) and (\ref{Chat1}), respectively. We suppose that $m_{1}=k_{1}, m_{2}=k_{2}$ in this section and the factor numbers $k_1,k_2$ are fixed. In Section \ref{numbers}, we proposed a method to estimate the numbers of factors.

	\begin{theorem}\label{no_iter_convergence}
	Under Assumptions \ref{Assumption 1} (1), \ref{Assumption 1} (3), \ref{Assumption 2}, Assumptions \ref{Assumption 4}-\ref{Assumption 7}, further assume that $\nu_{\min}^{2}(\Hb_{2}) \gg \max({1}/{T}, {1}/{p_{2}})$, as $T, p_{1}, p_{2}$ go to infinity,  there exists an asymptotic orthogonal matrix $\hat{\Hb}_{r}^{(1)}$, such that
	\begin{equation}\label{Rhatrate}
		\dfrac{1}{p_{1}}\|\hat{\Rb}^{(1)}-\Rb\hat{\Hb}_{r}^{(1)}\|_{F}^{2}=O_{p}\left(\dfrac{1}{Tp_{2}\nu_{\min}^{2}(\Hb_{2})}+\dfrac{1}{Tp_{1}p_{2}\nu_{\min}^{2}(\Hb_{1})\nu_{\min}^{4}(\Hb_{2})}+\dfrac{1}{p_{1}^{2}p_{2}^{2}\nu_{\min}^{2}(\Hb_{1})\nu_{\min}^{4}(\Hb_{2})}\right),
	\end{equation}
	where $\Hb_{1}=\Wb_{1}^{\top}\Rb/{p_{2}}, \Hb_{2}=\Wb_{2}^{\top}\Cb/{p_{2}}$.
	
\end{theorem}

    The condition $\nu_{\min}^{2}(\Hb_{2}) \gg \max({1}/{T}, {1}/{p_{2}})$ guarantees that the matrices $\hat{\Hb}_{r}$ are asymptotic orthogonal. This condition in essence requires that the space spanned by the columns of the initial projection matrix $\Wb_2$ does not deviate far from that spanned by the columns of $\Cb$, and fails to hold if the two spaces are orthogonal. This is conceivable that the iterative algorithm would never converge to the true space if we start from its orthogonal space.
Theorem \ref{no_iter_convergence} shows that the closer the initial projection directions (space spanned by the columns of $\Wb_2$) are to the true loading directions (space spanned by the columns of $\Cb$), the faster the estimated loading matrix $\hat\Rb^{(1)}$ converges to the true loading matrix $\Rb$ up to an orthogonal matrix. In particular, if $\nu_{\min}(\Hb_{1})=\nu_{\min}(\Hb_{2})=O_{p}(1)$ (as long as we start from the consistent $\alpha$-PCA estimators in \cite{fan2021} as projection directions), then we have  $$\dfrac{1}{p_{1}}\|\hat{\Rb}^{(1)}-\Rb\hat{\Hb}_{r}^{(1)}\|_{F}^{2}=O_{p}\left(\dfrac{1}{Tp_{2}}+\dfrac{1}{p_{1}^{2}p_{2}^{2}}\right).$$

In the following theorem, we establish the convergence rate of the one-step estimator $\hat\Cb^{(1)}$ defined in (\ref{Chat1}).

   \begin{theorem}\label{convergence of an iteration}
   Under the same conditions as in Theorem \ref{no_iter_convergence},  there exists an asymptotic orthogonal matrix $\hat{\Hb}_{c}^{(1)}$, such that
    $$w_{c}^{(1)}=\dfrac{1}{p_{2}}\|\hat{\Cb}^{(1)}-\Cb\hat{\Hb}_{c}^{(1)}\|_{F}^{2}=O_{p}\left(\dfrac{1}{Tp_{1}}+\dfrac{1}{p_{1}^{2}p_{2}^{2}\nu_{\min}^{2}(\Hb_{1})\nu_{\min}^{2}(\Hb_{2})}+\dfrac{1}{T^{2}p_{2}^{2}\nu_{\min}^{2}(\Hb_{2})}+\dfrac{1}{Tp_{1}p_{2}\nu_{\min}^{2}(\Hb_{1})\nu_{\min}^{2}(\Hb_{2})}\right).$$
   In addition, $w_{r}^{(1)}=\dfrac{1}{p_{1}}\|\hat{\Rb}^{(1)}-\Rb\hat{\Hb}_{r}^{(1)}\|_{F}^{2}$ is the rate derived in (\ref{Rhatrate}).
   \end{theorem}

   Assume that $\nu_{\min}(\Hb_{1})=\nu_{\min}(\Hb_{2})=O_{p}(1)$, then the derived convergence rate of $\hat{\Cb}^{(1)}$ is
   $$\dfrac{1}{p_{2}}\|\hat{\Cb}^{(1)}-\Cb\hat{\Hb}_{c}^{(1)}\|_{F}^{2}=O_{p}\left(\dfrac{1}{Tp_{1}}+\dfrac{1}{p_{1}^{2}p_{2}^{2}}+\dfrac{1}{T^{2}p_{2}^{2}}\right),$$
   which is the same as that derived in Corollary 3.1 of \cite{Yu2021Projected}.

   As long as we get the estimators of row/column loading matrices, i.e., $\hat{\Rb}^{(1)}$ and $\hat{\Cb}^{(1)}$, the update of the factor matrix can be obtained by $\hat{\Fb}_{t}^{(2)}=\hat{\Rb}^{(1)\top}\Xb_{t}\hat{\Cb}^{(1)}/(p_{1}p_{2})$ and the corresponding common-components matrix is then updated by $\hat{\Sbb}^{(2)}_{t}=\hat{\Rb}^{(1)}\hat{\Fb}_{t}^{(2)}\hat{\Cb}^{(1)\top}=(\hat{S}_{t,ij}^{(2)})$. The following theorem provides the convergence rates of the estimated factors and common components after one iteration.
   \begin{theorem}\label{factor and common component}
   Under the same conditions as in Theorem \ref{no_iter_convergence}, as $\min\{T,p_{1},p_{2}\}\to \infty$, for any $t\in [T]$, $i\in[p_{1}]$ and $j\in[p_{2}]$, we have
   $$ \left\|\hat{\Fb}_{t}^{(2)}-(\hat{\Hb}_{r}^{(1)})^{-1}\Fb_{t}\left((\hat{\Hb}_{c}^{(1)})^{-1}\right)^{\top}\right\|_{F}=O_{p}\left(\dfrac{1}{\sqrt{T}p_{1}}+\dfrac{1}{\sqrt{T}p_{2}\nu_{\min}(\Hb_{2})}+\dfrac{1}{\sqrt{p_{1}p_{2}}}+\gamma_{f}\right),$$
   where $$\gamma_{f}=\dfrac{1}{\sqrt{Tp_{1}p_{2}}\nu_{\min}(\Hb_{1})\nu_{\min}(\Hb_{2})}+\dfrac{1}{p_{1}p_{2}\nu_{\min}(\Hb_{1})\nu_{\min}^{2}(\Hb_{2})},$$
   and for the common components, we have
   $$ |\hat{S}_{t,ij}^{(2)}-S_{t,ij}|=O_{p}\left(\dfrac{1}{\sqrt{Tp_{1}}}+\dfrac{1}{\sqrt{p_{1}p_{2}}}+\dfrac{1}{\sqrt{Tp_{2}}\nu_{\min}(\Hb_{2})}+\dfrac{1}{\sqrt{Tp_{1}p_{2}}\nu_{\min}(\Hb_{1})\nu_{\min}^{2}(\Hb_{2})}+\dfrac{1}{p_{1}p_{2}\nu_{\min}(\Hb_{1})\nu_{\min}^{2}(\Hb_{2})}\right),$$
   where $S_{t,ij}$ is the $(i,j)$-th entry of $\Sbb_{t}=\Rb\Fb_t\Cb^\top$.
   \end{theorem}

   The derived convergence rates in Theorem \ref{factor and common component} are the same as those derived in Theorem 3.5 of \cite{Yu2021Projected} when $\nu_{\min}(\Hb_{1})=\nu_{\min}(\Hb_{2})=O_{p}\left(1\right)$.  In the following theorem, we establish the recurrence formula of the convergence rate for the estimators $\hat{\Rb}^{(s+1)}$ and $\hat{\Cb}^{(s+1)}$ given any integer $s\geq 1$.

   \begin{theorem} \label{the iteration results}
   Under the same conditions stated in Theorem \ref{no_iter_convergence}, there exist asymptotic orthogonal matrices $\hat{\Hb}_{r}^{(s+1)}$ and $\hat{\Hb}_{c}^{(s+1)}$ such that for any
   integer $s\geq1$, we have
   $$
   w_{r}^{(s+1)}=\dfrac{1}{p_{1}}\|\hat{\Rb}^{(s+1)}-\Rb\hat{\Hb}_{r}^{(s+1)}\|_{F}^{2}
   =O_{p}\left(\dfrac{1}{Tp_{2}}+\dfrac{1}{p_{1}^{2}p_{2}^{2}}+\gamma_{r}^{(s+1)}\right),$$
   $$
   w_{c}^{(s+1)}=\dfrac{1}{p_{2}}\|\hat{\Cb}^{(s+1)}-\Cb\hat{\Hb}_{c}^{(s+1)}\|_{F}^{2}=O_{p}
   \left(\dfrac{1}{Tp_{1}}+\dfrac{1}{p_{1}^{2}p_{2}^{2}}+\gamma_{c}^{(s+1)}\right),$$
   where $$\gamma_{r}^{(s+1)}=\dfrac{1}{p_{2}^{2}}w_{r}^{(s)}+\dfrac{1}{T}w_{c}^{(s)}+\dfrac{w_{r}^{(s)}w_{c}^{(s)}}{p_{2}}+w_{r}^{(s)}w_{c}^{(s)2}+\dfrac{1}{p_{1}^{2}}w_{c}^{(s)2},
   $$
   $$\gamma_{c}^{(s+1)}=\dfrac{1}{Tp_{2}}w_{r}^{(s)}+\dfrac{1}{T}w_{r}^{(s+1)}+\dfrac{1}{p_{1}^{2}}w_{c}^{(s)}+\dfrac{1}{p_{1}}w_{r}^{(s)}w_{c}^{(s)}+\dfrac{1}{T}w_{r}^{(s)}w_{c}^{(s)}+w_{r}^{(s)}w_{r}^{(s+1)}w_{c}^{(s)}+\dfrac{1}{p_{2}^{2}}w_{r}^{(s)}w_{r}^{(s+1)}.$$
   \end{theorem}

In the following theorem, we also establish the recurrence formula of the convergence rate for the estimators $\hat{\Fb}^{(s+1)}$ and $\hat{\Sbb}^{(s+1)}_{t}$ given any integer $s\geq 1$.

   \begin{theorem}\label{the iteration results2}
   	Under the same conditions stated in Theorem \ref{no_iter_convergence}, as $\min\{T,p_{1},p_{2}\}\to \infty$, for any $t\in [T]$, $i\in[p_{1}]$ and $j\in[p_{2}]$ and any integer $s>1$, we have
    $$\left\|\hat{\Fb}_{t}^{(s+1)}-(\hat{\Hb}_{r}^{(s)})^{-1}\Fb_{t}\left((\hat{\Hb}_{c}^{(s)})^{-1}\right)^{\top}\right\|_{F}=O_{p}\left(\sqrt{\dfrac{{w_{r}^{(s)}}}{{p_{2}}}}+\sqrt{\dfrac{{w_{r}^{(s-1)}}}{{Tp_{2}}}}+\sqrt{\dfrac{{w_{c}^{(s)}}}{{p_{1}}}}+\sqrt{\dfrac{{w_{c}^{(s-1)}}}{{Tp_{1}}}}+\dfrac{1}{\sqrt{p_{1}p_{2}}}+\gamma_{f}^{(s+1)}\right),$$
   where $$\gamma_{f}^{(s+1)}=\sqrt{w_{r}^{(s)}w_{c}^{(s)}}+\sqrt{\dfrac{{w_{r}^{(s-1)}w_{c}^{(s-1)}}}{{T}}}+\dfrac{\sqrt{w_{r}^{(s-1)}}w_{c}^{(s-1)}}{p_{1}},$$
and for the common components, we have

   $$|\hat{S}_{t,ij}^{(s+1)}-S_{t,ij}|=O_{p}\left(\sqrt{w_{r}^{(s)}}+\sqrt{\dfrac{{w_{r}^{(s-1)}}}{{Tp_{2}}}}+\sqrt{w_{c}^{(s)}}+\sqrt{\dfrac{{w_{c}^{(s-1)}}}{{Tp_{1}}}}+\dfrac{1}{\sqrt{p_{1}p_{2}}}+\gamma^{(s+1)}\right),$$
   where $\gamma^{(s+1)}=\sqrt{\dfrac{{w_{r}^{(s-1)}w_{c}^{(s-1)}}}{{T}}}+\dfrac{\sqrt{w_{r}^{(s-1)}}w_{c}^{(s-1)}}{p_{1}}$ and  $S_{t,ij}$ is the $(i,j)$-th entry of $\Sbb_{t}=\Rb\Fb_t\Cb^\top$.
   \end{theorem}

Theoretical analysis for the estimators of loading matrices above relies on the correct specification of factor numbers (note that we suppose $m_1=k_1,m_2=k_2$ in Theorem \ref{no_iter_convergence}) and the  strong factor conditions $\Rb^{\top}\Rb/{p_{1}}=\Ib_{k_{1}},\Cb^{\top}\Cb/{p_{2}}=\Ib_{k_{2}}$, which  means that the row and column factors are pervasive along both dimensions and is an
extension of the pervasive assumption in \cite{stockwatson02JASA} to the matrix regime.

 \subsection{Determining the Number of Factors}\label{numbers}
The dimensions $k_1$ and $k_2$ of the common factor matrix $\Fb_t$ in (\ref{MFM}) are unknown in practice and needs to be determined. In this study, we specify the numbers of row and column factors by borrowing the ideas of eigenvalue-ratio  discussed in \cite{lam2012factor} and \cite{ahn2013eigenvalue}. However, our method is based on the (eigenvalue-ratios of)  row/column sample covariance matrices of the estimated factor matrices  with  a pseudo large pair of factor numbers, in contrast to the row/column sample covariance matrices of the original/projected observations in the literature such as \cite{wang2019factor,fan2021,Yu2021Projected}, which is of independent interest. In detail, we first set $m_1=m_2=k_{\max}$ ($k_{\max}>\max\{k_1,k_2\}$) and thus by
Algorithm \ref{alg:1}, we can obtain $\tilde{\Fb}_{t}$, which is of dimension $k_{\max}\times k_{\max}$. Then the dimensions $k_1$ and $k_2$ are further  determined as follows:

\begin{equation}\label{equ:fn}
	\hat{k}_{1}=\arg\max_{j \leq k_{\max}}\dfrac{\lambda_{j}\left(\dfrac{1}{T}\sum_{t=1}^{T}\tilde{\Fb}_{t}\tilde{\Fb}_{t}^{\top}\right)}{\lambda_{j+1}\left(\dfrac{1}{T}\sum_{t=1}^{T}\tilde{\Fb}_{t}\tilde{\Fb}_{t}^{\top}\right)}, \ \
	\hat{k}_{2}=\arg\max_{j \leq k_{\max}}\dfrac{\lambda_{j}\left(\dfrac{1}{T}\sum_{t=1}^{T}\tilde{\Fb}_{t}^{\top}\tilde{\Fb}_{t}\right)}{\lambda_{j+1}\left(\dfrac{1}{T}\sum_{t=1}^{T}\tilde{\Fb}_{t}^{\top}\tilde{\Fb}_{t}\right)}.
\end{equation}

We have the following theorem which guarantee the consistency of the estimators in equation (\ref{equ:fn}).

\begin{theorem}\label{ER}
	Under the same assumptions as in Theorem \ref{no_iter_convergence}, if $k_{\max} > \max\{k_1,k_2\}$, when $T,p_1,p_2 \to \infty$, then we have $\PP(\hat{k}_1=k_{1})\to 1$ and $\PP(\hat{k}_2=k_2)\to 1$.
\end{theorem}

The factor numbers are assumed to be known in advance in the theoretical analysis of Section \ref{sec:3.2}, while in Section \ref{sec:factormatrix}, the inferences on the factor spaces  are  robust against overestimating the number of factors. Theorem \ref{ER} fills this gap as we can always determine the factor numbers correctly with probability tending to 1.

	\section{Simulation Study}
\subsection{Data Generation}\label{datagenerate}
In this section, we investigate the empirical performance of the Random Projection-based Iterative Least Squares (RPILS) procedure in terms of estimating the loading and factor spaces. We first introduce the data generation mechanism of the synthetic dataset, which is similar to \cite{he2022matrix}.
	We set $k_{1}=3, k_{2}=3$, draw the entries of $\Rb$ and $\Cb$ independently from uniform distribution $U(-1,1)$, and let
	$$\Fb_{t}=\phi\Fb_{t-1}+\sqrt{1-\phi^{2}}\bepsilon_{t},  \ \ \Eb_{t}=\psi\Eb_{t-1}+\sqrt{1-\psi^{2}}\Ub_{t},$$
	where $\text{Vec}(\bepsilon_{t}) \stackrel{i.i.d}{\sim} \cN(\bm{0}, \Ib_{k_{1} \times k_{2}}), \Ub_{t} \stackrel{i.i.d}{\sim} \cM\cN(\bm{0},\Ub_{E},\Vb_{E}),$  i.e., $\text{Vec}(\Ub_{t}) \stackrel{i.i.d}{\sim} \cN(\bm{0},\Vb_{E}\otimes\Ub_{E}).$ The parameters $\phi$ and  $\psi$ control the temporal correlations, and $\Ub_{E}$ and $\Vb_{E}$ are matrices with ones on the diagonal, and the off-diagonal entries are $1/p_{1}$ and $1/p_{2}$, respectively.

\subsection{Estimation error for loading spaces}
We compare  the performances of our Random Projection based Iterative Least Squares (RPILS) method with the $\alpha$-PCA method ($\alpha=0$) by \cite{fan2021} and the PE method by \cite{Yu2021Projected} in terms of estimating the loading and factor spaces. For the RPILS method, the initial weight matrices  $\Wb_{1},\Wb_{2}$ are Walsh-Hadamard matrices \citep{fan2022learning} and the column dimensions $m_1,m_2$ are set as the true dimensions of factor matrices, i.e, $m_1=k_1=3,m_2=k_2=3$.
To  show the impact of the initial weight matrices, we also compare with the One-Step Estimators (OSE), either with Walsh-Hadamard matrices as the initial weight matrices or with $\alpha$-PCA estimators as initial weight matrices, denoted as $\text{OSE}_1$ and  $\text{OSE}_2$ respectively. We point out that in unreported simulations, we have tried to use the matrices with all their entries from a standard normal distribution as the initial weight matrices for the RPILS method, and find that the performances are almost the same as using initial weight matrices with Walsh-Hadamard matrices. We consider the following two scenarios of parameter settings:

\vspace{0.5em}

    \textbf{Scenario A}:   $p_{1}=20, T=p_{2} \in \{20,50,100,150,200\}, \phi=0.1, \psi=0.1.$

\vspace{0.5em}

    \textbf{Scenario B}:  $p_{2}=20, T=p_{1} \in \{20,50,100,150,200\}, \phi=0.1, \psi=0.1.$

\vspace{0.5em}
	\begin{table}[htbp]
		\caption{Averaged estimation errors (standard errors in parentheses) in terms of $\cD(\hat{\Rb},\Rb)$, $\cD(\hat{\Cb},\Cb)$ and $\cD(\text{Vec}(\Fb_{t}),\text{Vec}(\hat{\Fb}_{t}))$ for Scenarios A and B under Matrix Normal distribution over 500 replications.} \label{table:1}
		\renewcommand{\arraystretch}{1.6}
		\centering
		\scalebox{0.85}{
			\begin{tabular}{ccccccccc}
				\toprule[2pt]
				Evaluation & $T$ & $p_{1}$ & $p_{2}$ &$\text{OSE}_1$ &$\text{OSE}_2$ &RPILS & $\alpha$-PCA &PE\\
				\cmidrule(r){1-9}	
				\multicolumn{9}{c}{Setting A: $p_{1}=20, p_{2}=T$}\\			
				\multirow{5}{*}{$\cD(\hat{\Rb},\Rb)$}&20&20&20&0.6178(0.1205) &0.0954(0.0173) &0.0938(0.0158) &0.1151(0.0308) &0.0947(0.0162) \\
				&50&20&50&0.5992(0.1197) &0.0357(0.0053) &0.0355(0.0052) &0.0588(0.0219) &0.0355(0.0052) \\
				&100&20&100&0.5797(0.1330) &0.0177(0.0025) &0.0176(0.0025) &0.0459(0.0221) &0.0176(0.0025)  \\
				&150&20&150&0.5835(0.1315) &0.0117(0.0017) &0.0117(0.0016) &0.0426(0.0199) &0.0117(0.0016)  \\
				&200&20&200&0.5762(0.1305) &0.0088(0.0013) &0.0088(0.0012) &0.0415(0.0203) &0.0088(0.0012)  \\
				\cmidrule(r){2-9}
				
				\multirow{5}{*}{$\cD(\hat{\Cb},\Cb)$}&20&20&20&0.6138(0.1182) &0.0947(0.0166)  &0.0933(0.0156) &0.1148(0.0290) &0.0942(0.0160) \\
				&50&20&50&0.5458(0.1446) &0.0572(0.0063) &0.0568(0.0062) &0.0594(0.0068) &0.0572(0.0063)  \\
				&100&20&100&0.4785(0.1508) &0.0402(0.0034) &0.0399(0.0033) &0.0407(0.0035) &0.0402(0.0034)  \\
				&150&20&150&0.4400(0.1643) &0.0326(0.0026) &0.0324(0.0025) &0.0328(0.0025) &0.0326(0.0026)  \\
				&200&20&200&0.4217(0.1619) &0.0282(0.0021) &0.0280(0.0021) &0.0282(0.0021) &0.0282(0.0021)  \\
				\cmidrule(r){2-9}
				
				\multirow{5}{*}{$\cD(\text{Vec}(\Fb_{t}),\text{Vec}(\hat{\Fb}_{t}))$}&20&20&20&0.6740(0.0380) &0.1837(0.0371) &0.1783(0.0338) &0.1837(0.0371) &0.1784(0.0339) \\
				&50&20&50&0.7928(0.0440) &0.1074(0.0122) &0.1059(0.0111) &0.1074(0.0122) &0.1059(0.0111) \\
				&100&20&100&0.8213(0.0494) &0.0747(0.0079) &0.0737(0.0070) &0.0747(0.0079) &0.0738(0.0070) \\
				&150&20&150&0.8342(0.0503) &0.0606(0.0059) &0.0599(0.0053) &0.0606(0.0059) &0.0599(0.0053)  \\
				&200&20&200&0.8394(0.0499) &0.0521(0.0049) &0.0515(0.0044) &0.0521(0.0049) &0.0515(0.0044)  \\
				\cmidrule(r){1-9}	
				\multicolumn{9}{c}{Setting B: $p_{2}=20, p_{1}=T$}\\
				\multirow{5}{*}{$\cD(\hat{\Rb},\Rb)$}&20&20&20&0.6178(0.1205) &0.0954(0.0173) &0.0938(0.0158) &0.1151(0.0308) &0.0947(0.0162) \\
				&50&50&20&0.5514(0.1440) &0.0574(0.0058) &0.0569(0.0056) &0.0595(0.0061) &0.0574(0.0058)  \\
				&100&100&20&0.4945(0.1526) &0.0402(0.0037) &0.0399(0.0036) &0.0406(0.0038) &0.0402(0.0037) \\
				&150&150&20&0.4481(0.1573) &0.0327(0.0027) &0.0325(0.0027) &0.0328(0.0027) &0.0327(0.0027) \\
				&200&200&20&0.4193(0.1671) &0.0281(0.0021) &0.0279(0.0021) &0.0281(0.0021) &0.0281(0.0021) \\
				\cmidrule(r){2-9}
				
				\multirow{5}{*}{$\cD(\hat{\Cb},\Cb)$}&20&20&20&0.6138(0.1182) &0.0947(0.0166) &0.0933(0.0156) &0.1148(0.0290) &0.0942(0.0160)  \\
				&50&50&20&0.5903(0.1269) &0.0359(0.0054) &0.0357(0.0053) &0.0599(0.0194) &0.0357(0.0053)  \\
				&100&100&20&0.5797(0.1291) &0.0177(0.0025) &0.0177(0.0025) &0.0469(0.0206) &0.0177(0.0025) \\
				&150&150&20&0.5728(0.1292) &0.0116(0.0017) &0.0116(0.0016) &0.0424(0.0185) &0.0116(0.0016)  \\
				&200&200&20&0.5761(0.1287) &0.0088(0.0012) &0.0088(0.0012) &0.0425(0.0204) &0.0088(0.0012)  \\
				\cmidrule(r){2-9}
				
				\multirow{5}{*}{$\cD(\text{Vec}(\Fb_{t}),\text{Vec}(\hat{\Fb}_{t}))$}&20&20&20&0.6740(0.0380) &0.1837(0.0371) &0.1783(0.0338) &0.1837(0.0371) &0.1784(0.0339) \\
				&50&50&20&0.7925(0.0454) &0.1080(0.0123) &0.1064(0.0113) &0.1080(0.0123) &0.1064(0.0113)  \\
				&100&100&20&0.8253(0.0480) &0.0751(0.0073) &0.0742(0.0066) &0.0751(0.0073) &0.0742(0.0066) \\
				&150&150&20&0.8349(0.0466) &0.0604(0.0061) &0.0597(0.0057) &0.0604(0.0061) &0.0597(0.0057)  \\
				&200&200&20&0.8382(0.0504) &0.0524(0.0051) &0.0518(0.0047) &0.0524(0.0051) &0.0518(0.0047)  \\
							
				\bottomrule[2pt]
			\end{tabular}
		}
	\end{table}

To measure the performances of various methods in terms of estimating loading/factor spaces, we adopt a metric between linear spaces which was also utilized in \cite{Yu2021Projected,He2021Matrix}. For two column-wise orthogonal matrices $(\bQ_1)_{p\times q_1}$ and $(\bQ_2)_{p\times q_2}$, we define
\[
\mathcal{D}(\bQ_1,\bQ_2)=\bigg(1-\frac{1}{\max{(q_1,q_2)}}\mbox{Tr}\Big(\bQ_1\bQ_1^{\top}\bQ_2\bQ_2^{\top}\Big)\bigg)^{1/2}.
\]
By the definition of $\mathcal{D}(\bQ_1,\bQ_2)$, we can easily see that $0\leq \mathcal{D}(\bQ_1,\bQ_2)\leq 1$, which measures the distance between the column spaces spanned by  $\bQ_1$ and $\bQ_2$, i.e., $\text{span}(\bQ_1)$ and $\text{span}(\bQ_2)$. In particular, $\text{span}(\bQ_1)$ and $\text{span}(\bQ_2)$  are the same when $\mathcal{D}(\bQ_1,\bQ_2)=0$, while  $\text{span}(\bQ_1)$ and $\text{span}(\bQ_2)$   are orthogonal when $\mathcal{D}(\bQ_1,\bQ_2)=1$.  The Gram-Schmidt orthogonalization can be used to make $\bQ_1$ and $\bQ_2$  column-orthogonal matrices.
	
Table \ref{table:1} reported the averaged estimation errors (standard errors in parentheses) in terms of $\cD(\hat{\Rb},\Rb)$, $\cD(\hat{\Cb},\Cb)$ and $\cD(\text{Vec}(\Fb_{t}),\text{Vec}(\hat{\Fb}_{t}))$ for Scenarios A and B over 500 replications. All methods benefit from large dimensions in terms of estimating loading spaces. By comparing the results for $\text{OSE}_1$ and $\text{OSE}_2$, we conclude that the better the initial projection directions, the faster the loading/factor spaces converge to the corresponding true ones. As for RPILS, the results indicate that even if we start from a random guess of the projection directions, we can finally get  satisfactory estimators via the iterative procedure in Algorithm \ref{alg:1}. In addition, the RPILS performs comparably with the PE method and better than the $\alpha$-PCA method, which is consistent with our theoretical analysis.  In other words, the RPILS method also reduces the magnitudes of the
idiosyncratic error components, thereby increasing the signal-to-noise ratio and enjoying the same advantage as PE.
However, compared with the eigen-decomposition-based PE method, the  RPILS method is computationally simpler.

\subsection{Estimating the numbers of factors}
In this section, we compare the empirical performances of the proposed RPILS-ER method with the $\alpha$-PCA based ER method ($\alpha$-PCA-ER) by \cite{fan2021}, the IterER method by \cite{Yu2021Projected}, the  ER method based on iTOPUP and iTIPUP by \cite{han2022rank} (denote as iTOP-ER and iTIP-ER) and the TCorTh method by \cite{lam2021rank} in terms of estimating the pair of factor numbers. Table \ref{table:2}  presents the frequencies of exact estimation and underestimation over 500 replications under Scenario A and Scenario B by different methods. We adopt the same data generating scheme as introduced in Section \ref{datagenerate} except that we set $k_1=3,k_2=2$. And we set $k_{\max}$ = 8 for all the methods. From Table \ref{table:2}, we see that the RPILS-ER and IterER perform comparably and both perform better than the others. As $p_{1}$ and $p_2$ increase, it can be seen that all the methods' performances get better except the TCorTh method. And it can also be seen that as the dimension of $p_1$ and $p_2$ increase, the proportion of exact estimation by RPILS-ER has the tendency to converge to 1, which is consistent with our theoretical analysis.

\begin{table}[htbp]
	\caption{The frequencies of exact estimation and underestimation of the numbers of factors under Settings A and B over 500 replications.} \label{table:2}
	\renewcommand{\arraystretch}{1.4}
	\centering
	\scalebox{1}{
		\begin{tabular}{ccccccc}
			\toprule[2pt]
			$T$ &RPILS-ER &IterER &$\alpha$-PCA-ER &iTOP-ER &iTIP-ER & TCorTh  \\
			\cmidrule(r){1-7}
			\multicolumn{7}{c}{Scenario A: $p_{1}=20, p_{2}=T, k_{1}=3, k_{2}=2, k_{\max}=8$}\\
			20 &0.6880(0.3120) &0.9800(0.0200) &0.4660(0.5340) &0.5920(0.4080) &0.1160(0.8840) &0.5100(0.4900) \\
			50 &0.9840(0.0160) &1.0000(0.0000) &0.7040(0.2960) &0.8600(0.1400) &0.1880(0.8120) &0.9560(0.0440) \\
			100 &0.9940(0.0060) &1.0000(0.0000) &0.6900(0.3100) &0.9240(0.0760) &0.4140(0.5860) &0.9900(0.0100) \\
			150 &1.0000(0.0000) &1.0000(0.0000) &0.7040(0.2960) &0.9680(0.0320) &0.5340(0.4660) &0.9880(0.0120) \\
			200 &1.0000(0.0000) &1.0000(0.0000) &0.6960(0.3040) &0.9800(0.0200) &0.7200(0.2800) &0.9400(0.0600) \\
			\cmidrule(r){1-7}
			\multicolumn{7}{c}{Scenario B: $p_{2}=20, p_{1}=T, k_{1}=3, k_{2}=2, k_{\max}=8$}\\
			20 &0.6880(0.3120) &0.9800(0.0200) &0.4660(0.5340) &0.5920(0.4080) &0.1160(0.8840) &0.5100(0.4900) \\
			50 &0.9980(0.0020) &1.0000(0.0000) &0.9440(0.0560) &0.8980(0.1020) &0.2360(0.7640) &1.0000(0.0000) \\
			100 &1.0000(0.0000) &1.0000(0.0000) &0.9600(0.0400) &0.9720(0.0280) &0.3560(0.6440) &0.9760(0.0240) \\
			150 &1.0000(0.0000) &1.0000(0.0000) &0.9640(0.0360) &0.9860(0.0140) &0.4760(0.5240) &0.8600(0.1400) \\
			200 &1.0000(0.0000) &1.0000(0.0000) &0.9400(0.0600) &0.9940(0.0060) &0.6100(0.3900) &0.7200(0.2800) \\
			\bottomrule[2pt]
		\end{tabular}
	}
\end{table}

    \section{Real Data Example}

  \subsection{Fama-French $10\times10 $ portfolios}
  In this section, we study a financial portfolio dataset studied in   \cite{wang2019factor} and \cite{Yu2021Projected}. The dataset is composed of monthly returns of 100 portfolios, well structured into a $10 \times 10$ matrix at each time point, with rows corresponding to 10 levels of market capital size (denoted as S1-S10) and columns corresponding to 10 levels of book-to-equity ratio (denoted as BE1-BE10). The dataset collects monthly returns from January 1964 to December 2019 covering a total of 672 months. The details are available at the website \url{ http:
//mba.tuck.dartmouth.edu/pages/faculty/ken.french/data_library.html}.
  Following the same preprocessing as in \cite{Yu2021Projected} and \cite{wang2019factor}, we adjusted the return series by first subtracting the corresponding monthly market
excess returns and then standardizing each of the series. We imputed the missing values by the factor-model-based
method introduced in \cite{xiong2019large}. All augmented Dickey-Fuller tests reject the null hypothesis, which indicates the stationarity of  all
the series.

  As for the factor numbers, the RPILS-ER, $\alpha$-PCA-ER, iTOP-IC and iTIP-IC all suggest $(k_1, k_2) = (1, 1)$, while the IterER suggests that $(k_1, k_2) = (2, 1)$, iTOP-ER, iTIP-ER and TCorTh all suggest $(k_1, k_2) = (2, 2)$. For better illustration, we take $(k_{1},k_{2})=(2,2)$ as in \cite{Yu2021Projected}. The estimated loading matrices after varimax rotation and scaling are reported in Table \ref{table:3}.

  Table \ref{table:3} shows the estimated row and column loading matrices after varimax rotation and scaling. From the table, we can see the  proposed RPILS  method performs similarly  to PE, $\alpha$-PCA and ACCE methods in terms of the estimated loading matrices. From the perspective of Size, small Size portfolios load heavily on the first factor while large Size portfolios load on the second. From the perspective of Book-to-Equity, small BE portfolios load heavily on the second factor while large BE portfolios load mainly on the first factor.

  \begin{table}[!h]
  	\caption{Loading matrices for Fama-French data set after varimax rotation and scaling by 30.} \label{table:3}
  	\renewcommand{\arraystretch}{1}
  	\centering
  	\scalebox{1}{
  		\begin{tabular}{cccccccccccc}
  			\toprule[2pt]
  			\multicolumn{12}{c}{Size}\\
  			\cmidrule(r){1-12}
  			Method &Factor &S1 &S2 &S3 &S4 &S5 &S6 &S7 &S8 &S9 &S10  \\
  			\cmidrule(r){1-12}
  			\multirow{2}{*}{RPILS} &1 &\cellcolor {Lavender}16 &\cellcolor {Lavender}15 &\cellcolor {Lavender}12 &\cellcolor {Lavender}10 &\cellcolor {Lavender}8 &5 &3 &1 &-4 &\cellcolor {Lavender}-7 \\
  			&2 &5 &1 &-3 &-5 &\cellcolor {Lavender}-8 &\cellcolor {Lavender}-10 &\cellcolor {Lavender}-12 &\cellcolor {Lavender}-13 &\cellcolor {Lavender}-15 &\cellcolor {Lavender}-11 \\
  			\cmidrule(r){1-12}
  			\multirow{2}{*}{PE} &1 &\cellcolor {Lavender}-16 &\cellcolor {Lavender}-15 &\cellcolor {Lavender}-12 &\cellcolor {Lavender}-10 &\cellcolor {Lavender}-8 &-5 &-3 &-1 &4 &\cellcolor {Lavender}7  \\
  			&2 &\cellcolor {Lavender}-6 &-1 &3 &5 &\cellcolor {Lavender}8 &\cellcolor {Lavender}10 &\cellcolor {Lavender}12 &\cellcolor {Lavender}13 &\cellcolor {Lavender}15 &\cellcolor {Lavender}11 \\
  			\cmidrule(r){1-12}
  			\multirow{2}{*}{$\alpha$-PCA} &1 &\cellcolor {Lavender}-14 &\cellcolor {Lavender}-14 &\cellcolor {Lavender}-13 &\cellcolor {Lavender}-11 &\cellcolor {Lavender}-9 &\cellcolor {Lavender}-7 &-4 &-2 &3 &\cellcolor {Lavender}7 \\
  			&2 &-4 &-2 &1 &3 &6 &\cellcolor {Lavender}9 &\cellcolor {Lavender}12 &\cellcolor {Lavender}13 &\cellcolor {Lavender}16 &\cellcolor {Lavender}14  \\
  			\cmidrule(r){1-12}
  			\multirow{2}{*}{ACCE} &1 &\cellcolor {Lavender}-12 &\cellcolor {Lavender}-14 &\cellcolor {Lavender}-12 &\cellcolor {Lavender}-13 &\cellcolor {Lavender}-10 &\cellcolor {Lavender}-6 &-3 &-1 &4 &\cellcolor {Lavender}9 \\
  			&2 &-1 &-1 &-1 &2 &5 &\cellcolor {Lavender}10 &\cellcolor {Lavender}11 &\cellcolor {Lavender}18 &\cellcolor {Lavender}15 &\cellcolor {Lavender}11 \\
  			\cmidrule(r){1-12}
  			\multicolumn{12}{c}{Book-to-Equity}\\
  			\cmidrule(r){1-12}
  			Method &Factor &BE1 &BE2 &BE3 &BE4 &BE5 &BE6 &BE7 &BE8 &BE9 &BE10  \\
  			\cmidrule(r){1-12}
  			\multirow{2}{*}{RPILS} &1 &\cellcolor {Lavender}-6 &-1 &4 &\cellcolor {Lavender}7 &\cellcolor {Lavender}10 &\cellcolor {Lavender}11 &\cellcolor {Lavender}12 &\cellcolor {Lavender}12 &\cellcolor {Lavender}12 &\cellcolor {Lavender}10 \\
  			&2 &\cellcolor {Lavender}20 &\cellcolor {Lavender}17 &\cellcolor {Lavender}11 &\cellcolor {Lavender}8 &4 &2 &0 &-1 &-1 &0 \\
  			\cmidrule(r){1-12}
  			\multirow{2}{*}{PE} &1 &\cellcolor {Lavender}6 &1 &-4 &\cellcolor {Lavender}-7 &\cellcolor {Lavender}-10 &\cellcolor {Lavender}-11 &\cellcolor {Lavender}-12 &\cellcolor {Lavender}-12 &\cellcolor {Lavender}-12 &\cellcolor {Lavender}-10 \\
  			&2 &\cellcolor {Lavender}20 &\cellcolor {Lavender}17 &\cellcolor {Lavender}11 &\cellcolor {Lavender}8 &4 &2 &0 &-1 &-1 &0 \\
  			\cmidrule(r){1-12}
  			\multirow{2}{*}{$\alpha$-PCA} &1 &\cellcolor {Lavender}6 &2 &-4 &\cellcolor {Lavender}-7 &\cellcolor {Lavender}-10 &\cellcolor {Lavender}-11 &\cellcolor {Lavender}-12 &\cellcolor {Lavender}-13 &\cellcolor {Lavender}-12 &\cellcolor {Lavender}-11 \\
  			&2 &\cellcolor {Lavender}19 &\cellcolor {Lavender}18 &\cellcolor {Lavender}12 &\cellcolor {Lavender}8 &4 &2 &0 &-1 &-1 &-1 \\
  			\cmidrule(r){1-12}
  			\multirow{2}{*}{ACCE} &1 &\cellcolor {Lavender}6 &-1 &-4 &\cellcolor {Lavender}-8 &\cellcolor {Lavender}-8 &\cellcolor {Lavender}-9 &\cellcolor {Lavender}-10 &\cellcolor {Lavender}-13 &\cellcolor {Lavender}-15 &\cellcolor {Lavender}-12 \\
  			&2 &\cellcolor {Lavender}21 &\cellcolor {Lavender}15 &\cellcolor {Lavender}11 &\cellcolor {Lavender}6 &5 &2 &1 &-2 &-3 &1\\
  			
  			\bottomrule[2pt]
  		\end{tabular}
  	}
  \end{table}

  \begin{table}[!h]
  	\caption{Rolling validation for the Fama-French portfolios.  The sample size of the training set is $12n$ and $k_1 = k_2 = k$. $\bar{\text{MSE}},\bar{\rho},\bar {v} $ are the
mean pricing error, mean unexplained proportion of total variances and mean variation of the estimated loading space.} \label{table:4}
  	\renewcommand{\arraystretch}{1.4}
  	\centering
  	\scalebox{0.85}{
  		\begin{tabular}{cccccccccccccc}
  			\toprule[2pt]

  			\multicolumn{6}{c}{$\overline{\text{MSE}}$} &\multicolumn{4}{c}{$\bar{\rho}$} &\multicolumn{4}{c}{$\bar{v}$} \\
  			\cmidrule(r){3-6} \cmidrule(r){7-10} \cmidrule(r){11-14}
  			$n$ &$k$ &RPILS &PE &$\alpha$-PCA &ACCE &RPILS &PE &$\alpha$-PCA &ACCE &RPILS &PE &$\alpha$-PCA &ACCE \\
  \hline
  			5 &1 &0.8766 &0.8703 &$\mathbf{0.8624}$ &0.8846 &0.8001 &0.8022 &$\mathbf{0.7960}$ &0.8284 &0.2653 &$\mathbf{0.1757}$ &0.2412 &0.3032 \\
  			10 &1 &0.8735 &$\mathbf{0.8548}$ &0.8596 &0.8797 &0.7948 &$\mathbf{0.7836}$ &0.7913 &0.8149 &0.2341 &$\mathbf{0.0847}$ &0.2027 &0.1654 \\
  			15 &1 &0.8570 &$\mathbf{0.8530}$ &0.8599 &0.8844 &0.7828 &$\mathbf{0.7822}$ &0.7918 &0.8118 &$\mathbf{0.0606}$ &0.0636 &0.2335 &0.1520 \\
  			5 &2 &$\mathbf{0.5954}$ &0.5965 &0.6010 &0.6673 &$\mathbf{0.6231}$ &0.6248 &0.6284 &0.6727 &$\mathbf{0.2331}$ &0.2390 &0.3497 &0.4605  \\
  			10 &2 &0.6014 &$\mathbf{0.6013}$ &0.6108 &0.6545 &$\mathbf{0.6273}$ &0.6276 &0.6364 &0.6684 &$\mathbf{0.0892}$ &0.0924 &0.2606 &0.2568 \\
  			15 &2 &0.6027 &$\mathbf{0.6025}$ &0.6115 &0.6375 &$\mathbf{0.6261}$ &0.6262 &0.6302 &0.6516 &$\mathbf{0.0564}$ &0.0573 &0.1735 &0.1894  \\
  			5 &3 &$\mathbf{0.5204}$ &0.5216 &0.5291 &0.5639 &$\mathbf{0.5473}$ &0.5495 &0.5558 &0.5900 &$\mathbf{0.2781}$ &0.2865 &0.4321 &0.4974  \\
  			10 &3 &$\mathbf{0.5181}$ &0.5193 &0.5262 &0.5728 &$\mathbf{0.5465}$ &0.5481 &0.5549 &0.5936 &$\mathbf{0.1074}$ &0.1142 &0.3532 &0.3029  \\
  			15 &3 &$\mathbf{0.5166}$ &0.5172 &0.5220 &0.5601 &$\mathbf{0.5438}$ &0.5446 &0.5444 &0.5825 &$\mathbf{0.0783}$ &0.0839 &0.3082 &0.2986  \\
  		
  			\bottomrule[2pt]
  		\end{tabular}
  	}
  \end{table}

  We also use a rolling-validation scheme as in \cite{Yu2021Projected} and \cite{wang2019factor} to further compare the methods. For each year $t$ from 1996 to 2019, we repeatedly use $n$ (bandwidth) years before $t$ to fit the matrix-variate factor model and
  estimate the loading matrices. The estimated loadings are then used to estimate the factors and corresponding residuals of the
12 months in the current year. In detail,
  let $\Yb_{t}^{i}$ and $\hat{\Yb}_{t}^{i}$ be the observed and estimated price matrix of month $i$ in year $t$, denote $\bar{\Yb}_{t}$ as the mean price matrix, and define $$\text{MSE}_{t}=\dfrac{1}{12\times10\times10}\sum_{i=1}^{12}\|\hat{\Yb}_{t}^{i}-\Yb_{t}^{i}\|_{F}^{2}, \ \ \rho_{t}=\dfrac{\sum_{i=1}^{12}\|\hat{\Yb}_{t}^{i}-\Yb_{t}^{i}\|_{F}^{2}}{\sum_{i=1}^{12}\|\Yb_{t}^{i}-\bar{\Yb}_{t}\|_{F}^{2}},$$ as the mean squared pricing error and unexplained proportion of total variances, respectively. The variation of loading space is measured by $$v_{t}=\cD(\hat{\Cb}_{t}\otimes\hat{\Rb}_{t},\hat{\Cb}_{t-1}\otimes\hat{\Rb}_{t-1}), $$ during the rolling-validation procedure.

  Table \ref{table:4} reports the results of the means of $\text{MSE}, \rho, v$ by the proposed RPILS method and the competitors. For different bandwidth $n$ and the number of factors $k_{1}, k_{2}$, our RPILS method is comparable to the PE method and better than the other methods in terms of the averaged $\text{MSE}, \rho, v$.

  \subsection{Multinational macroeconomic indices}
  In this section, we analyze a multinational macroeconomic index dataset collected from Organization for Economic Co-operation and Development (OECD), which contains 10 macroeconomic indices across 8 countries over 130 quarters from 1988-Q1 to 2020-Q2. The 8 countries are the United States, the United Kingdom, Canada, France, Germany, Norway, Australia and New Zealand. The indices are from 4 major groups, namely consumer price, interest rate, production, and international trade. For the preprocessing procedure of the dataset, we refer to \cite{Yu2021Projected} for details.

  As for the factor numbers, the RPILS-ER and $\alpha$-PCA suggest $(k_1,k_2)=(1,2)$ , while IterER suggests $(k_1,k_2)=(1,5)$, we take the advice  $(k_{1},k_{2})=(3,4)$ by \cite{Yu2021Projected} for better illustration. The estimated loading matrices are shown in Table \ref{table:5} and Table \ref{table:6}. The  proposed RPILS method behaves almost the same as the PE method. As concluded in \cite{Yu2021Projected}, the countries excluding Germany naturally divide into 3 groups, Oceania, North American and European. On the other hand, the macroeconomic indices divide into 4 groups, consumer price, interest rate, production and international trade, which coincide with economic interpretations.

  \begin{table}[!h]
  	\caption{Row loading matrices by different methods for multinational macroeconomic index dataset, varimax rotated and multiplied by 10.} \label{table:5}
  	\renewcommand{\arraystretch}{1}
  	\centering
  	\scalebox{1}{
  		\begin{tabular}{cccccccccc}
  			\toprule[2pt]
  			Method &Factor &AUS &NZL &USA &CAN &NOR &DEU &FRA &GBR  \\
  			\cmidrule(r){1-10}
  			\multirow{3}{*}{RPILS} &1 &5 &-4 &1 &-4 &\cellcolor {Lavender}-20 &-7 &\cellcolor {Lavender}-15 &\cellcolor {Lavender}-12 \\
  			&2 &0 &1 &\cellcolor {Lavender}-20 &\cellcolor {Lavender}-16 &\cellcolor {Lavender}11 &-8 &-5 &-4 \\
  			&3 &\cellcolor {Lavender}24 &\cellcolor {Lavender}16 &1 &-3 &0 &6 &-1 &2 \\
  			\cmidrule(r){1-10}
  			\multirow{3}{*}{PE} &1 &0 &1 &\cellcolor {Lavender}-7 &\cellcolor {Lavender}-6 &3 &-3 &-2 &-1 \\
  			&2 &2 &-2 &1 &-1 &\cellcolor {Lavender}-7 &-2 &\cellcolor {Lavender}-5 &\cellcolor {Lavender}-5 \\
  			&3 &\cellcolor {Lavender}8 &\cellcolor {Lavender}6 &0 &-1 &0 &2 &-1 &1 \\
  			\cmidrule(r){1-10}
  			\multirow{3}{*}{$\alpha$-PCA} &1 &-1 &1 &\cellcolor {Lavender}-7 &\cellcolor {Lavender}-5 &3 &-3 &-2 &-1 \\
  			&2 &1 &-1 &0 &-1 &\cellcolor {Lavender}-7 &-2 &\cellcolor {Lavender}-5 &\cellcolor {Lavender}-4 \\
  			&3 &\cellcolor {Lavender}-7 &\cellcolor {Lavender}-7 &0 &1 &0 &-1 &1 &-1 \\
  			\cmidrule(r){1-10}
  			\multirow{3}{*}{ACCE} &1 &2 &-2 &1 &-2 &\cellcolor {Lavender}-6 &0 &\cellcolor {Lavender}-6 &\cellcolor {Lavender}-5 \\
  			&2 &\cellcolor {Lavender}7 &\cellcolor {Lavender}5 &0 &0 &0 &\cellcolor {Lavender}5 &0 &0 \\
  			&3 &0 &-2 &\cellcolor {Lavender}8 &\cellcolor {Lavender}4 &-2 &1 &1 &2 \\
  			
  			\bottomrule[2pt]
  		\end{tabular}
  	}
  \end{table}

  \begin{table}[!h]
  	\caption{Column loading matrices by different methods for multinational macroeconomic index dataset, varimax rotated and multiplied by 10.} \label{table:6}
  	\renewcommand{\arraystretch}{1}
  	\centering
  	\scalebox{0.85}{
  		\begin{tabular}{cccccccccccc}
  			\toprule[2pt]
  			Method &Factor &CPI:Tot &CPI:Enter &CPI:NFNE &IR:3-Mon &IR:Long &P:TIEC &P:TM &GDP &IT:Ex &IT:Im \\
  			\cmidrule(r){1-12}
  			\multirow{4}{*}{RPILS} &1 &-2 &4 &-7 &-3 &3 &\cellcolor {Lavender}-20 &\cellcolor {Lavender}-20 &-5 &2 &0 \\
  			&2 &1 &1 &-2 &\cellcolor {Lavender}-18 &\cellcolor {Lavender}-24 &1 &-1 &2 &1 &-3 \\
  			&3 &-2 &6 &-7 &1 &0 &1 &1 &\cellcolor {Lavender}16 &\cellcolor {Lavender}19 &\cellcolor {Lavender}14 \\
  			&4 &\cellcolor {Lavender}-19 &\cellcolor {Lavender}-20 &\cellcolor {Lavender}-10 &4 &-4 &1 &-1 &-1 &-1 &1 \\
  			\cmidrule(r){1-12}
  			\multirow{3}{*}{PE} &1 &1 &-2 &3 &1 &-1 &\cellcolor {Lavender}6 &\cellcolor {Lavender}7 &2 &-1 &0 \\
  			&2 &\cellcolor {Lavender}6 &\cellcolor {Lavender}7 &3 &-1 &1 &0 &0 &0 &0 &0 \\
  			&3 &0 &0 &-1 &\cellcolor {Lavender}-6 &\cellcolor {Lavender}-8 &0 &0 &1 &0 &-1 \\
  			&4 &1 &-2 &3 &0 &0 &-1 &0 &\cellcolor {Lavender}-5 &\cellcolor {Lavender}-6 &\cellcolor {Lavender}-5 \\
  			\cmidrule(r){1-12}
  			\multirow{3}{*}{$\alpha$-PCA} &1 &0 &-1 &1 &1 &-1 &\cellcolor {Lavender}7 &\cellcolor {Lavender}6 &\cellcolor {Lavender}4 &0 &0 \\
  			&2 &\cellcolor {Lavender}7 &\cellcolor {Lavender}5 &\cellcolor {Lavender}5 &-1 &1 &0 &1 &0 &0 &0 \\
  			&3 &0 &0 &0 &\cellcolor {Lavender}-7 &\cellcolor {Lavender}-7 &1 &0 &-1 &1 &0 \\
  			&4 &0 &2 &-2 &0 &0 &0 &0 &2 &\cellcolor {Lavender}7 &\cellcolor {Lavender}6 \\
  			\cmidrule(r){1-12}
  			\multirow{4}{*}{ACCE} &1 &0 &0 &0 &0 &1 &\cellcolor {Lavender}-7 &\cellcolor {Lavender}-7 &0 &0 &0 \\
  			&2 &1 &0 &0 &\cellcolor {Lavender}-5 &\cellcolor {Lavender}-4 &1 &-1 &-2 &\cellcolor {Lavender}-4 &\cellcolor {Lavender}-6 \\
  			&3 &\cellcolor {Lavender}-4 &2 &\cellcolor {Lavender}-9 &0 &0 &2 &-1 &2 &0 &0 \\
  			&4 &\cellcolor {Lavender}6 &\cellcolor {Lavender}7 &0 &-1 &3 &0 &0 &2 &1 &-1 \\
  			
  			\bottomrule[2pt]
  		\end{tabular}
  	}
  \end{table}

  We also adopt a rolling prediction procedure to further investigate  the practical use of different methods. First, we consider  the change of inflation (second-order difference of the log level of the total consumer price index--CPI:Tot) of a selected country at time $t$, denoted as $y_t$. Let $\bx_t$ be the vector of all the other 9 indices of the selected country at time $t$, and $\Zb_t$ be the $8\times 10$ panel at time $t$, with rows corresponding to the countries and column corresponding to all macroeconomic indices. We predict $y_{t+1}$ by the following Auto-Regression (AR) model (Model 1) and Factor-Augmented-Auto-Regression (FAAR) models (Models 2--4), similar to the Diffusion Index forecasting by \cite{stock2002macroeconomic}.

 \begin{description}
	\item[Model 1]  $y_{t+1}=a+by_t+\epsilon_{t+1}$,
	\item[Model 2]  $y_{t+1}=a+by_t+\bbeta^\top\bbf_{1t}+\epsilon_{t+1}$, where $\bbf_{1t}$'s are estimated  from the vector factor model with observations $\{\bx_t\}$.
	\item[Model 3] $y_{t+1}=a+by_t+\bbeta^\top\bbf_{2t}+\epsilon_{t+1}$, where $\bbf_{2t}$'s are estimated from the vector factor model with observations $\{\text{Vec}(\Zb_t)\}$.
	\item[Model 4]  $y_{t+1}=a+by_t+\bbeta^\top\text{Vec}(\Fb_t)+\epsilon_{t+1}$, where $\Fb_t$'s are estimated from the matrix factor model with observations $\{\Zb_t\}$, by the RPILS, PE, ACCE, and $\alpha$-PCA, respectively.
\end{description}

The models for comparison here are exactly the same with \cite{Yu2021Projected}, we explain these models here again for completeness.
First, Model 1 is a simple auto-regression model. Model 2 adds  common index factors of the selected country into the auto-regression model in Model 1. In Model 3 and Model 4, both index and country factors are taken into account. The difference is that Model 4 considers the more parsimonious matrix factor structure while Model 3 vectorizes the matrix time series and considers the vector factor structure. To avoid possible over-fitting in prediction, we also use the LASSO  \citep{tibshirani1996regression} to select factors and estimate the coefficients for Models 2-4.

  \begin{table}[!h]
  	\caption{MAPEs for inflation and the growth rate of GDP (both at an annual rate) for different countries with different methods, $(k_{1},k_{2})=(3,4)$.} \label{table:8}
  	\renewcommand{\arraystretch}{1.6}
  	\centering
  	\scalebox{1}{
  		\begin{tabular}{ccccccccc}
  			\toprule[2pt]
  			Model &AUS &NZL &USA &CAN &NOR &DEU &FRA &GBR  \\
  			\cmidrule(r){1-9}
  			\multicolumn{9}{c}{MAPEs for inflation rates}\\
  			Model 1 &1.5880 &$\mathbf{1.8866}$ &2.8346 &2.4605 &$\mathbf{1.9359}$ &1.8751 &1.8060 &1.5776 \\
  			Model 2 &1.6258 &1.9086 &2.5835 &2.2349 &2.2607 &1.9574 &1.8322 &1.5019 \\
  			Model 3 &1.5989 &2.0507 &2.4709 &1.9653 &2.3865 &$\mathbf{1.7279}$ &$\mathbf{1.2411}$ &$\mathbf{1.1980}$ \\
  			Model 4 (RPILS) &$\mathbf{1.5607}$ &1.9319 &2.6829 &2.2736 &2.1377 &1.8756 &1.7236 &1.3863 \\
  			Model 4 (PE) &1.5886 &2.0038 &2.3527 &1.8515 &2.1550 &$\mathbf{1.7330}$ &1.4074 &1.3241 \\
  			Model 4 (ACCE) &1.5853 &1.9292 &$\mathbf{2.2466}$ &$\mathbf{1.8273}$ &2.4442 &1.8040 &1.3822 &1.3348 \\
  			Model 4 ($\alpha$-PCA) &1.5981 &1.9177 &2.3777 &1.9023 &2.3132 &1.7402 &1.4335 &1.3490 \\
  				\cmidrule(r){1-9}
  			\multicolumn{9}{c}{MAPEs for the growth rate of GDP}\\
  	Model 1 &1.8897 &2.9211 &$\mathbf{2.4317}$ &2.8279 &3.7794 &$\mathbf{3.2970}$ &$\mathbf{2.6777}$ &$\mathbf{3.2515}$ \\
  			Model 2 &$\mathbf{1.8709}$ &2.9117 &2.4683 &2.8584 &3.8977 &3.3413 &2.8395 &3.3231 \\
  			Model 3 &$\mathbf{1.8669}$ &$\mathbf{2.8391}$ &2.5664 &2.6364 &3.7375 &3.9167 &2.8535 &3.3285 \\
  			Model 4 (RPILS) &1.8904 &3.0531 &2.5087 &2.7084 &3.7467 &$\mathbf{3.3315}$ &2.7672 &3.3198 \\
  			Model 4 (PE) &1.8731 &2.9268 &2.5819 &$\mathbf{2.5505}$ &3.9012 &3.4485 &2.8079 &3.3523 \\
  			Model 4 (ACCE) &1.9118 &3.0370 &2.5067 &2.6869 &$\mathbf{3.6739}$ &3.5838 &2.7446 &3.4251 \\
  			Model 4 ($\alpha$-PCA) &$\mathbf{1.8709}$ &3.0364 &2.6078 &2.6544 &3.8955 &3.5226 &2.8701 &3.3843  \\
  			\bottomrule[2pt]
  		\end{tabular}
  	}
  \end{table}

  For each quarter $t$ from 2008-Q1 to 2020-Q2, we use the 80 neighboring observations before $t$ to train the models and predict $y_{t+1}$ (denoted as $\hat{y}_{t+1}$). As $y_{t}$ was standardized in preprocessing, we transformed the predicted $y_{t+1}$ to match the change of inflation rate by multiplying the standard deviation and adding back the sample mean. For simplicity of notation, we still denote the transformed predictor as $\hat{y}_{t+1}$. The inflation $I_{t+1}$ is then predicted by integrating $\hat y_{t+1}$ and $I_t$, i.e., $\hat{I}_{t+1}=\hat{y}_{t+1}+I_t$.
In Model 2 and Model 3, the factor numbers before model selection are set as $k_2$ and $k_1\times k_2$, respectively. We also focus on the case that $(k_1, k_2)=(3, 4)$.
The top panel of
  Table \ref{table:8} shows the mean absolute prediction errors (MAPEs) for the annualized inflation rates. For the largest Oceania country, Australia, Model 4 with RPILS has the best prediction performance in terms of MAPE. For Norway and New Zealand,  Model 1 performs the best, indicating that the index and country factors act as noises in Model 2-4. For the USA, Canada, France, Great Britain and Germany, both index and country factors are useful for improving prediction performance. Note that the results also show that the matrix factor structure can further improve the prediction for two American countries.
  We also consider the rolling prediction of the GDP growth rate (first-order difference of the log level of GDP) for all countries.  The results shown in the bottom panel of Table  \ref{table:8} demonstrate that for the strong manufacturing American and European countries, USA, Germany, France and Great Britain, the simple AR model suffices to predict the GDP growth rates well.  For the other countries, the country and the index factors  contribute to improving the prediction performance of the GDP growth rates, while for Canada and Norway, the advantage of the matrix factor structure is more obvious.

	\section{Discussion}
	
We propose a simple iterative least squares algorithm for the matrix factor model. In the first step, we   estimate the latent factor matrices by projecting the observations with two deterministic weight matrices. We show that the inferences on  factors  are
still asymptotically valid under some regularity conditions, even if both row and column factor numbers are overestimated.
In the second step, we  estimate the row/column loading matrices by minimizing the squared Frobenius loss function under some identifiability conditions. The resultant estimators of the loading matrices are further treated as the new weight/projection matrices and we iteratively perform the above two steps until convergence. Given the true dimensions of the factor matrices, we establish the convergence rates of the estimators for loading matrices and common components at the $s$-th iteration for any $s\geq 1$. To determine the pair of factor numbers, we proposed an eigenvalue-ratio method based on the iterated results, and the resultant estimators are proven to be consistent.

As a future direction, our methodology could be generalized
 to tensor-valued time series.
Intuitively, if one substitutes the squared loss in (\ref{M2.1}) with the Huber loss, it would lead to a more robust estimator, which is of independent interest because real-world financial returns and macroeconomic indexes often exhibit heavy tails. Since a significant amount of additional work is still needed, we leave this to  future work.

	\bibliographystyle{Ref.bib}
	\bibliography{Ref}

\setlength{\bibsep}{1pt}

\renewcommand{\baselinestretch}{1}
\setcounter{footnote}{0}
\clearpage
\setcounter{page}{1}
\setcounter{section}{0}
\renewcommand{\thesection}{S\arabic{section}}
\renewcommand{\thetable}{S\arabic{table}}

\title{
	\begin{center}
		\Large Supplementary Materials for ``An Efficient Iterative Least Squares Algorithm for Large-dimensional Matrix Factor Model via Random Projection"
	\end{center}
}
\date{}
\begin{center}
	\author{
	Yong He
		\footnotemark[1],
Ran Zhao\footnotemark[1],
Wen-Xin Zhou\footnotemark[2],
	}
\renewcommand{\thefootnote}{\fnsymbol{footnote}}
\footnotetext[1]{Institute of Financial Studies, Shandong University, China. E-mail:{\tt heyong@sdu.edu.cn, Zhaoran@mail.sdu.edu.cn }}
\footnotetext[2]{Department of Mathematical Sciences, University of California, San Diego, USA. E-mail:{\tt wez243@ucsd.edu}}

\end{center}
\maketitle

This document provides the  detailed proofs of the main theorems and additional lemmas and propositions.

\vspace{1em}

\section{Proofs of the main theorems}

\subsection{Proof of Theorem \ref{Convergence}}
	
	\begin{proof}
	By the fact that $\Hb_{1}^{\top}(\Hb_{1}\Hb_{1}^{\top})^{+}\Hb_{1}=\Ib_{k_{1}}$, $\Hb_{2}^{\top}(\Hb_{2}\Hb_{2}^{\top})^{+}\Hb_{2}=\Ib_{k_{2}}$, then $\hat{\Fb}_{t}=\Hb_{1}\Fb_{t}\Hb_{2}^{\top}+\bm{\cE}_{t}$ implies $\Mb_{1}^{\top}\hat{\Fb}_{t}\Mb_{2}-\Fb_{t}=\Mb_{1}^{\top}\bm{\cE}_{t}\Mb_{2}$ with $\Mb_{1}=(\Hb_{1}\Hb_{1}^{\top})^{+}\Hb_{1}, \Mb_{2}=(\Hb_{2}\Hb_{2}^{\top})^{+}\Hb_{2}$. As $\|(\Hb_{1}\Hb_{1}^{\top})^{+}\Hb_{1}\|_{2}=O_{p}(\nu_{\min}^{-1}(\Hb_{1}))$ and $\|(\Hb_{2}\Hb_{2}^{\top})^{+}\Hb_{2}\|_{2}=O_{p}(\nu_{\min}^{-1}(\Hb_{2}))$, by Lemma \ref{lemma1} (1), we have
	$$\|\Mb_{1}^{\top}\hat{\Fb}_{t}\Mb_{2}-\Fb_{t}\|_{2}=O_{p}\left(\dfrac{1}{\sqrt{p_{1}}}\nu_{\min}^{-1}(\Hb_{1})\nu_{\min}^{-1}(\Hb_{2})\right).$$

	\end{proof}
	
	\subsection{Proof of Theorem \ref{Vec}}
	\begin{proof}
	By Proposition \ref{Pro}, $\lambda_{\min}(\dfrac{1}{T}\Mb^{\top}\hat{\Fb}^{\top}\hat{\Fb}\Mb) \geq \lambda_{\min}(\dfrac{1}{T}\hat{\Fb}^{\top}\hat{\Fb})\lambda_{\min}(\Mb^{\top}\Mb) \geq c(p_{1}p_{2})^{-1}\lambda_{\min}(\bD_{\Hb}^{-2})$ with large probability. By the SVD of $\Hb^{\top}$, i.e., $\Hb^{\top}=\Ub_{\Hb}(\Db_{\Hb},0)\Eb_{\Hb}^{\top}$,
we conclude that $\Pb_{\hat{\Fb}\Mb}$ is well defined.
	
	As $\hat{\Fb}=\Fb\Hb^{\top}+\bm{\cE}$, then we have $\hat{\Fb}\Mb-\Fb=\bm{\cE}(\Hb\Hb^{\top})^{+}\Hb$ with $\Mb=(\Hb\Hb^{\top})^{+}\Hb$. Further by the fact that $\|(\Hb\Hb^{\top})^{+}\Hb\|_{2}=O_{p}\left(\nu_{\min}^{-1}\right)$ and Lemma \ref{lemma1} (2), (3), we have
	$$\dfrac{1}{\sqrt{T}}\|\hat{\Fb}\Mb-\Fb\|_{2}=\dfrac{1}{\sqrt{T}}\|\bm{\cE}(\Hb\Hb^{\top})^{+}\Hb\|_{2}\leq\dfrac{1}{\sqrt{T}}\|\bm{\cE}\|_{2}\|(\Hb\Hb^{\top})^{+}\Hb\|_{2}=O_{p}\left(\dfrac{1}{\sqrt{p_{1}p_{2}}}\nu_{\min}^{-1}\right),$$
	$$\dfrac{1}{T}\|\Fb^{\top}(\hat{\Fb}\Mb-\Fb)\|_{2}\leq\dfrac{1}{T}\|\Fb^{\top}\bm{\cE}\|_{2}\|(\Hb\Hb^{\top})^{+}\Hb\|_{2}=O_{p}\left(\dfrac{1}{\sqrt{Tp_{1}p_{2}}}\nu_{\min}^{-1}\right).$$
Further by $$
	\begin{aligned}
	\left\|\dfrac{1}{T}\Mb^{\top}\hat{\Fb}^{\top}\hat{\Fb}\Mb-\dfrac{1}{T}\Fb^{\top}\Fb\right\|_{2}=&\left\|\dfrac{1}{T}(\hat{\Fb}\Mb-\Fb)^{\top}(\hat{\Fb}\Mb-\Fb)+\dfrac{1}{T}\Fb^{\top}(\hat{\Fb}\Mb-\Fb)+\dfrac{1}{T}(\hat{\Fb}\Mb-\Fb)^{\top}\Fb\right\|_{2}\\
	&\leq \dfrac{1}{T}\|\hat{\Fb}\Mb-\Fb\|_{2}^{2}+\dfrac{2}{T}\|\Fb^{\top}(\hat{\Fb}\Mb-\Fb)\|_{2}=O_{p}\left(\dfrac{1}{\sqrt{Tp_{1}p_{2}}}\nu_{\min}^{-1}+\dfrac{1}{p_{1}p_{2}}\nu_{\min}^{-2}\right),
	\end{aligned}$$
we have
	$$\left\|(\dfrac{1}{T}\Mb^{\top}\hat{\Fb}^{\top}\hat{\Fb}\Mb)^{-1}\right\|_{2}=O_{p}\left(1\right)$$ and due to $\Ab^{-1}-\Bb^{-1}=\Ab^{-1}(\Bb-\Ab)\Bb^{-1}$, we get
	$$\left\|\left(\dfrac{1}{T}\Mb^{\top}\hat{\Fb}^{\top}\hat{\Fb}\Mb\right)^{-1}-\left(\dfrac{1}{T}\Fb^{\top}\Fb\right)^{-1}\right\|_{2}=O_{p}\left(\dfrac{1}{\sqrt{Tp_{1}p_{2}}}\nu_{\min}^{-1}+\dfrac{1}{p_{1}p_{2}}\nu_{\min}^{-2}\right).$$

	As a result,
	$$\begin{aligned}
	\Pb_{\hat{\Fb}\Mb}-\Pb_{\Fb}&=\dfrac{1}{\sqrt{T}}\hat{\Fb}\Mb\left[(\dfrac{1}{T}\Mb^{\top}\hat{\Fb}^{\top}\hat{\Fb}\Mb)^{-1}-(\dfrac{1}{T}\Fb^{\top}\Fb)^{-1}\right](\dfrac{1}{\sqrt{T}}\hat{\Fb}\Mb)^{\top}\\
	&+\dfrac{1}{\sqrt{T}}\hat{\Fb}\Mb(\dfrac{1}{T}\Fb^{\top}\Fb)^{-1}\dfrac{1}{\sqrt{T}}(\hat{\Fb}\Mb-\Fb)^{\top}+\dfrac{1}{\sqrt{T}}(\hat{\Fb}\Mb-\Fb)(\dfrac{1}{T}\Fb^{\top}\Fb)^{-1}\dfrac{1}{\sqrt{T}}\Fb^{\top},
	\end{aligned}$$
	$$\|\Pb_{\hat{\Fb}\Mb}-\Pb_{\Fb}\|_{2}=O_{p}\left(\dfrac{1}{\sqrt{p_{1}p_{2}}}\nu_{\min}^{-1}\right).$$	
	Finally, noting that $\Pb_{\hat{\Fb}}\Pb_{\hat{\Fb}\Mb}=\Pb_{\hat{\Fb}\Mb}$, we have
	$$\|\Pb_{\hat{\Fb}}\Pb_{\Fb}-\Pb_{\Fb}\|_{2}\leq\|\Pb_{\hat{\Fb}}(\Pb_{\Fb}-\Pb_{\hat{\Fb}\Mb})\|_{2}+\|\Pb_{\hat{\Fb}\Mb}-\Pb_{\Fb}\|_{2}=O_{p}\left(\dfrac{1}{\sqrt{p_{1}p_{2}}}\nu_{\min}^{-1}\right).$$
		
	\end{proof}
	
\subsection{Proof of Theorem \ref{no_iter_convergence}}

\begin{proof}
	Without loss of generality, we assume in the following that $m_{1}=k_{1}=1, m_{2}=k_{2}=1.$
	$$
	\begin{aligned}
	\dfrac{1}{T}\sum_{t=1}^{T}\Xb_{t}\Wb_{2}\hat{\Fb}_{t}^{\top}&=\dfrac{1}{Tp_{1}p_{2}}\sum_{t=1}^{T}\Xb_{t}\Wb_{2}\Wb_{2}^{\top}\Xb_{t}^{\top}\Wb_{1}\\
	&=\dfrac{1}{Tp_{1}p_{2}}\sum_{t=1}^{T}(\Rb\Fb_{t}\Cb^{\top}+\Eb_{t})\Wb_{2}\Wb_{2}^{\top}(\Rb\Fb_{t}\Cb^{\top}+\Eb_{t})^{\top}\Wb_{1}\\
	&=\dfrac{1}{Tp_{1}p_{2}}\sum_{t=1}^{T}\Rb\Fb_{t}\Cb^{\top}\Wb_{2}\Wb_{2}^{\top}\Cb\Fb_{t}^{\top}\Rb^{\top}\Wb_{1}+\dfrac{1}{Tp_{1}p_{2}}\sum_{t=1}^{T}\Rb\Fb_{t}\Cb^{\top}\Wb_{2}\Wb_{2}^{\top}\Eb_{t}^{\top}\Wb_{1}\\
	&+\dfrac{1}{Tp_{1}p_{2}}\sum_{t=1}^{T}\Eb_{t}\Wb_{2}\Wb_{2}^{\top}\Cb\Fb_{t}^{\top}\Rb^{\top}\Wb_{1}+\dfrac{1}{Tp_{1}p_{2}}\sum_{t=1}^{T}\Eb_{t}\Wb_{2}\Wb_{2}^{\top}\Eb_{t}^{\top}\Wb_{1}\\
	&:=\bdelta_{1}+\bdelta_{2}+\bdelta_{3}+\bdelta_{4}.
	\end{aligned}
	$$
In the following, we analyze $\bdelta_{1},\bdelta_{2},\bdelta_{3},\bdelta_{4}$ term by term.

For the first term $\bdelta_{1}$, on the one hand, we have
	$$\begin{aligned}
	\|\bdelta_{1}\|_{F}&=\left\|\dfrac{1}{Tp_{1}p_{2}}\sum_{t=1}^{T}\Rb\Fb_{t}\Cb^{\top}\Wb_{2}\Wb_{2}^{\top}\Cb\Fb_{t}^{\top}\Rb^{\top}\Wb_{1}\right\|_{F}=
	\left\|\dfrac{p_{2}}{T}\sum_{t=1}^{T}\Rb\Fb_{t}\Hb_{2}^{\top}\Hb_{2}\Fb_{t}^{\top}\Hb_{1}\right\|_{F}\\
	&\leq p_{2}\|\Rb\|_{F}\left\|\dfrac{1}{T}\sum_{t=1}^{T}\Fb_{t}\Fb_{t}^{\top}\right\|_{F}\|\|\Hb_{2}\|_{2}^{2}\|\Hb_{1}\|_{2}=O_{p}\left(\sqrt{p_{1}}p_{2}\nu_{\max}(\Hb_{1})\nu_{\max}^{2}(\Hb_{2})\right)\\
	&=O_{p}\left(\sqrt{p_{1}}p_{2}\nu_{\min}(\Hb_{1})\nu_{\min}^{2}(\Hb_{2})\right).
	\end{aligned}$$
On the other hand, $$\begin{aligned}
	\|\bdelta_{1}\|_{F} &\geq\ \|\bdelta_{1}\|_{2}=\left\|\dfrac{p_{2}}{T}\sum_{t=1}^{T}\Rb\Fb_{t}\Hb_{2}^{\top}\Hb_{2}\Fb_{t}^{\top}\Hb_{1}\right\|_{2}=\sqrt{p_{1}}p_{2}\left\|\dfrac{1}{T}\sum_{t=1}^{T}\Fb_{t}\Hb_{2}^{\top}\Hb_{2}\Fb_{t}^{\top}\Hb_{1}\right\|_{2}\\
	&\geq \sqrt{p_{1}}p_{2}\left\|\dfrac{1}{T}\sum_{t=1}^{T}\Fb_{t}\Fb_{t}^{\top}\right\|_{2}\nu_{\min}(\Hb_{1})\nu_{\min}^{2}(\Hb_{2}) \gtrsim \sqrt{p_{1}}p_{2}\nu_{\min}(\Hb_{1})\nu_{\min}^{2}(\Hb_{2}).
\end{aligned}$$
Thus, we have $$\|\bdelta_{1}\|_{F} \asymp O_{p}\left(\sqrt{p_{1}}p_{2}\nu_{\min}(\Hb_{1})\nu_{\min}^{2}(\Hb_{2})\right).$$

	For the second term $\bdelta_{2}$, we have
	$$\begin{aligned}
	\|\bdelta_{2}\|_{F}&=\left\|\dfrac{1}{Tp_{1}p_{2}}\sum_{t=1}^{T}\Rb\Fb_{t}\Cb^{\top}\Wb_{2}\Wb_{2}^{\top}\Eb_{t}^{\top}\Wb_{1}\right\|_{F} \leq \left\|\dfrac{1}{Tp_{1}p_{2}}\sum_{t=1}^{T}\Rb\Fb_{t}\Wb_{2}^{\top}\Eb_{t}^{\top}\Wb_{1}\right\|_{F}\|\Cb^{\top}\Wb_{2}\|_{2}\\
	&\leq \dfrac{1}{Tp_{1}}\|\Rb\|_{F}\|\sum_{t=1}^{T}\Fb_{t}\Wb_{2}^{\top}\Eb_{t}^{\top}\Wb_{1}\|_{F}\|\Hb_{2}\|_{2}=O_{p}\left(\sqrt{\dfrac{p_{2}}{T}}\nu_{\min}(\Hb_{2})\right),\\		
	\end{aligned}$$
	where the last equation is derived according to Lemma \ref{basic-lemma} (3).
	
	By  Lemma \ref{basic-lemma} (3), we have $\EE\|\sum_{t=1}^{T}\Fb_{t}\Wb_{2}^{\top}\Eb_{t}^{\top}\|_{F}^{2}=O(Tp_{1}p_{2})$, thus,
	$$\begin{aligned}
	\|\bdelta_{3}\|_{F}&=\left\|\dfrac{1}{Tp_{1}p_{2}}\sum_{t=1}^{T}\Eb_{t}\Wb_{2}\Wb_{2}^{\top}\Cb\Fb_{t}^{\top}\Rb^{\top}\Wb_{1}\right\|_{F}\leq \dfrac{1}{Tp_{1}p_{2}}\|\sum_{t=1}^{T}\Eb_{t}\Wb_{2}\Fb_{t}^{\top}\|_{F}\|\Wb_{2}^{\top}\Cb\|_{2}\|\Rb^{\top}\Wb_{1}\|_{2}\\
	&=O_{p}\left(\sqrt{\dfrac{p_{1}p_{2}}{T}}\nu_{\min}(\Hb_{1})\nu_{\min}(\Hb_{2})\right).
	\end{aligned}$$
	
	By Lemma \ref{basic-lemma} (5),
	$$\|\bdelta_{4}\|_{F}=\left\|\dfrac{1}{Tp_{1}p_{2}}\sum_{t=1}^{T}\Eb_{t}\Wb_{2}\Wb_{2}^{\top}\Eb_{t}^{\top}\Wb_{1}\right\|_{F}=O_{p}\left(\sqrt{\dfrac{p_{2}}{T}}+\dfrac{1}{\sqrt{p_{1}}}\right).$$
	
	Let $\Zb=\dfrac{1}{T^{2}}\left(\sum_{t=1}^{T}\hat{\Fb}_{t}\Wb_{2}^{\top}\Xb_{t}^{\top}\right)\left(\sum_{t=1}^{T}\Xb_{t}\Wb_{2}\hat{\Fb}_{t}^{\top}\right),$ then
	$$
	\begin{aligned}
	\Zb&=\dfrac{1}{p_{1}^{2}p_{2}^{2}T^{2}}(\sum_{t=1}^{T}\Wb_{1}^{\top}\Xb_{t}\Wb_{2}\Wb_{2}^{\top}\Xb_{t}^{\top})(\sum_{t=1}^{T}\Xb_{t}\Wb_{2}\Wb_{2}^{\top}\Xb_{t}^{\top}\Wb_{1})\\
	&=(\bdelta_{1}+\bdelta_{2}+\bdelta_{3}+\bdelta_{4})^{\top}(\bdelta_{1}+\bdelta_{2}+\bdelta_{3}+\bdelta_{4})=\sum_{i=1}^{4}\sum_{j=1}^{4}\bdelta_{i}^{\top}\bdelta_{j} .
	\end{aligned}$$
	We can prove that when $\nu_{\min}^{2}(\Hb_{2}) \gg \max(\dfrac{1}{T},\dfrac{1}{p_{2}})$, $\|\Zb\|_{F}\geq \|\Zb\|_{2} \gtrsim \|\bdelta_{1}\|_{2}^{2} \gtrsim p_{1}p_{2}^{2}\nu_{\min}^{2}(\Hb_{1})\nu_{\min}^{4}(\Hb_{2})$, thus  $$\|\Zb^{-1/2}\|_{F}=O_{p}\left(\dfrac{1}{\sqrt{p_{1}}p_{2}\nu_{\min}(\Hb_{1})\nu_{\min}^{2}(\Hb_{2})}\right).$$
	
	$$\begin{aligned}
	\hat{\Rb}^{(1)}&=\sqrt{p_{1}}\left(\sum_{t=1}^{T}\Xb_{t}\Wb_{2}\hat{\Fb}_{t}^{\top}\right)\left[\left(\sum_{t=1}^{T}\hat{\Fb}_{t}\Wb_{2}^{\top}\Xb_{t}^{\top}\right)\left(\sum_{t=1}^{T}\Xb_{t}\Wb_{2}\hat{\Fb}_{t}^{\top}\right)\right]^{-1/2}\\
	&=\sqrt{p_{1}}\left(\dfrac{1}{T}\sum_{t=1}^{T}\Xb_{t}\Wb_{2}\hat{\Fb}_{t}^{\top}\right)\Zb^{-1/2}\\
	&=\Rb\left[\dfrac{1}{\sqrt{p_{1}}p_{2}}\left(\dfrac{1}{T}\sum_{t=1}^{T}\Fb_{t}\Cb^{\top}\Wb_{2}\Wb_{2}^{\top}\Cb\Fb_{t}^{\top}\Rb^{\top}\Wb_{1}\right)\Zb^{-1/2}\right]+\dfrac{1}{\sqrt{p_{1}}p_{2}}\left(\dfrac{1}{T}\sum_{t=1}^{T}\Rb\Fb_{t}\Cb^{\top}\Wb_{2}\Wb_{2}^{\top}\Eb_{t}^{\top}\Wb_{1}\right)\Zb^{-1/2}\\
	&+\dfrac{1}{\sqrt{p_{1}}p_{2}}\left(\dfrac{1}{T}\sum_{t=1}^{T}\Eb_{t}\Wb_{2}\Wb_{2}^{\top}\Cb\Fb_{t}^{\top}\Rb^{\top}\Wb_{1}\right)\Zb^{-1/2}+\dfrac{1}{\sqrt{p_{1}}p_{2}}\left(\dfrac{1}{T}\sum_{t=1}^{T}\Eb_{t}\Wb_{2}\Wb_{2}^{\top}\Eb_{t}^{\top}\Wb_{1}\right)\Zb^{-1/2}\\
	&:=\cI+\cI\cI+\cI\cI\cI+\cI\cV.
	\end{aligned}
	$$
	
	Let $\hat{\Hb}_{r}^{(1)}=\dfrac{1}{\sqrt{p_{1}}p_{2}T}\left(\sum_{t=1}^{T}\Fb_{t}\Cb^{\top}\Wb_{2}\Wb_{2}^{\top}\Cb\Fb_{t}^{\top}\Rb^{\top}\Wb_{1}\right)\Zb^{-1/2}$
	then,
	\begin{equation}\label{S1.1}
		\hat{\Rb}^{(1)}-\Rb\hat{\Hb}_{r}^{(1)}=\cI\cI+\cI\cI\cI+\cI\cV.
	\end{equation}
	As $$
	\begin{aligned}
	&\left\|\dfrac{1}{T}\sum_{t=1}^{T}\Fb_{t}\Cb^{\top}\Wb_{2}\Wb_{2}^{\top}\Cb\Fb_{t}^{\top}\Rb^{\top}\Wb_{1}\right\|_{F} \leq \left\|\dfrac{1}{T}\sum_{t=1}^{T}\Fb_{t}\Fb_{t}^{\top}\right\|_{F}\|\Cb^{\top}\Wb_{2}\|_{2}\|\Wb_{2}^{\top}\Cb\|_{2}\|\Rb^{\top}\Wb_{1}\|_{2}\\
	&= p_{1}p_{2}^{2}\left\|\dfrac{1}{T}\sum_{t=1}^{T}\Fb_{t}^{\top}\Fb_{t}\right\|_{2}\|\Hb_{2}\|_{2}^{2}\|\Hb_{1}\|_{2}
	=O_{p}\left(p_{1}p_{2}^{2}\nu_{\min}(\Hb_{1})\nu_{\min}^{2}(\Hb_{2})\right),
	\end{aligned}$$
	thus we have the following results:
	$$
	\begin{aligned}
	\|\hat{\Hb}_{r}^{(1)}\|_{F}&=\left\|\dfrac{1}{\sqrt{p_{1}}p_{2}}(\dfrac{1}{T}\sum_{t=1}^{T}\Fb_{t}\Cb^{\top}\Wb_{2}\Wb_{2}^{\top}\Cb\Fb_{t}^{\top}\Rb^{\top}\Wb_{1})\Zb^{-1/2}\right\|_{F}\\
	&\leq \dfrac{1}{\sqrt{p_{1}}p_{2}}\left\|\dfrac{1}{T}\sum_{t=1}^{T}\Fb_{t}\Cb^{\top}\Wb_{2}\Wb_{2}^{\top}\Cb\Fb_{t}^{\top}\Rb^{\top}\Wb_{1}\right\|_{F}\|\Zb^{-1/2}\|_{F}=O_{p}\left(1\right);
	\end{aligned}$$
	
	$$
	\|\cI\cI\|_{F}=\left\|\dfrac{1}{\sqrt{p_{1}}p_{2}}\left(\dfrac{1}{T}\sum_{t=1}^{T}\Rb\Fb_{t}\Cb^{\top}\Wb_{2}\Wb_{2}^{\top}\Eb_{t}^{\top}\Wb_{1}\right)\Zb^{-1/2}\right\|_{F}=O_{p}\left(\dfrac{1}{\sqrt{Tp_{2}}}\nu_{\min}^{-1}(\Hb_{1})\nu_{\min}^{-1}(\Hb_{2})\right);$$
	
	$$\|\cI\cI\cI\|_{F}=\left\|\dfrac{1}{\sqrt{p_{1}}p_{2}}\left(\dfrac{1}{T}\sum_{t=1}^{T}\Eb_{t}\Wb_{2}\Wb_{2}^{\top}\Cb\Fb_{t}^{\top}\Rb^{\top}\Wb_{1}\right)\Zb^{-1/2}\right\|_{F}=O_{p}\left(\sqrt{\dfrac{p_{1}}{Tp_{2}}}\nu_{\min}^{-1}(\Hb_{2})\right);
	$$
	
	$$\|\cI\cV\|_{F}=\left\|\dfrac{1}{\sqrt{p_{1}}p_{2}}\left(\dfrac{1}{T}\sum_{t=1}^{T}\Eb_{t}\Wb_{2}\Wb_{2}^{\top}\Eb_{t}^{\top}\Wb_{1}\right)\Zb^{-1/2}\right\|_{F}=O_{p}\left(\dfrac{1}{\sqrt{Tp_{2}}}\nu_{\min}^{-1}(\Hb_{1})\nu_{\min}^{-2}(\Hb_{2})+\dfrac{1}{\sqrt{p_{1}}p_{2}}\nu_{\min}^{-1}(\Hb_{1})\nu_{\min}^{-2}(\Hb_{2})\right);$$
	Finally, we can get
	$$\dfrac{1}{p_{1}}\|\hat{\Rb}^{(1)}-\Rb\hat{\Hb}_{r}^{(1)}\|_{F}^{2}=O_{p}\left(\dfrac{1}{Tp_{2}\nu_{\min}^{2}(\Hb_{2})}+\dfrac{1}{Tp_{1}p_{2}\nu_{\min}^{2}(\Hb_{1})\nu_{\min}^{4}(\Hb_{2})}+\dfrac{1}{p_{1}^{2}p_{2}^{2}\nu_{\min}^{2}(\Hb_{1})\nu_{\min}^{4}(\Hb_{2})}\right).$$
	
	It remains to show that $\hat{\Hb}^{(1)\top}_{r}\hat{\Hb}_{r}^{(1)} \stackrel{p}{\rightarrow} \Ib_{k_{1}}$.
	Under the condition that $\nu_{\min}^{2}(\Hb_{2}) \gg \max(\dfrac{1}{T}, \dfrac{1}{p_{2}})$, we can obtain $$\dfrac{1}{p_{1}}\|\hat{\Rb}^{(1)}-\Rb\hat{\Hb}_{r}^{(1)}\|_{F}^{2} = o_{p}\left(1\right),$$
	$$\|\dfrac{1}{p_{1}}\Rb^{\top}(\hat{\Rb}^{(1)}-\Rb\hat{\Hb}_{r}^{(1)})\|_{F} \leq (\dfrac{\|\Rb\|_{F}^{2}}{p_{1}}\dfrac{\|\hat{\Rb}^{(1)}-\Rb\hat{\Hb}_{r}^{(1)}\|_{F}^{2}}{p_{1}})^{1/2}=o_{p}(1), \ \|\dfrac{1}{p_{1}}\hat{\Rb}^{(1)\top}(\hat{\Rb}^{(1)}-\Rb\hat{\Hb}_{r}^{(1)})\|_{F}=o_{p}(1).$$
	Note that $p_{1}^{-1}\hat{\Rb}^{(1)\top}\hat{\Rb}^{(1)}=\Ib_{k_{1}}$, while $p_{1}^{-1}\Rb^{\top}\Rb=\Ib_{k_{1}}$, then $$\Ib_{k_{1}}=\dfrac{1}{p_{1}}\hat{\Rb}^{(1)\top}\Rb\hat{\Hb}_{r}^{(1)}+o_{p}(1)=\hat{\Hb}_{r}^{(1)\top}\hat{\Hb}_{r}^{(1)}+o_{p}(1).$$
	
	In the following we show the row-wise consistency of $\hat{\Rb}^{(1)}$.
	By equation (\ref{S1.1}), we have
	$$\begin{aligned}
	\hat{\bR}_{i\cdot}^{(1)}-\hat{\Hb}_{r}^{(1)\top}\bR_{i\cdot}&=\dfrac{1}{T\sqrt{p_{1}}p_{2}}\Zb^{-1/2}\left(\sum_{t=1}^{T}\Wb_{1}^{\top}\Eb_{t}\Wb_{2}\Wb_{2}^{\top}\Cb\Fb_{t}^{\top}\bR_{i\cdot}\right)+\dfrac{1}{T\sqrt{p_{1}}p_{2}}\Zb^{-1/2}\left(\sum_{t=1}^{T}\Wb_{1}^{\top}\Rb\Fb_{t}\Cb^{\top}\Wb_{2}\Wb_{2}^{\top}\be_{t,i\cdot}\right)\\
	&+\dfrac{1}{T\sqrt{p_{1}}p_{2}}\Zb^{-1/2}\left(\sum_{t=1}^{T}\Wb_{1}^{\top}\Eb_{t}\Wb_{2}\Wb_{2}^{\top}\be_{t,i\cdot}\right).
	\end{aligned}$$
	By Lemma \ref{basic-lemma} (3), we have
	$$\begin{aligned} \left\|\dfrac{1}{T\sqrt{p_{1}}p_{2}}\Zb^{-1/2}\left(\sum_{t=1}^{T}\Wb_{1}^{\top}\Eb_{t}\Wb_{2}\Wb_{2}^{\top}\Cb\Fb_{t}^{\top}\bR_{i\cdot}\right)\right\|_{2}&\leq \dfrac{1}{T\sqrt{p_{1}}p_{2}}\left\|\Zb^{-1/2}\right\|_{F}\left\|\sum_{t=1}^{T}\Wb_{1}^{\top}\Eb_{t}\Wb_{2}\Fb_{t}^{\top}\right\|_{F}\left\|\Wb_{2}^{\top}\Cb\right\|_{2}\\
	&=O_{p}\left(\dfrac{1}{\sqrt{Tp_{1}p_{2}}\nu_{\min}(\Hb_{1})\nu_{\min}(\Hb_{2})}\right).
	\end{aligned}$$
	Similar to the proof of Lemma \ref{basic-lemma} (3), we can get $\left\|\sum_{t=1}^{T}\Fb_{t}\Wb_{2}^{\top}\be_{t,i\cdot}\right\|_{F}^{2}=O_{p}\left(Tp_{2}\right)$, then
	$$\begin{aligned}
	\left\|\dfrac{1}{T\sqrt{p_{1}}p_{2}}\Zb^{-1/2}\left(\sum_{t=1}^{T}\Wb_{1}^{\top}\Rb\Fb_{t}\Cb^{\top}\Wb_{2}\Wb_{2}^{\top}\be_{t,i\cdot}\right)\right\|_{2}&\leq\dfrac{1}{T\sqrt{p_{1}}p_{2}}\left\|\Zb^{-1/2}\right\|_{F}\left\|\Wb_{1}^{\top}\Rb\right\|_{2}\left\|\Cb^{\top}\Wb_{2}\right\|_{2}\left\|\sum_{t=1}^{T}\Fb_{t}\Wb_{2}^{\top}\be_{t,i\cdot}\right\|_{F}\\
	&=O_{p}\left(\dfrac{1}{\sqrt{Tp_{2}}\nu_{\min}(\Hb_{2})}\right).
	\end{aligned}$$
	Similar to the proof of Lemma \ref{basic-lemma} (5), we can also get $\left\|\sum_{t=1}^{T}\Wb_{1}^{\top}\Eb_{t}\Wb_{2}\Wb_{2}^{\top}\be_{t,i\cdot}\right\|_{F}^{2}=O_{p}\left(Tp_{1}p_{2}^{3}+T^{2}p_{2}^{2}\right)$,
	
	$$\begin{aligned}
	\left\|\dfrac{1}{T\sqrt{p_{1}}p_{2}}\Zb^{-1/2}\left(\sum_{t=1}^{T}\Wb_{1}^{\top}\Eb_{t}\Wb_{2}\Wb_{2}^{\top}\be_{t,i\cdot}\right)\right\|_{2}&\leq\dfrac{1}{T\sqrt{p_{1}}p_{2}}\left\|\Zb^{-1/2}\right\|_{F}\left\|\sum_{t=1}^{T}\Wb_{1}^{\top}\Eb_{t}\Wb_{2}\Wb_{2}^{\top}\be_{t,i\cdot}\right\|_{F}\\
	&=O_{p}\left(\dfrac{1}{\sqrt{Tp_{1}p_{2}}\nu_{\min}(\Hb_{1})\nu_{\min}^{2}(\Hb_{2})}+\dfrac{1}{p_{1}p_{2}\nu_{\min}(\Hb_{1})\nu_{\min}^{2}(\Hb_{2})}\right).
	\end{aligned}$$

   Combining the above results, for any $i \leq p_{1}$, we can get
   $$\|\hat{\bR}_{i\cdot}^{(1)}-\hat{\Hb}_{r}^{(1)\top}\bR_{i\cdot}\|_{2}^{2}=O_{p}\left(\dfrac{1}{Tp_{2}\nu_{\min}^{2}(\Hb_{2})}+\dfrac{1}{Tp_{1}p_{2}\nu_{\min}^{2}(\Hb_{1})\nu_{\min}^{4}(\Hb_{2})}+\dfrac{1}{p_{1}^{2}p_{2}^{2}\nu_{\min}^{2}(\Hb_{1})\nu_{\min}^{4}(\Hb_{2})}\right).$$
	
    \end{proof}
	
	\subsection{Proof of Theorem \ref{convergence of an iteration}}

\begin{proof}
	First we decompose $\sum_{t=1}^{T}\Xb_{t}^{\top}\hat{\Rb}^{(1)}\hat{\Fb}_{t}$ into four terms:
	$$\begin{aligned}
		\sum_{t=1}^{T}\Xb_{t}^{\top}\hat{\Rb}^{(1)}\hat{\Fb}_{t}&=\dfrac{1}{p_{1}p_{2}}\sum_{t=1}^{T}\Xb_{t}^{\top}\hat{\Rb}^{(1)}\Wb_{1}^{\top}\Xb_{t}\Wb_{2}\\
		&=\dfrac{1}{p_{1}p_{2}}\sum_{t=1}^{T}(\Rb\Fb_{t}\Cb^{\top}+\Eb_{t})^{\top}\hat{\Rb}^{(1)}\Wb_{1}^{\top}(\Rb\Fb_{t}\Cb^{\top}+\Eb_{t})\Wb_{2}\\
		&=\dfrac{1}{p_{1}p_{2}}\sum_{t=1}^{T}\Cb\Fb_{t}^{\top}\Rb^{\top}\hat{\Rb}^{(1)}\Wb_{1}^{\top}\Rb\Fb_{t}\Cb^{\top}\Wb_{2}+\dfrac{1}{p_{1}p_{2}}\sum_{t=1}^{T}\Eb_{t}^{\top}\hat{\Rb}^{(1)}\Wb_{1}^{\top}\Rb\Fb_{t}\Cb^{\top}\Wb_{2}\\
		&+\dfrac{1}{p_{1}p_{2}}\sum_{t=1}^{T}\Cb\Fb_{t}^{\top}\Rb^{\top}\hat{\Rb}^{(1)}\Wb_{1}^{\top}\Eb_{t}\Wb_{2}+\dfrac{1}{p_{1}p_{2}}\sum_{t=1}^{T}\Eb_{t}^{\top}\hat{\Rb}^{(1)}\Wb_{1}^{\top}\Eb_{t}\Wb_{2}\\
		&=\bdelta_{1}^{(1)}+\bdelta_{2}^{(1)}+\bdelta_{3}^{(1)}+\bdelta_{4}^{(1)}.
	\end{aligned}$$
	For the first term,
	$$\begin{aligned}
		\left\|\dfrac{1}{T}\bdelta_{1}^{(1)}\right\|_{F}&=\left\|\dfrac{1}{Tp_{1}p_{2}}\sum_{t=1}^{T}\Cb\Fb_{t}^{\top}\Rb^{\top}\hat{\Rb}^{(1)}\Wb_{1}^{\top}\Rb\Fb_{t}\Cb^{\top}\Wb_{2}\right\|_{F} \asymp \left\|\dfrac{1}{Tp_{2}}\sum_{t=1}^{T}\Cb\Fb_{t}^{\top}\Wb_{1}^{\top}\Rb\Fb_{t}\Cb^{\top}\Wb_{2}\right\|_{F}\\
		&=p_{1}\|\Cb\|_{F}\left\|\dfrac{1}{T}\sum_{t=1}^{T}\Fb_{t}^{\top}\Fb_{t}\right\|_{F}\left\|\Hb_{1}\right\|_{2}\left\|\Hb_{2}\right\|_{2}=O_{p}\left(p_{1}\sqrt{p_{2}}\nu_{\min}(\Hb_{1})\nu_{\min}(\Hb_{2})\right).
	\end{aligned}$$
In addition, we also have $\|\bdelta_{1}^{(1)}\|_{F} \geq \|\bdelta_{1}^{(1)}\|_{2} \gtrsim Tp_{1}\sqrt{p_{2}}\nu_{\min}(\Hb_{1})\nu_{\min}(\Hb_{2})$, thus, we have  $$\|\bdelta_{1}^{(1)}\|_{F}\asymp
O_{p}\left(Tp_{1}\sqrt{p_{2}}\nu_{\min}(\Hb_{1})\nu_{\min}(\Hb_{2})\right).$$
	
	By Lemma \ref{need-in-proof}, we get  $$\left\|\sum_{s=1}^{T}\Eb_{s}^{\top}(\hat{\Rb}^{(1)}-\Rb\hat{\Hb}_{r}^{(1)})\Fb_{s}\right\|_{F}^{2}=O_{p}\left(\dfrac{p_{1}^{2}}{p_{2}\nu_{\min}^{2}(\Hb_{2})}+\dfrac{p_{1}}{\nu_{\min}^{2}(\Hb_{1})\nu_{\min}^{4}(\Hb_{2})}+\dfrac{T}{p_{2}\nu_{\min}^{2}(\Hb_{1})\nu_{\min}^{4}(\Hb_{2})}\right),$$ then
	$$\begin{aligned}
		\|\bdelta_{2}^{(1)}\|_{F}&=\left\|\dfrac{1}{p_{1}p_{2}}\sum_{t=1}^{T}\Eb_{t}^{\top}\hat{\Rb}^{(1)}\Wb_{1}^{\top}\Rb\Fb_{t}\Cb^{\top}\Wb_{2}\right\|_{F}=\left\|\dfrac{1}{p_{1}p_{2}}\sum_{t=1}^{T}\Eb_{t}^{\top}(\hat{\Rb}^{(1)}-\Rb\hat{\Hb}_{r}^{(1)}+\Rb\hat{\Hb}_{r}^{(1)})\Wb_{1}^{\top}\Rb\Fb_{t}\Cb^{\top}\Wb_{2}\right\|_{F}\\
		&\leq \left\|\dfrac{1}{p_{1}p_{2}}\sum_{t=1}^{T}\Eb_{t}^{\top}(\hat{\Rb}^{(1)}-\Rb\hat{\Hb}_{r}^{(1)})\Wb_{1}^{\top}\Rb\Fb_{t}\Cb^{\top}\Wb_{2}\right\|_{F}+\left\|\dfrac{1}{p_{1}p_{2}}\sum_{t=1}^{T}\Eb_{t}^{\top}\Rb\hat{\Hb}_{r}^{(1)}\Wb_{1}^{\top}\Rb\Fb_{t}\Cb^{\top}\Wb_{2}\right\|_{F}\\
		& \leq \left\|\sum_{t=1}^{T}\Eb_{t}^{\top}(\hat{\Rb}^{(1)}-\Rb\hat{\Hb}_{r}^{(1)})\Fb_{t}\right\|_{F}\left\|\Hb_{1}\right\|_{2}\left\|\Hb_{2}\right\|_{2}+\left\|\sum_{t=1}^{T}\Eb_{t}^{\top}\Rb\Fb_{t}\right\|_{F}\left\|\Hb_{1}\right\|_{2}\left\|\Hb_{2}\right\|_{2}\left\|\hat{\Hb}_{r}^{(1)}\right\|_{F}\\
		&=O_{p}\left(\sqrt{Tp_{1}p_{2}}\nu_{\min}(\Hb_{1})\nu_{\min}(\Hb_{2})+\dfrac{p_{1}}{\sqrt{p_{2}}}\nu_{\min}(\Hb_{1})+\dfrac{\sqrt{p_{1}}}{\nu_{\min}(\Hb_{2})}\right).
	\end{aligned}$$
	
	Similar to the proof of Lemma \ref{basic-lemma} (3), we also have
	$\EE\|\sum_{t=1}^{T}\Fb_{t}^{\top}\Wb_{1}^{\top}\Eb_{t}\Wb_{2}\|_{F}^{2}=O\left(Tp_{1}p_{2}\right)$, thus
	
	$$\begin{aligned}
		\left\|\bdelta_{3}^{(1)}\right\|_{F}&=\left\|\dfrac{1}{p_{1}p_{2}}\sum_{t=1}^{T}\Cb\Fb_{t}^{\top}\Rb^{\top}\hat{\Rb}^{(1)}\Wb_{1}^{\top}\Eb_{t}\Wb_{2}\right\|_{F} \asymp\dfrac{1}{p_{2}}\left\|\sum_{t=1}^{T}\Cb\Fb_{t}^{\top}\Wb_{1}^{\top}\Eb_{t}\Wb_{2}\right\|_{F}\\
		&\leq\dfrac{1}{p_{2}}\|\Cb\|_{F}\left\|\sum_{t=1}^{T}\Fb_{t}^{\top}\Wb_{1}^{\top}\Eb_{t}\Wb_{2}\right\|_{F}=O_{p}\left(\sqrt{Tp_{1}}\right).
	\end{aligned}$$
	
	In the proof of Lemma \ref{basic-lemma} (5), we have proved $\|\sum_{t=1}^{T}\Eb_{t}^{\top}\Wb_{1}^{\top}\Eb_{t}\Wb_{2}\|_{F}^{2}=O_{p}\left(Tp_{1}^{2}p_{2}^{2}+T^{2}p_{1}p_{2}\right)$, hence
	$$\begin{aligned}
		\left\|\bdelta_{4}^{(1)}\right\|_{F}&=\left\|\dfrac{1}{p_{1}p_{2}}\sum_{t=1}^{T}\Eb_{t}^{\top}\hat{\Rb}^{(1)}\Wb_{1}^{\top}\Eb_{t}\Wb_{2}\right\|_{F}=\left\|\dfrac{1}{p_{1}p_{2}}\sum_{t=1}^{T}\Eb_{t}^{\top}(\hat{\Rb}^{(1)}-\Rb\hat{\Hb}_{r}^{(1)}+\Rb\hat{\Hb}_{r}^{(1)})\Wb_{1}^{\top}\Eb_{t}\Wb_{2}\right\|_{F}\\
		&\leq \left\|\dfrac{1}{p_{1}p_{2}}\sum_{t=1}^{T}\Eb_{t}^{\top}(\hat{\Rb}^{(1)}-\Rb\hat{\Hb}_{r}^{(1)})\Wb_{1}^{\top}\Eb_{t}\Wb_{2}\right\|_{F}+\left\|\dfrac{1}{p_{1}p_{2}}\sum_{t=1}^{T}\Eb_{t}^{\top}\Rb\hat{\Hb}_{r}^{(1)}\Wb_{1}^{\top}\Eb_{t}\Wb_{2}\right\|_{F}\\
		&\leq \dfrac{1}{p_{1}p_{2}}\left\|\sum_{t=1}^{T}\Eb_{t}^{\top}\Wb_{1}^{\top}\Eb_{t}\Wb_{2}\right\|_{F}\left\|\hat{\Rb}^{(1)}-\Rb\hat{\Hb}_{r}^{(1)}\right\|_{F}+\dfrac{1}{p_{1}p_{2}}\left\|\sum_{t=1}^{T}\Eb_{t}^{\top}\Wb_{1}^{\top}\Eb_{t}\Wb_{2}\right\|_{F}\left\|\Rb\right\|_{F}\|\hat{\Hb}_{r}^{(1)}\|_{F}\\
		&=O_{p}\left(\sqrt{Tp_{1}}+\dfrac{T}{\sqrt{p_{2}}}\right).
	\end{aligned}$$
	
	Let $\Zb^{(1)}=(\sum_{t=1}^{T}\hat{\Fb}_{t}^{\top}\hat{\Rb}^{(1)\top}\Xb_{t})(\sum_{t=1}^{T}\Xb_{t}^{\top}\hat{\Rb}^{(1)}\hat{\Fb}_{t})$, then $\Zb^{(1)}=\sum_{i=1}^{4}\sum_{j=1}^{4}\bdelta_{i}^{(1)\top}\bdelta_{j}^{(1)}$. We can prove that $\|\Zb^{(1)}\|_{F} \gtrsim T^{2}p_{1}^{2}p_{2}\nu_{\min}^{2}(\Hb_{1})\nu_{\min}^{2}(\Hb_{2})$, hence $\|(\Zb^{(1)})^{-1/2}\|_{F}=O_{p}\left(\dfrac{1}{Tp_{1}\sqrt{p_{2}}\nu_{\min}(\Hb_{1})\nu_{\min}(\Hb_{2})}\right)$.
	
	$$\begin{aligned}
		\hat{\Cb}^{(1)}&=\sqrt{p_{2}}\left(\sum_{t=1}^{T}\Xb_{t}^{\top}\hat{\Rb}^{(1)}\hat{\Fb}_{t}\right)\left[\left(\sum_{t=1}^{T}\hat{\Fb}_{t}^{\top}\hat{\Rb}^{(1)\top}\Xb_{t}\right)\left(\sum_{t=1}^{T}\Xb_{t}^{\top}\hat{\Rb}^{(1)}\hat{\Fb}_{t}\right)\right]^{-1/2}\\
		&=\sqrt{p_{2}}\left(\dfrac{1}{p_{1}p_{2}}\sum_{t=1}^{T}\Cb\Fb_{t}^{\top}\Rb^{\top}\hat{\Rb}^{(1)}\Wb_{1}^{\top}\Rb\Fb_{t}\Cb^{\top}\Wb_{2}\right)\left(\Zb^{(1)}\right)^{-1/2}\\
		&+\sqrt{p_{2}}\left(\dfrac{1}{p_{1}p_{2}}\sum_{t=1}^{T}\Eb_{t}^{\top}\hat{\Rb}^{(1)}\Wb_{1}^{\top}\Rb\Fb_{t}\Cb^{\top}\Wb_{2}\right)\left(\Zb^{(1)}\right)^{-1/2}\\
		&+\sqrt{p_{2}}\left(\dfrac{1}{p_{1}p_{2}}\sum_{t=1}^{T}\Cb\Fb_{t}^{\top}\Rb^{\top}\hat{\Rb}^{(1)}\Wb_{1}^{\top}\Eb_{t}\Wb_{2}\right)\left(\Zb^{(1)}\right)^{-1/2}\\
		&+\sqrt{p_{2}}\left(\dfrac{1}{p_{1}p_{2}}\sum_{t=1}^{T}\Eb_{t}^{\top}\hat{\Rb}^{(1)}\Wb_{1}^{\top}\Eb_{t}\Wb_{2}\right)\left(\Zb^{(1)}\right)^{-1/2}.\\
		&=\sqrt{p_{2}}\left(\bdelta_{1}^{(1)}(\Zb^{(1)})^{-1/2}+\bdelta_{2}^{(1)}(\Zb^{(1)})^{-1/2}+\bdelta_{3}^{(1)}(\Zb^{(1)})^{-1/2}+\bdelta_{4}^{(1)}(\Zb^{(1)})^{-1/2}\right).
	\end{aligned}$$
	Let $\hat{\Hb}_{c}^{(1)}=\left(\dfrac{1}{p_{1}\sqrt{p_{2}}}\sum_{t=1}^{T}\Fb_{t}^{\top}\Rb^{\top}\hat{\Rb}^{(1)}\Wb_{1}^{\top}\Rb\Fb_{t}\Cb^{\top}\Wb_{2}\right)\left(\Zb^{(1)}\right)^{-1/2}$, then
	$$\hat{\Cb}^{(1)}-\Cb\hat{\Hb}_{c}^{(1)}=\sqrt{p_{2}}\bdelta_{2}^{(1)}(\Zb^{(1)})^{-1/2}+\sqrt{p_{2}}\bdelta_{3}^{(1)}(\Zb^{(1)})^{-1/2}+\sqrt{p_{2}}\bdelta_{4}^{(1)}(\Zb^{(1)})^{-1/2}.$$
	
	$$\begin{aligned}
		\|\hat{\Hb}_{c}^{(1)}\|_{F}&=\left\|\left(\dfrac{1}{p_{1}\sqrt{p_{2}}}\sum_{t=1}^{T}\Fb_{t}^{\top}\Rb^{\top}\hat{\Rb}^{(1)}\Wb_{1}^{\top}\Rb\Fb_{t}\Cb^{\top}\Wb_{2}\right)\left(\Zb^{(1)}\right)^{-1/2}\right\|_{F}\\
		&\leq \left\|\dfrac{1}{p_{1}\sqrt{p_{2}}}\sum_{t=1}^{T}\Fb_{t}^{\top}\Rb^{\top}\hat{\Rb}^{(1)}\Wb_{1}^{\top}\Rb\Fb_{t}\Cb^{\top}\Wb_{2}\right\|_{F}\left\|\left(\Zb^{(1)}\right)^{-1/2}\right\|_{F}\\
		&=O_{p}\left(1\right).
	\end{aligned}$$
	$$\|\bdelta_{2}^{(1)}(\Zb^{(1)})^{-1/2}\|_{F}=O_{p}\left(\dfrac{1}{\sqrt{Tp_{1}}}+\dfrac{1}{Tp_{2}}\nu_{\min}^{-1}(\Hb_{2})+\dfrac{1}{T\sqrt{p_{1}p_{2}}}\nu_{\min}^{-1}(\Hb_{1})\nu_{\min}^{-2}(\Hb_{2})\right).$$
	$$\|\bdelta_{3}^{(1)}(\Zb^{(1)})^{-1/2}\|_{F}=O_{p}\left(
	\dfrac{1}{\sqrt{Tp_{1}p_{2}}}\nu_{\min}^{-1}(\Hb_{1})\nu_{\min}^{-1}(\Hb_{2})\right).$$
	$$\|\bdelta_{4}^{(1)}(\Zb^{(1)})^{-1/2}\|_{F}=O_{p}\left(\dfrac{1}{\sqrt{Tp_{1}p_{2}}}\nu_{\min}^{-1}(\Hb_{1})\nu_{\min}^{-1}(\Hb_{2})+\dfrac{1}{p_{1}p_{2}}\nu_{\min}^{-1}(\Hb_{1})\nu_{\min}^{-1}(\Hb_{2})\right).$$
	As a result $$\dfrac{1}{p_{2}}\|\hat{\Cb}^{(1)}-\Cb\hat{\Hb}_{c}^{(1)}\|_{F}^{2}=O_{p}\left(\dfrac{1}{Tp_{1}}+\dfrac{1}{p_{1}^{2}p_{2}^{2}\nu_{\min}^{2}(\Hb_{1})\nu_{\min}^{2}(\Hb_{2})}+\dfrac{1}{T^{2}p_{2}^{2}\nu_{\min}^{2}(\Hb_{2})}+\dfrac{1}{Tp_{1}p_{2}\nu_{\min}^{2}(\Hb_{1})\nu_{\min}^{2}(\Hb_{2})}\right).$$
	
	$$\Big\|\dfrac{1}{p_{2}}\Cb^{\top}(\hat{\Cb}^{(1)}-\Cb\hat{\Hb}_{c}^{(1)})\Big\|_{F}\leq (\dfrac{\|\Cb\|_{F}^{2}}{p_{2}}\dfrac{\|\hat{\Cb}^{(1)}-\Cb\hat{\Hb}_{c}^{(1)}\|_{F}^{2}}{p_{2}})^{1/2}=o_{p}\left(1\right), \Big\|\dfrac{1}{p_{2}}\hat{\Cb}^{(1)\top}(\hat{\Cb}^{(1)}-\Cb\hat{\Hb}_{c}^{(1)})\Big\|_{F}=o_{p}\left(1\right).$$
	Note that $p_{2}^{-1}\Cb^{\top}\Cb=\Ib_{k_{2}}, \ p_{2}^{-1}\hat{\Cb}^{(1)\top}\hat{\Cb}^{(1)}=\Ib_{k_{2}}$, then
	$$\Ib_{k_{2}}=\dfrac{1}{p_{2}}\hat{\Cb}^{(1)\top}\Cb\hat{\Hb}_{c}^{(1)}+o_{p}\left(1\right)=\hat{\Hb}_{c}^{(1)\top}\hat{\Hb}_{c}^{(1)}+o_{p}\left(1\right), \ \hat{\Hb}_{c}^{(1)\top}\hat{\Hb}_{c}^{(1)}\stackrel{p}{\rightarrow} \Ib_{k_{2}}.$$
	
	In the  following, we shows the row-wise consistency of $\hat{\Cb}^{(1)}$.
	
	For $j \leq p_{2}$,
	$$\begin{aligned}
		\hat{\bC}_{j\cdot}^{(1)}-\hat{\Hb}_{c}^{(1)\top}\bC_{j\cdot}&=\dfrac{1}{p_{1}\sqrt{p_{2}}}(\Zb^{(1)})^{-1/2}\left(\sum_{t=1}^{T}\Wb_{2}^{\top}\Cb\Fb_{t}^{\top}\Rb^{\top}\Wb_{1}\hat{\Rb}^{(1)\top}\be_{t,\cdot j}\right)\\
		&+\dfrac{1}{p_{1}\sqrt{p_{2}}}(\Zb^{(1)})^{-1/2}\left(\sum_{t=1}^{T}\Wb_{2}^{\top}\Eb_{t}^{\top}\Wb_{1}\hat{\Rb}^{(1)\top}\Rb\Fb_{t}\bC_{j\cdot}\right)\\
		&+\dfrac{1}{p_{1}\sqrt{p_{2}}}(\Zb^{(1)})^{-1/2}\left(\sum_{t=1}^{T}\Wb_{2}^{\top}\Eb_{t}^{\top}\Wb_{1}\hat{\Rb}^{(1)\top}\be_{t,\cdot j}\right).
	\end{aligned}$$
	
	Firstly, similar to the proof of Lemma \ref{need-in-proof}, we can obtain
	$$\left\|\sum_{t=1}^{T}\Fb_{t}^{\top}(\hat{\Rb}^{(1)}-\Rb\hat{\Hb}_{r}^{(1)})^{\top}\be_{t,\cdot j}\right\|_{F}^{2}=O_{p}\left(\dfrac{p_{1}^{2}}{p_{2}^{2}\nu_{\min}^{2}(\Hb_{2})}+\dfrac{p_{1}}{p_{2}\nu_{\min}^{2}(\Hb_{1})\nu_{\min}^{4}(\Hb_{2})}+\dfrac{T}{p_{2}^{2}\nu_{\min}^{2}(\Hb_{1})\nu_{\min}^{4}(\Hb_{2})}\right),$$
	
	$$\begin{aligned}
		\left\|\dfrac{1}{p_{1}\sqrt{p_{2}}}\sum_{t=1}^{T}\Wb_{2}^{\top}\Cb\Fb_{t}^{\top}\Rb^{\top}\Wb_{1}\hat{\Rb}^{(1)\top}\be_{t,\cdot j}\right\|_{2}&\leq\dfrac{1}{p_{1}\sqrt{p_{2}}}\left\|\sum_{t=1}^{T}\Fb_{t}^{\top}\hat{\Rb}^{(1)\top}\be_{t,\cdot j}\right\|_{F}\|\Wb_{2}^{\top}\Cb\|_{2}\|\Rb^{\top}\Wb_{1}\|_{2}\\
		&\leq\dfrac{1}{p_{1}\sqrt{p_{2}}}\left\|\sum_{t=1}^{T}\Fb_{t}^{\top}(\hat{\Rb}^{(1)}-\Rb_{t}\hat{\Hb}_{r}^{(1)})^{\top}\be_{t,\cdot j}\right\|_{F}\|\Wb_{2}^{\top}\Cb\|_{2}\|\Rb^{\top}\Wb_{1}\|_{2}\\
		&+\dfrac{1}{p_{1}\sqrt{p_{2}}}\left\|\sum_{t=1}^{T}\Fb_{t}^{\top}\Rb_{t}^{\top}\be_{t,\cdot j}\right\|_{F}\|\Wb_{2}^{\top}\Cb\|_{2}\|\Rb^{\top}\Wb_{1}\|_{2}\|\hat{\Hb}_{c}^{(1)}\|_{F},
	\end{aligned}$$
	then,
	$$\begin{aligned}
		\left\|\dfrac{1}{p_{1}\sqrt{p_{2}}}(\Zb^{(1)})^{-1/2}\left(\sum_{t=1}^{T}\Wb_{2}^{\top}\Cb\Fb_{t}^{\top}\Rb^{\top}\Wb_{1}\hat{\Rb}^{(1)\top}\be_{t,\cdot j}\right)\right\|_{2}&\leq \dfrac{1}{p_{1}\sqrt{p_{2}}}\left\|\sum_{t=1}^{T}\Wb_{2}^{\top}\Cb\Fb_{t}^{\top}\Rb^{\top}\Wb_{1}\hat{\Rb}^{(1)\top}\be_{t,\cdot j}\right\|_{F}\left\|(\Zb^{(1)})^{-1/2}\right\|_{F}\\
		&=O_{p}\left(\dfrac{1}{Tp_{2}\nu_{\min}(\Hb_{2})}+\dfrac{1}{T\sqrt{p_{1}p_{2}}\nu_{\min}(\Hb_{1})\nu_{\min}^{2}(\Hb_{2})}+\dfrac{1}{\sqrt{Tp_{1}}}\right).
	\end{aligned}$$
	Secondly, $$\begin{aligned}
		\left\|\dfrac{1}{p_{1}\sqrt{p_{2}}}(\Zb^{(1)})^{-1/2}\left(\sum_{t=1}^{T}\Wb_{2}^{\top}\Eb_{t}^{\top}\Wb_{1}\hat{\Rb}^{(1)\top}\Rb\Fb_{t}\bC_{j\cdot}\right)\right\|_{2}&\asymp\left\|\dfrac{1}{\sqrt{p_{2}}}(\Zb^{(1)})^{-1/2}\left(\sum_{t=1}^{T}\Wb_{2}^{\top}\Eb_{t}^{\top}\Wb_{1}\Fb_{t}\right)\right\|_{2}\\
		&=O_{p}\left(\dfrac{1}{\sqrt{Tp_{1}p_{2}}\nu_{\min}(\Hb_{1})\nu_{\min}(\Hb_{2})}\right).
	\end{aligned}$$
	Thirdly,
	$$\begin{aligned}
		\left\|\sum_{t=1}^{T}\Wb_{2}^{\top}\Eb_{t}^{\top}\Wb_{1}\hat{\Rb}^{(1)\top}\be_{t,\cdot j}\right\|_{F}&=\left\|\sum_{t=1}^{T}\Wb_{2}^{\top}\Eb_{t}^{\top}\Wb_{1}(\hat{\Rb}^{(1)}-\Rb\hat{\Hb}_{r}^{(1)}+\Rb\hat{\Hb}_{r}^{(1)})^{\top}\be_{t,\cdot j}\right\|_{F}\\
		&\leq\left\|\sum_{t=1}^{T}\Wb_{2}^{\top}\Eb_{t}^{\top}\Wb_{1}\be_{t,\cdot j}^{\top}\right\|_{F}\left\|\hat{\Rb}^{(1)}-\Rb\hat{\Hb}_{r}^{(1)}\right\|_{F}\\
		&+\left\|\sum_{t=1}^{T}\Wb_{2}^{\top}\Eb_{t}^{\top}\Wb_{1}\be_{t,\cdot j}^{\top}\right\|_{F}\left\|\Rb\right\|_{F}\|\hat{\Hb}_{r}^{(1)\top}\|_{F},
	\end{aligned}$$
	
	$$\begin{aligned}
		\left\|\dfrac{1}{p_{1}\sqrt{p_{2}}}(\Zb^{(1)})^{-1/2}\left(\sum_{t=1}^{T}\Wb_{2}^{\top}\Eb_{t}^{\top}\Wb_{1}\hat{\Rb}^{(1)\top}\be_{t,\cdot j}\right)\right\|_{F}&\leq \dfrac{1}{p_{1}\sqrt{p_{2}}}\left\|(\Zb^{(1)})^{-1/2}\left(\sum_{t=1}^{T}\Wb_{2}^{\top}\Eb_{t}^{\top}\Wb_{1}\hat{\Rb}^{(1)\top}\be_{t,\cdot j}\right)\right\|_{F}\\
		&=O_{p}\left(\dfrac{1}{p_{1}p_{2}\nu_{\min}(\Hb_{1})\nu_{\min}(\Hb_{2})}+\dfrac{1}{\sqrt{Tp_{1}p_{2}}\nu_{\min}(\Hb_{1})\nu_{\min}(\Hb_{2})}\right).
	\end{aligned}$$

    As a result, we get $$\|\hat{\bC}_{j\cdot}^{(1)}-\hat{\Hb}_{c}^{(1)\top}\bC_{j\cdot}\|_{2}^{2}=O_{p}\left(\dfrac{1}{Tp_{1}}+\dfrac{1}{p_{1}^{2}p_{2}^{2}\nu_{\min}^{2}(\Hb_{1})\nu_{\min}^{2}(\Hb_{2})}+\dfrac{1}{T^{2}p_{2}^{2}\nu_{\min}^{2}(\Hb_{2})}+\dfrac{1}{Tp_{1}p_{2}\nu_{\min}^{2}(\Hb_{1})\nu_{\min}^{2}(\Hb_{2})}\right).$$
	
    \end{proof}

\subsection{	 Proof of Theorem \ref{factor and common component}}
	\begin{proof}
		By definition,
		$$\begin{aligned}
			\hat{\Fb}_{t}^{(2)}&=\dfrac{1}{p_{1}p_{2}}\hat{\Rb}^{(1)\top}\Xb_{t}\hat{\Cb}^{(1)}=\dfrac{1}{p_{1}p_{2}}\hat{\Rb}^{(1)\top}\Rb\Fb_{t}\Cb^{\top}\hat{\Cb}^{(1)}+\dfrac{1}{p_{1}p_{2}}\hat{\Rb}^{(1)\top}\Eb_{t}\hat{\Cb}^{(1)}\\
			&=\dfrac{1}{p_{1}p_{2}}\hat{\Rb}^{(1)\top}\left(\Rb-\hat{\Rb}^{(1)}(\hat{\Hb}_{r}^{(1)})^{-1}+\hat{\Rb}^{(1)}(\hat{\Hb}_{r}^{(1)})^{-1}\right)\Fb_{t}\left(\Cb-\hat{\Cb}^{(1)}(\hat{\Hb}_{c}^{(1)})^{-1}+\hat{\Cb}^{(1)}(\hat{\Hb}_{c}^{(1)})^{-1}\right)^{\top}\hat{\Cb}^{(1)}\\
			&+\dfrac{1}{p_{1}p_{2}}\left(\hat{\Rb}^{(1)}-\Rb\hat{\Hb}_{r}^{(1)}+\Rb\hat{\Hb}_{r}^{(1)}\right)^{\top}\Eb_{t}\left(\hat{\Cb}^{(1)}-\Cb\hat{\Hb}_{c}^{(1)}+\Cb\hat{\Hb}_{c}^{(1)}\right).
		\end{aligned}$$
		
		Note that $\hat{\Rb}^{(1)\top}\hat{\Rb}^{(1)}=p_{1}\Ib_{k_{1}}$, $\hat{\Cb}^{(1)\top}\hat{\Cb}^{(1)}=p_{2}\Ib_{k_{2}}$, then
		$$\begin{aligned}
			\hat{\Fb}_{t}^{(2)}-(\hat{\Hb}_{r}^{(1)})^{-1}\Fb_{t}\left((\hat{\Hb}_{c}^{(1)})^{-1}\right)^{\top}&=\dfrac{1}{p_{1}p_{2}}\hat{\Rb}^{(1)\top}\left(\Rb-\hat{\Rb}^{(1)}(\hat{\Hb}_{r}^{(1)})^{-1}\right)\Fb_{t}\left(\Cb-\hat{\Cb}(\hat{\Hb}_{c}^{(1)})^{-1}\right)^{\top}\hat{\Cb}^{(1)}\\
			&+\dfrac{1}{p_{1}}\hat{\Rb}^{(1)\top}\left(\Rb-\hat{\Rb}^{(1)}(\hat{\Hb}_{r}^{(1)})^{-1}\right)\Fb_{t}\left((\hat{\Hb}_{c}^{(1)})^{-1}\right)^{\top}\\
			&+\dfrac{1}{p_{2}}(\hat{\Hb}_{r}^{(1)})^{-1}\Fb_{t}\left(\Cb-\hat{\Cb}(\hat{\Hb}_{c}^{(1)})^{-1}\right)^{\top}\hat{\Cb}^{(1)}\\
			&+\dfrac{1}{p_{1}p_{2}}\left(\hat{\Rb}^{(1)}-\Rb\hat{\Hb}_{r}^{(1)}\right)^{\top}\Eb_{t}\left(\hat{\Cb}^{(1)}-\Cb\hat{\Hb}_{c}^{(1)}\right)\\
			&+\dfrac{1}{p_{1}p_{2}}\left(\hat{\Rb}^{(1)}-\Rb\hat{\Hb}_{r}^{(1)}\right)^{\top}\Eb_{t}\Cb\hat{\Hb}_{c}^{(1)}+\dfrac{1}{p_{1}p_{2}}\hat{\Hb}_{r}^{(1)\top}\Rb^{\top}\Eb_{t}\left(\hat{\Cb}^{(1)}-\Cb\hat{\Hb}_{c}^{(1)}\right)\\
			&+\dfrac{1}{p_{1}p_{2}}\hat{\Hb}_{r}^{(1)\top}\Rb^{\top}\Eb_{t}\Cb\hat{\Hb}_{c}^{(1)}.
		\end{aligned}$$
		
		In Lemma \ref{Product of loadings}, we prove that
		$$\left\|\Rb^{\top}(\hat{\Rb}^{(1)}-\Rb\hat{\Hb}_{r}^{(1)})\right\|_{F}=O_{p}\left(\dfrac{\sqrt{p_{1}}}{\sqrt{Tp_{2}}\nu_{\min}(\Hb_{1})\nu_{\min}(\Hb_{2})}+\dfrac{1}{p_{2}\nu_{\min}(\Hb_{1})\nu_{\min}^{2}(\Hb_{2})}\right),$$
		$$\left\|\Cb^{\top}(\hat{\Cb}^{(1)}-\Cb\hat{\Hb}_{c}^{(1)})\right\|_{F}=O_{p}\left(\dfrac{1}{T\nu_{\min}(\Hb_{2})}+\dfrac{\sqrt{p_{2}}}{\sqrt{Tp_{1}}\nu_{\min}(\Hb_{1})\nu_{\min}(\Hb_{2})}+\dfrac{1}{p_{1}\nu_{\min}(\Hb_{1})\nu_{\min}(\Hb_{2})}\right).$$
		
		Therefore, by Cauchy-Schwartz inequality, Theorem \ref{no_iter_convergence}, Theorem \ref{convergence of an iteration}, and the bounds for $\|\Rb^{\top}\Eb_{t}\|_{F}^{2}$, $\|\Eb_{t}\Cb\|_{F}^{2}$ and $\|\Rb^{\top}\Eb_{t}\Cb\|_{F}^{2}$, we have
		$$\left\|\hat{\Fb}_{t}^{(2)}-(\hat{\Hb}_{r}^{(1)})^{-1}\Fb_{t}\left((\hat{\Hb}_{c}^{(1)})^{-1}\right)^{\top}\right\|_{F}=O_{p}\left(\dfrac{1}{\sqrt{T}p_{1}}+\dfrac{1}{\sqrt{T}p_{2}\nu_{\min}(\Hb_{2})}+\dfrac{1}{\sqrt{p_{1}p_{2}}}+\gamma_{f}\right),$$
		where $$\gamma_{f}=\dfrac{1}{\sqrt{Tp_{1}p_{2}}\nu_{\min}(\Hb_{1})\nu_{\min}(\Hb_{2})}+\dfrac{1}{p_{1}p_{2}\nu_{\min}(\Hb_{1})\nu_{\min}^{2}(\Hb_{2})}.$$
		
		Next, for any $t,i,j$
		$$\begin{aligned}
			\hat{S}_{t,ij}^{(2)}-S_{t,ij}&=\hat{\bR}_{i\cdot}^{(1)\top}\hat{\Fb}_{t}^{(2)}\hat{\bC}_{j\cdot}^{(1)}-\bR_{i\cdot}^{\top}\Fb_{t}\bC_{j\cdot}\\
			&=(\hat{\bR}_{i\cdot}^{(1)}-\hat{\Hb}_{r}^{(1)\top}\bR_{i\cdot}+\hat{\Hb}_{r}^{(1)\top}\bR_{i\cdot})^{\top}\hat{\Fb}_{t}^{(2)}(\hat{\bC}_{j\cdot}^{(1)}-\hat{\Hb}_{c}^{(1)\top}\bC_{j\cdot}+\hat{\Hb}_{c}^{(1)\top}\bC_{j\cdot})-\bR_{i\cdot}^{\top}\Fb_{t}\bC_{j\cdot}\\
			&=(\hat{\bR}_{i\cdot}^{(1)}-\hat{\Hb}_{r}^{(1)\top}\bR_{i\cdot})^{\top}\hat{\Fb}_{t}^{(2)}(\hat{\bC}_{j\cdot}^{(1)}-\hat{\Hb}_{c}^{(1)\top}\bC_{j\cdot})+\bR_{i\cdot}^{\top}\hat{\Hb}_{r}^{(1)}\hat{\Fb}_{t}^{(2)}(\hat{\bC}_{j\cdot}^{(1)}-\hat{\Hb}_{c}^{(1)\top}\bC_{j\cdot})\\
			&+(\hat{\bR}_{i\cdot}^{(1)}-\hat{\Hb}_{r}^{(1)\top}\bR_{i\cdot})^{\top}\hat{\Fb}_{t}^{(2)}\hat{\Hb}_{c}^{(1)\top}\bC_{j\cdot}+\bR_{i\cdot}^{\top}(\hat{\Hb}_{r}^{(1)}\hat{\Fb}_{t}^{(2)}\hat{\Hb}_{c}^{(1)\top}-\Fb_{t})\bC_{j\cdot}.
		\end{aligned}$$
		Then, by Cauchy-Schwartz inequality, Theorem \ref{no_iter_convergence}, Theorem \ref{convergence of an iteration} and the consistency of the estimators for the factor matrix, we have
		$$|\hat{S}_{t,ij}^{(2)}-S_{t,ij}|=O_{p}\left(\dfrac{1}{\sqrt{Tp_{1}}}+\dfrac{1}{\sqrt{p_{1}p_{2}}}+\dfrac{1}{\sqrt{Tp_{2}}\nu_{\min}(\Hb_{2})}+\dfrac{1}{\sqrt{Tp_{1}p_{2}}\nu_{\min}(\Hb_{1})\nu_{\min}^{2}(\Hb_{2})}+\dfrac{1}{p_{1}p_{2}\nu_{\min}(\Hb_{1})\nu_{\min}^{2}(\Hb_{2})}\right).$$
	\end{proof}

\subsection{		Proof of Theorem \ref{the iteration results}}	
	\begin{proof}

		Recall first that $$\hat{\Rb}^{(s+1)}=\sqrt{p_{1}}\left(\sum_{t=1}^{T}\Xb_{t}\hat{\Cb}^{(s)}\hat{\Fb}_{t}^{(s+1)\top}\right)\left[\left(\sum_{t=1}^{T}\hat{\Fb}_{t}^{(s+1)}\hat{\Cb}^{(s)\top}\Xb_{t}^{\top}\right)\left(\sum_{t=1}^{T}\Xb_{t}\hat{\Cb}^{(s)}\hat{\Fb}_{t}^{(s+1)\top}\right)\right]^{-1/2}.$$
		
		$$
		\begin{aligned}
			\sum_{t=1}^{T}\Xb_{t}\hat{\Cb}^{(s)}\hat{\Fb}_{t}^{(s+1)\top}&=\dfrac{1}{p_{1}p_{2}}\sum_{t=1}^{T}\Xb_{t}\hat{\Cb}^{(s)}\hat{\Cb}^{(s)\top}\Xb_{t}^{\top}\hat{\Rb}^{(s)}\\
			&=\dfrac{1}{p_{1}p_{2}}\sum_{t=1}^{T}(\Rb\Fb_{t}\Cb^{\top}+\Eb_{t})\hat{\Cb}^{(s)}\hat{\Cb}^{(s)\top}(\Rb\Fb_{t}\Cb^{\top}+\Eb_{t})^{\top}\hat{\Rb}^{(s)}\\
			&=\dfrac{1}{p_{1}p_{2}}\sum_{t=1}^{T}\Rb\Fb_{t}\Cb^{\top}\hat{\Cb}^{(s)}\hat{\Cb}^{(s)\top}\Cb\Fb_{t}^{\top}\Rb^{\top}\hat{\Rb}^{(s)}+\dfrac{1}{p_{1}p_{2}}\sum_{t=1}^{T}\Rb\Fb_{t}\Cb^{\top}\hat{\Cb}^{(s)}\hat{\Cb}^{(s)\top}\Eb_{t}^{\top}\hat{\Rb}^{(s)}\\
			&+\dfrac{1}{p_{1}p_{2}}\sum_{t=1}^{T}\Eb_{t}\hat{\Cb}^{(s)}\hat{\Cb}^{(s)\top}\Cb\Fb_{t}^{\top}\Rb^{\top}\hat{\Rb}^{(s)}+\dfrac{1}{p_{1}p_{2}}\sum_{t=1}^{T}\Eb_{t}\hat{\Cb}^{(s)}\hat{\Cb}^{(s)\top}\Eb_{t}^{\top}\hat{\Rb}^{(s)}\\
			&=\bdelta_{1}^{(s+1)}+\bdelta_{2}^{(s+1)}+\bdelta_{3}^{(s+1)}+\bdelta_{4}^{(s+1)}.
		\end{aligned}$$
		
		As
		$$\begin{aligned}
			\|\dfrac{1}{T}\bdelta_{1}^{(s+1)}\|_{F}^{2}&=\|\dfrac{1}{Tp_{1}p_{2}}\sum_{t=1}^{T}\Rb\Fb_{t}\Cb^{\top}\hat{\Cb}^{(s)}\hat{\Cb}^{(s)\top}\Cb\Fb_{t}^{\top}\Rb^{\top}\hat{\Rb}^{(s)}\|_{F}^{2}\\
			&\asymp p_{2}^{2}\|\dfrac{1}{T}\sum_{t=1}^{T}\Rb\Fb_{t}\Fb_{t}^{\top}\|_{F}^{2}\leq p_{2}^{2}\|\Rb\|_{F}^{2}\|\dfrac{1}{T}\sum_{t=1}^{T}\Fb_{t}\Fb_{t}^{\top}\|_{F}^{2}=O_{p}\left(p_{1}p_{2}^{2}\right),
		\end{aligned}$$
	we can also prove that $\|\bdelta_{1}^{(s+1)}\|_{F}^{2} \gtrsim T^{2}p_{1}p_{2}^{2}$, thus, $\|\bdelta_{1}^{(s+1)}\|_{F}^{2} \asymp O_{p}\left(T^{2}p_{1}p_{2}^{2}\right)$.

		By Lemma \ref{lemma11} (1) and (2), we can obtain,
		$$\begin{aligned}
		\|\bdelta_{2}^{(s+1)}\|_{F}^{2}&=\|\dfrac{1}{p_{1}p_{2}}\sum_{t=1}^{T}\Rb\Fb_{t}\Cb^{\top}\hat{\Cb}^{(s)}\hat{\Cb}^{(s)\top}\Eb_{t}^{\top}\hat{\Rb}^{(s)}\|_{F}^{2} \asymp \dfrac{1}{p_{1}^{2}}\|\sum_{t=1}^{T}\Rb\Fb_{t}\hat{\Cb}^{(s)\top}\Eb_{t}^{\top}\hat{\Rb}^{(s)}\|_{F}^{2}\\
		&=\dfrac{1}{p_{1}^{2}}\|\sum_{t=1}^{T}\Rb\Fb_{t}\hat{\Cb}^{(s)\top}\Eb_{t}^{\top}(\hat{\Rb}^{(s)}-\Rb\hat{\Hb}_{r}^{(s)})\|_{F}^{2}+\dfrac{1}{p_{1}^{2}}\|\sum_{t=1}^{T}\Rb\Fb_{t}\hat{\Cb}^{(s)\top}\Eb_{t}^{\top}\Rb\hat{\Hb}_{r}^{(s)}\|_{F}^{2}\\
		&\leq \dfrac{1}{p_{1}^{2}}\|\Rb\|_{F}^{2}\|\sum_{t=1}^{T}\Fb_{t}\hat{\Cb}^{(s)\top}\Eb_{t}^{\top}\|_{F}^{2}\|\hat{\Rb}^{(s)}-\Rb\hat{\Hb}_{r}^{(s)}\|_{F}^{2}+	\dfrac{1}{p_{1}^{2}}\|\Rb\|_{F}^{2}\|\sum_{t=1}^{T}\Fb_{t}\hat{\Cb}^{(s)\top}\Eb_{t}^{\top}\Rb\|_{F}^{2}\|\hat{\Hb}_{r}^{(s)}\|_{F}^{2}\\
		&=O_{p}\left(Tp_{1}p_{2}^{2}w_{r}^{(s)}w_{c}^{(s)}+Tp_{1}p_{2}w_{r}^{(s)}+Tp_{2}^{2}w_{c}^{(s)}+Tp_{2}\right).
		\end{aligned}$$
	
	    Similarly,
		$$
		\|\bdelta_{3}^{(s+1)}\|_{F}^{2}=\|\dfrac{1}{p_{1}p_{2}}\sum_{t=1}^{T}\Eb_{t}\hat{\Cb}^{(s)}\hat{\Cb}^{(s)\top}\Cb\Fb_{t}^{\top}\Rb^{\top}\hat{\Rb}^{(s)}\|_{F}^{2} \asymp  \|\sum_{t=1}^{T}\Eb_{t}\hat{\Cb}^{(s)}\Fb_{t}^{\top}\|_{F}^{2}=O_{p}\left(Tp_{1}p_{2}^{2}w_{c}^{(s)}+Tp_{1}p_{2}\right).$$
		
		$$\begin{aligned}
			\|\bdelta_{4}^{(s+1)}\|_{F}^{2}&=\|\dfrac{1}{p_{1}p_{2}}\sum_{t=1}^{T}\Eb_{t}\hat{\Cb}^{(s)}\hat{\Cb}^{(s)\top}\Eb_{t}^{\top}\hat{\Rb}^{(s)}\|_{F}^{2}\\
			& \leq \|\dfrac{1}{p_{1}p_{2}}\sum_{t=1}^{T}\Eb_{t}\hat{\Cb}^{(s)}\hat{\Cb}^{(s)\top}\Eb_{t}^{\top}(\hat{\Rb}^{(s)}-\Rb\hat{\Hb}_{r}^{(s)})\|_{F}^{2}+\|\dfrac{1}{p_{1}p_{2}}\sum_{t=1}^{T}\Eb_{t}\hat{\Cb}^{(s)}\hat{\Cb}^{(s)\top}\Eb_{t}^{\top}\Rb\hat{\Hb}_{r}^{(s)}\|_{F}^{2}\\
			&\leq \dfrac{1}{p_{1}^{2}p_{2}^{2}}\left(\sum_{t=1}^{T}\|\Eb_{t}\hat{\Cb}^{(s)}\|_{F}^{2}\right)^{2}\|\hat{\Rb}^{(s)}-\Rb\hat{\Hb}_{r}^{(s)}\|_{F}^{2}+\dfrac{1}{p_{1}^{2}p_{2}^{2}}\|\sum_{t=1}^{T}\Eb_{t}\hat{\Cb}^{(s)}\hat{\Cb}^{(s)\top}\Eb_{t}^{\top}\Rb\|_{F}^{2}\|\hat{\Hb}_{r}^{(s)}\|_{F}^{2}\\
			&=O_{p}\left(Tp_{2}+\dfrac{T^{2}}{p_{1}}+T^{2}p_{1}w_{r}^{(s)}+T^{2}p_{1}p_{2}w_{r}^{(s)}w_{c}^{(s)}+T^{2}p_{1}p_{2}^{2}w_{r}^{(s)}w_{c}^{(s)2}+\dfrac{T^{2}p_{2}^{2}}{p_{1}}w_{c}^{(s)2}+Tp_{2}^{2}w_{c}^{(s)2}\right),
		\end{aligned}$$
		the last equality is according to Lemma \ref{lemma9} (1) and Lemma \ref{lemma10} (1).
		
		Let $\Zb^{(s+1)}=\left(\sum_{t=1}^{T}\hat{\Fb}_{t}^{(s+1)}\hat{\Cb}^{(s)\top}\Xb_{t}^{\top}\right)\left(\sum_{t=1}^{T}\Xb_{t}\hat{\Cb}^{(s)}\hat{\Fb}_{t}^{(s+1)\top}\right)$,
		
		$$
		\begin{aligned}
			\Zb^{(s+1)}&=\dfrac{1}{p_{1}^{2}p_{2}^{2}}\left(\sum_{t=1}^{T}\hat{\Rb}^{(s)\top}\Xb_{t}\hat{\Cb}^{(s)}\hat{\Cb}^{(s)\top}\Xb_{t}^{\top}\right)\left(\sum_{t=1}^{T}\Xb_{t}\hat{\Cb}^{(s)}\hat{\Cb}^{(s)\top}\Xb_{t}^{\top}\hat{\Rb}^{(s)}\right)\\
			&=\sum_{i=1}^{4}\sum_{j=1}^{4}\bdelta_{i}^{(s+1)\top}\bdelta_{j}^{(s+1)},
		\end{aligned}
		$$
		then we can prove $\left\|(\Zb^{(s+1)})^{-1/2}\right\|_{F}^{2}=O_{p}\left(\dfrac{1}{T^{2}p_{1}p_{2}^{2}}\right).$
		
		$$
		\begin{aligned}
			\hat{\Rb}^{(s+1)}&=\sqrt{p_{1}}\left(\sum_{t=1}^{T}\Xb_{t}\hat{\Cb}^{(s)}\hat{\Fb}_{t}^{(s+1)\top}\right)\left[\left(\sum_{t=1}^{T}\hat{\Fb}_{t}^{(s+1)}\hat{\Cb}^{(s)\top}\Xb_{t}^{\top}\right)\left(\sum_{t=1}^{T}\Xb_{t}\hat{\Cb}^{(s)}\hat{\Fb}_{t}^{(s+1)\top}\right)\right]^{-1/2}\\
			&=\sqrt{p_{1}}\left(\sum_{t=1}^{T}\Xb_{t}\hat{\Cb}^{(s)}\hat{\Fb}_{t}^{(s+1)\top}\right)\left(\Zb^{(s+1)}\right)^{1/2}\\
			&=\Rb\left[\dfrac{1}{\sqrt{p_{1}}p_{2}}\left(\sum_{t=1}^{T}\Fb_{t}\Cb^{\top}\hat{\Cb}^{(s)}\hat{\Cb}^{(s)\top}\Cb\Fb_{t}^{\top}\Rb^{\top}\hat{\Rb}^{(s)}\right)\left(\Zb^{(s+1)}\right)^{-1/2}\right]\\
			&+\dfrac{1}{\sqrt{p_{1}}p_{2}}\left(\sum_{t=1}^{T}\Rb\Fb_{t}\Cb^{\top}\hat{\Cb}^{(s)}\hat{\Cb}^{(s)\top}\Eb_{t}^{\top}\hat{\Rb}^{(s)}\right)\left(\Zb^{(s+1)}\right)^{-1/2}\\
			&=\dfrac{1}{\sqrt{p_{1}}p_{2}}\left(\sum_{t=1}^{T}\Eb_{t}\hat{\Cb}^{(s)}\hat{\Cb}^{(s)\top}\Cb\Fb_{t}^{\top}\Rb^{\top}\hat{\Rb}^{(s)}\right)\left(\Zb^{(s+1)}\right)^{-1/2}\\
			&+\dfrac{1}{\sqrt{p_{1}}p_{2}}\left(\sum_{t=1}^{T}\Eb_{t}\hat{\Cb}^{(s)}\hat{\Cb}^{(s)\top}\Eb_{t}^{\top}\hat{\Rb}^{(s)}\right)\left(\Zb^{(s+1)}\right)^{-1/2}.\\
		\end{aligned}
		$$
		Denote $\hat{\Hb}_{r}^{(s+1)}=\dfrac{1}{\sqrt{p_{1}}p_{2}}\left(\sum_{t=1}^{T}\Fb_{t}\Cb^{\top}\hat{\Cb}^{(s)}\hat{\Cb}^{(s)\top}\Cb\Fb_{t}^{\top}\Rb^{\top}\hat{\Rb}^{(s)}\right)\left(\Zb^{(s+1)}\right)^{-1/2}$, we have
		$$
		\begin{aligned}
		\hat{\Rb}^{(s+1)}-\Rb\hat{\Hb}_{r}^{(s+1)}&=\dfrac{1}{\sqrt{p_{1}}p_{2}}\left(\sum_{t=1}^{T}\Rb\Fb_{t}\Cb^{\top}\hat{\Cb}^{(s)}\hat{\Cb}^{(s)\top}\Eb_{t}^{\top}\hat{\Rb}^{(s)}\right)\left(\Zb^{(s+1)}\right)^{-1/2}\\
		&+\dfrac{1}{\sqrt{p_{1}}p_{2}}\left(\sum_{t=1}^{T}\Eb_{t}\hat{\Cb}^{(s)}\hat{\Cb}^{(s)\top}\Cb\Fb_{t}^{\top}\Rb^{\top}\hat{\Rb}^{(s)}\right)\left(\Zb^{(s+1)}\right)^{-1/2}\\
		&+\dfrac{1}{\sqrt{p_{1}}p_{2}}\left(\sum_{t=1}^{T}\Eb_{t}\hat{\Cb}^{(s)}\hat{\Cb}^{(s)\top}\Eb_{t}^{\top}\hat{\Rb}^{(s)}\right)\left(\Zb^{(s+1)}\right)^{-1/2}\\
		&=\sqrt{p_{1}}(\bdelta_{2}^{(s+1)}+\bdelta_{3}^{(s+1)}+\bdelta_{4}^{(s+1)})(\Zb^{(s+1)})^{-1/2},
		\end{aligned}
		$$
		thus, $$
		\dfrac{1}{p_{1}}\|\hat{\Rb}^{(s+1)}-\Rb\hat{\Hb}_{r}^{(s+1)}\|_{F}^{2}
		=O_{p}\left(\dfrac{1}{Tp_{2}}+\dfrac{1}{p_{1}^{2}p_{2}^{2}}+\dfrac{1}{p_{2}^{2}}w_{r}^{(s)}+\dfrac{1}{T}w_{c}^{(s)}+\dfrac{w_{r}^{(s)}w_{c}^{(s)}}{p_{2}}+w_{r}^{(s)}w_{c}^{(s)2}+\dfrac{1}{p_{1}^{2}}w_{c}^{(s)2}\right).$$
		
		For the row-consistency of $\hat{\Rb}^{(s+1)}$, we have
		$$\begin{aligned}
		\hat{\bR}_{i\cdot}^{(s+1)}-\hat{\Hb}_{r}^{(s+1)\top}\bR_{i\cdot}&=
		\dfrac{1}{\sqrt{p_{1}}p_{2}}\left(\Zb^{(s+1)}\right)^{-1/2}\left(\sum_{t=1}^{T}\hat{\Rb}^{(s)\top}\Eb_{t}\hat{\Cb}^{(s)}\hat{\Cb}^{(s)\top}\Cb\Fb_{t}^{\top}\bR_{i\cdot}\right)\\
		&+\dfrac{1}{\sqrt{p_{1}}p_{2}}\left(\Zb^{(s+1)}\right)^{-1/2}\left(\sum_{t=1}^{T}\hat{\Rb}^{(s)\top}\Rb\Fb_{t}\Cb^{\top}\hat{\Cb}^{(s)}\hat{\Cb}^{(s)\top}\be_{t,i\cdot}\right)\\
		&+\dfrac{1}{\sqrt{p_{1}}p_{2}}\left(\Zb^{(s+1)}\right)^{-1/2}\left(\sum_{t=1}^{T}\hat{\Rb}^{(s)\top}\Eb_{t}\hat{\Cb}^{(s)}\hat{\Cb}^{(s)\top}\be_{t,i\cdot}\right).
		\end{aligned}$$
	
	    First, by Lemma \ref{lemma11} (1) and (2),
	    $$\begin{aligned}
	    \left\|\sum_{t=1}^{T}\hat{\Rb}^{(s)\top}\Eb_{t}\hat{\Cb}^{(s)}\hat{\Cb}^{(s)\top}\Cb\Fb_{t}^{\top}\bR_{i\cdot}\right\|_{F}^{2}&\asymp p_{2}^{2}\left\|\sum_{t=1}^{T}\hat{\Rb}^{(s)\top}\Eb_{t}\hat{\Cb}^{(s)}\Fb_{t}^{\top}\right\|_{F}^{2}\\
	    &\leq p_{2}^{2}\left\|\sum_{t=1}^{T}\Eb_{t}\hat{\Cb}^{(s)}\Fb_{t}^{\top}\right\|_{F}^{2}\left\|\hat{\Rb}^{(s)}-\Rb\hat{\Hb}_{r}^{(s)}\right\|_{F}^{2}+p_{2}^{2}\left\|\sum_{t=1}^{T}\Rb^{\top}\Eb_{t}\hat{\Cb}^{(s)}\Fb_{t}^{\top}\right\|_{F}^{2}\\
	    &=O_{p}\left(Tp_{1}^{2}p_{2}^{4}w_{r}^{(s)}w_{c}^{(s)}+Tp_{1}^{2}p_{2}^{3}w_{r}^{(s)}+Tp_{1}p_{2}^{4}w_{c}^{(s)}+Tp_{1}p_{2}^{3}\right),
	    \end{aligned}$$
	    then, $$\begin{aligned}
	    \left\|\dfrac{1}{\sqrt{p_{1}}p_{2}}\left(\Zb^{(s+1)}\right)^{-1/2}\left(\sum_{t=1}^{T}\hat{\Rb}^{(s)\top}\Eb_{t}\hat{\Cb}^{(s)}\hat{\Cb}^{(s)\top}\Cb\Fb_{t}^{\top}\bR_{i\cdot}\right)\right\|_{F}^{2}& \leq \dfrac{1}{p_{1}p_{2}^{2}}\left\|(\Zb^{(s+1)})^{-1/2}\right\|_{F}^{2}\left\|\sum_{t=1}^{T}\hat{\Rb}^{(s)\top}\Eb_{t}\hat{\Cb}^{(s)}\hat{\Cb}^{(s)\top}\Cb\Fb_{t}^{\top}\bR_{i\cdot}\right\|_{F}^{2}\\
	    &=O_{p}\left(\dfrac{w_{r}^{(s)}w_{c}^{(s)}}{T}+\dfrac{w_{r}^{(s)}}{Tp_{2}}+\dfrac{w_{c}^{(s)}}{Tp_{1}}+\dfrac{1}{Tp_{1}p_{2}}\right).
	    \end{aligned}$$

        Second, similar to the proof of Lemma \ref{lemma11} (1), we have
        $$\left\|\sum_{t=1}^{T}\Fb_{t}\hat{\Cb}^{(s)\top}\be_{t,i\cdot}\right\|_{F}^{2}=O_{p}\left(Tp_{2}^{2}w_{c}^{(s)}+Tp_{2}\right),$$

        $$\begin{aligned}
        \left\|\dfrac{1}{\sqrt{p_{1}}p_{2}}\left(\Zb^{(s+1)}\right)^{-1/2}\left(\sum_{t=1}^{T}\hat{\Rb}^{(s)\top}\Rb\Fb_{t}\Cb^{\top}\hat{\Cb}^{(s)}\hat{\Cb}^{(s)\top}\be_{t,i\cdot}\right)\right\|_{F}^{2}& \asymp p_{1}\left\|\left(\Zb^{(s+1)}\right)^{-1/2}\sum_{t=1}^{T}\Fb_{t}\hat{\Cb}^{(s)\top}\be_{t,i\cdot}\right\|_{F}^{2} \\
        & \leq p_{1}\left\|\left(\Zb^{(s+1)}\right)^{-1/2}\right\|_{F}^{2} \left\|\sum_{t=1}^{T}\Fb_{t}\hat{\Cb}^{(s)\top}\be_{t,i\cdot}\right\|_{F}^{2}\\
        &=O_{p}\left(\dfrac{w_{c}^{(s)}}{T}+\dfrac{1}{Tp_{2}}\right).
        \end{aligned}$$

        Third, similar to the proof of Lemma \ref{lemma9} (1) and Lemma \ref{lemma10} (1), we have
        $$\left\|\sum_{t=1}^{T}\hat{\Cb}^{(s)\top}\be_{t,i\cdot}\right\|_{F}^{2}=O_{p}\left(Tp_{2}^{2}w_{c}^{(s)}+Tp_{2}\right),$$

        $$\left\|\sum_{t=1}^{T}\Rb^{\top}\Eb_{t}\hat{\Cb}^{(s)}\hat{\Cb}^{(s)\top}\be_{t,i\cdot}\right\|_{F}^{2}=O_{p}\left(Tp_{1}p_{2}^{3}+T^{2}p_{2}^{2}+(T^{2}p_{2}^{4}+Tp_{1}p_{2}^{4})w_{c}^{(s)2}\right),$$
        hence,
        $$\begin{aligned}
        &\left\|\sum_{t=1}^{T}\hat{\Rb}^{(s)\top}\Eb_{t}\hat{\Cb}^{(s)}\hat{\Cb}^{(s)\top}\be_{t,i\cdot}\right\|_{F}^{2} \leq \left\|\sum_{t=1}^{T}(\hat{\Rb}^{(s)}-\Rb\hat{\Hb}_{r}^{(s)})^{\top}\Eb_{t}\hat{\Cb}^{(s)}\hat{\Cb}^{(s)\top}\be_{t,i\cdot}\right\|_{F}^{2}+\left\|\sum_{t=1}^{T}\Rb^{\top}\Eb_{t}\hat{\Cb}^{(s)}\hat{\Cb}^{(s)\top}\be_{t,i\cdot}\right\|_{F}^{2}\\
        &\leq \left\|\hat{\Rb}^{(s)}-\Rb\hat{\Hb}_{r}^{(s)}\right\|_{F}^{2}\left\|\sum_{t=1}^{T}\Eb_{t}\hat{\Cb}^{(s)}\right\|_{F}^{2}\left\|\sum_{t=1}^{T}\hat{\Cb}^{(s)\top}\be_{t,i\cdot}\right\|_{F}^{2}+\left\|\sum_{t=1}^{T}\Rb^{\top}\Eb_{t}\hat{\Cb}^{(s)}\hat{\Cb}^{(s)\top}\be_{t,i\cdot}\right\|_{F}^{2}\\
        &=O_{p}\left(Tp_{1}p_{2}^{3}+T^{2}p_{2}^{2}+T^{2}p_{1}^{2}p_{2}^{2}w_{r}^{(s)}+Tp_{1}p_{2}^{4}w_{c}^{(s)2}+T^{2}p_{2}^{4}w_{c}^{(s)2}+T^{2}p_{1}^{2}p_{2}^{4}w_{c}^{(s)2}w_{r}^{s}+T^{2}p_{1}^{2}p_{2}^{3}w_{r}^{(s)}w_{c}^{(s)}\right)
        \end{aligned}$$

        $$\begin{aligned}
        &\left\|\dfrac{1}{\sqrt{p_{1}}p_{2}}\left(\Zb^{(s+1)}\right)^{-1/2}\left(\sum_{t=1}^{T}\hat{\Rb}^{(s)\top}\Eb_{t}\hat{\Cb}^{(s)}\hat{\Cb}^{(s)\top}\be_{t,i\cdot}\right)\right\|_{F}^{2} \leq \dfrac{1}{p_{1}p_{2}^{2}}\left\|\left(\Zb^{(s+1)}\right)^{-1/2}\right\|_{F}^{2}\left\|\sum_{t=1}^{T}\hat{\Rb}^{(s)\top}\Eb_{t}\hat{\Cb}^{(s)}\hat{\Cb}^{(s)\top}\be_{t,i\cdot}\right\|_{F}^{2}\\
        &=O_{p}\left(\dfrac{1}{Tp_{1}p_{2}}+\dfrac{1}{p_{1}^{2}p_{2}^{2}}+\dfrac{w_{r}^{(s)}}{p_{2}^{2}}+\dfrac{w_{c}^{(s)2}}{Tp_{1}}+\dfrac{w_{c}^{(s)2}}{p_{1}^{2}}+w_{r}^{(s)}w_{c}^{(s)2}+\dfrac{w_{r}^{(s)}w_{c}^{(s)}}{p_{2}}\right).
        \end{aligned}$$
        As a result $$\|\hat{\bR}_{i\cdot}^{(s+1)}-\hat{\Hb}_{r}^{(s+1)\top}\bR_{i\cdot}\|_{2}^{2}=O_{p}\left(\dfrac{1}{Tp_{2}}+\dfrac{1}{p_{1}^{2}p_{2}^{2}}+\dfrac{1}{p_{2}^{2}}w_{r}^{(s)}+\dfrac{1}{T}w_{c}^{(s)}+\dfrac{w_{r}^{(s)}w_{c}^{(s)}}{p_{2}}+w_{r}^{(s)}w_{c}^{(s)2}+\dfrac{1}{p_{1}^{2}}w_{c}^{(s)2}\right).$$

		As for the $\hat{\Cb}^{(s+1)}$,
		$$\hat{\Cb}^{(s+1)}=\sqrt{p_{2}}\left(\sum_{t=1}^{T}\Xb_{t}^{\top}\hat{\Rb}^{(s+1)}\hat{\Fb}_{t}^{(s+1)}\right)\left[\left(\sum_{t=1}^{T}\hat{\Fb}_{t}^{(s+1)\top}\hat{\Rb}^{(s+1)\top}\Xb_{t}\right)\left(\sum_{t=1}^{T}\Xb_{t}^{\top}\hat{\Rb}^{(s+1)}\hat{\Fb}_{t}^{(s+1)}\right)\right]^{-1/2}.$$
		$$\begin{aligned}
			\sum_{t=1}^{T}\Xb_{t}^{\top}\hat{\Rb}^{(s+1)}\hat{\Fb}_{t}^{(s+1)}&=\dfrac{1}{p_{1}p_{2}}\sum_{t=1}^{T}\Xb_{t}^{\top}\hat{\Rb}^{(s+1)}\hat{\Rb}^{(s)\top}\Xb_{t}\hat{\Cb}^{(s)}\\
			&=\dfrac{1}{p_{1}p_{2}}\sum_{t=1}^{T}(\Rb\Fb_{t}\Cb^{\top}+\Eb_{t})^{\top}\hat{\Rb}^{(s+1)}\hat{\Rb}^{(s)\top}(\Rb\Fb_{t}\Cb^{\top}+\Eb_{t})\hat{\Cb}^{(s)}\\
			&=\dfrac{1}{p_{1}p_{2}}\sum_{t=1}^{T}\Cb\Fb_{t}^{\top}\Rb^{\top}\hat{\Rb}^{(s+1)}\hat{\Rb}^{(s)\top}\Rb\Fb_{t}\Cb^{\top}\hat{\Cb}^{(s)}+\dfrac{1}{p_{1}p_{2}}\sum_{t=1}^{T}\Eb_{t}^{\top}\hat{\Rb}^{(s+1)}\hat{\Rb}^{(s)\top}\Rb\Fb_{t}\Cb^{\top}\hat{\Cb}^{(s)}\\
			&+\dfrac{1}{p_{1}p_{2}}\sum_{t=1}^{T}\Cb\Fb_{t}^{\top}\Rb^{\top}\hat{\Rb}^{(s+1)}\hat{\Rb}^{(s)\top}\Eb_{t}\hat{\Cb}^{(s)}+\dfrac{1}{p_{1}p_{2}}\sum_{t=1}^{T}\Eb_{t}^{\top}\hat{\Rb}^{(s+1)}\hat{\Rb}^{(s)\top}\Eb_{t}\hat{\Cb}^{(s)}\\
			&=\bDelta_{1}^{(s+1)}+\bDelta_{2}^{(s+1)}+\bDelta_{3}^{(s+1)}+\bDelta_{4}^{(s+1)}.
		\end{aligned}$$
		
		As $$\begin{aligned}
			\|\dfrac{1}{T}\bDelta_{1}^{(s+1)}\|_{F}&=\|\dfrac{1}{Tp_{1}p_{2}}\sum_{t=1}^{T}\Cb\Fb_{t}^{\top}\Rb^{\top}\hat{\Rb}^{(s+1)}\hat{\Rb}^{(s)\top}\Rb\Fb_{t}\Cb^{\top}\hat{\Cb}^{(s)}\|_{F}^{2}\\
			&\asymp p_{1}^{2}\|\dfrac{1}{T}\sum_{t=1}^{T}\Cb\Fb_{t}^{\top}\Fb_{t}\|_{F}^{2}=O_{p}\left(p_{1}^{2}p_{2}\right),
		\end{aligned}$$
	$\|\bDelta_{1}^{(s+1)}\|_{F}^{2} \geq \|\bDelta_{1}^{(s+1)}\|_{2}^{2} \gtrsim T^2p_1^2p_2$, then $\|\bDelta_{1}^{(s+1)}\|_{F}^{2} \asymp O_{p}\left(T^{2}p_{1}^{2}p_{2}\right)$.
		
		By Lemma \ref{lemma11} (3),
		$$
		\|\bDelta_{2}^{(s+1)}\|_{F}^{2}=\|\dfrac{1}{p_{1}p_{2}}\sum_{t=1}^{T}\Eb_{t}^{\top}\hat{\Rb}^{(s+1)}\hat{\Rb}^{(s)\top}\Rb\Fb_{t}\Cb^{\top}\hat{\Cb}^{(s)}\|_{F}^{2}\asymp\|\sum_{t=1}^{T}\Eb_{t}^{\top}\hat{\Rb}^{(s+1)}\Fb_{t}\|_{F}=O_{p}\left(Tp_{1}^{2}p_{2}w_{r}^{(s+1)}+Tp_{1}p_{2}\right).$$
		
		Similarly, by Lemma \ref{lemma11} (3) and (4), we get
		$$\begin{aligned}
			\|\bDelta_{3}^{(s+1)}\|_{F}^{2}&=\|\dfrac{1}{p_{1}p_{2}}\sum_{t=1}^{T}\Cb\Fb_{t}^{\top}\Rb^{\top}\hat{\Rb}^{(s+1)}\hat{\Rb}^{(s)\top}\Eb_{t}\hat{\Cb}^{(s)}\|_{F}^{2}\asymp \dfrac{1}{p_{2}^{2}}\|\sum_{t=1}^{T}\Cb\Fb_{t}^{\top}\hat{\Rb}^{(s)\top}\Eb_{t}\hat{\Cb}^{(s)}\|_{F}^{2}\\
			&\leq \dfrac{1}{p_{2}^{2}}\|\sum_{t=1}^{T}\Cb\Fb_{t}^{\top}\hat{\Rb}^{(s)\top}\Eb_{t}(\hat{\Cb}^{(s)}-\Cb\hat{\Hb}_{c}^{(s)})\|_{F}^{2}+\dfrac{1}{p_{2}^{2}}\|\sum_{t=1}^{T}\Cb\Fb_{t}^{\top}\hat{\Rb}^{(s)\top}\Eb_{t}\Cb\hat{\Hb}_{c}^{(s)}\|_{F}^{2}\\
			&\leq \dfrac{1}{p_{2}^{2}}\|\Cb\|_{F}^{2}\|\sum_{t=1}^{T}\Fb_{t}^{\top}\hat{\Rb}^{(s)\top}\Eb_{t}\|_{F}^{2}\|\hat{\Cb}^{(s)}-\Cb\hat{\Hb}_{c}^{(s)}\|_{F}^{2}+\dfrac{1}{p_{2}^{2}}\|\Cb\|_{F}^{2}\|\sum_{t=1}^{T}\Fb_{t}^{\top}\hat{\Rb}^{(s)\top}\Eb_{t}\Cb\|_{F}^{2}\|\hat{\Hb}_{c}^{(s)}\|_{F}^{2}\\
			&=O_{p}\left(Tp_{1}^{2}p_{2}w_{r}^{(s)}w_{c}^{(s)}+Tp_{1}p_{2}w_{c}^{(s)}+Tp_{1}^{2}w_{r}^{(s)}+Tp_{1}\right).
		\end{aligned}$$
		According to Lemma \ref{lemma9} (2) and Lemma \ref{lemma10} (2), we have
		$$\begin{aligned}
			\|\bDelta_{4}^{(s+1)}\|_{F}^{2}&=\|\dfrac{1}{p_{1}p_{2}}\sum_{t=1}^{T}\Eb_{t}^{\top}\hat{\Rb}^{(s+1)}\hat{\Rb}^{(s)\top}\Eb_{t}\hat{\Cb}^{(s)}\|_{F}^{2}\\
			&\leq \dfrac{1}{p_{1}^{2}p_{2}^{2}}\|\sum_{t=1}^{T}\Eb_{t}^{\top}\hat{\Rb}^{(s+1)}\hat{\Rb}^{(s)\top}\Eb_{t}(\hat{\Cb}^{(s)}-\Cb\hat{\Hb}_{c}^{(s)})\|_{F}^{2}+\dfrac{1}{p_{1}^{2}p_{2}^{2}}\|\sum_{t=1}^{T}\Eb_{t}^{\top}\hat{\Rb}^{(s+1)}\hat{\Rb}^{(s)\top}\Eb_{t}\Cb\hat{\Hb}_{c}^{(s)}\|_{F}^{2}\\
			&\leq \dfrac{1}{p_{1}^{2}p_{2}^{2}}\sum_{t=1}^{T}\|\Eb_{t}^{\top}\hat{\Rb}^{(s+1)}\|_{F}^{2}\sum_{t=1}^{T}\|\Eb_{t}^{\top}\hat{\Rb}^{(s)}\|_{F}^{2}\|\hat{\Cb}^{(s)}-\Cb\hat{\Hb}_{c}^{(s)}\|_{F}^{2}+\dfrac{1}{p_{1}^{2}p_{2}^{2}}\|\sum_{t=1}^{T}\Eb_{t}^{\top}\hat{\Rb}^{(s+1)}\hat{\Rb}^{(s)\top}\Eb_{t}\Cb\|_{F}^{2}\|\hat{\Hb}_{c}^{(s)}\|_{F}^{2}\\
			&=O_{p}\left(Tp_{1}+\dfrac{T^{2}}{p_{2}}+T^{2}p_{2}w_{c}^{(s)}+T^{2}p_{1}p_{2}w_{r}^{(s)}w_{c}^{(s)}+T^{2}p_{1}^{2}p_{2}w_{r}^{(s)}w_{r}^{(s+1)}w_{c}^{(s)}+Tp_{1}^{2}w_{r}^{(s)}w_{r}^{(s+1)}+\dfrac{T^{2}p_{1}^{2}}{p_{2}}w_{r}^{(s)}w_{r}^{(s+1)}\right).
		\end{aligned}$$
		
		Let $\Yb^{(s+1)}=\left(\sum_{t=1}^{T}\hat{\Fb}_{t}^{(s+1)\top}\hat{\Rb}^{(s+1)\top}\Xb_{t}\right)\left(\sum_{t=1}^{T}\Xb_{t}^{\top}\hat{\Rb}^{(s+1)}\hat{\Fb}_{t}^{(s+1)}\right)$, then $\Yb^{(s+1)}=\sum_{i=1}^{4}\sum_{j=1}^{4}\bDelta_{i}^{(s+1)\top}\bDelta_{j}^{(s+1)}$, we can prove $$\left\|\left(\Yb^{(s+1)}\right)^{-1/2}\right\|_{F}=O_{p}\left(\dfrac{1}{T^{2}p_{1}^{2}p_{2}}\right).$$
		
		$$\begin{aligned}
			\hat{\Cb}^{(s+1)}&=\sqrt{p_{2}}\left(\sum_{t=1}^{T}\Xb_{t}^{\top}\hat{\Rb}^{(s+1)}\hat{\Fb}_{t}^{(s+1)}\right)\left[\left(\sum_{t=1}^{T}\hat{\Fb}_{t}^{(s+1)\top}\hat{\Rb}^{(s+1)\top}\Xb_{t}\right)\left(\sum_{t=1}^{T}\Xb_{t}^{\top}\hat{\Rb}^{(s+1)}\hat{\Fb}_{t}^{(s+1)}\right)\right]^{-1/2}\\
			&=\sqrt{p_{2}}\left(\sum_{t=1}^{T}\Xb_{t}^{\top}\hat{\Rb}^{(s+1)}\hat{\Fb}_{t}^{(s+1)}\right)\left(\Yb^{(s+1)}\right)^{-1/2}\\
			&=\Cb\left[\dfrac{1}{p_{1}\sqrt{p_{2}}}\left(\sum_{t=1}^{T}\Fb_{t}^{\top}\Rb^{\top}\hat{\Rb}^{(s+1)}\hat{\Rb}^{(s)\top}\Rb\Fb_{t}\Cb^{\top}\hat{\Cb}^{(s)}\right)\left(\Yb^{(s+1)}\right)^{-1/2}\right]\\
			&+\dfrac{1}{p_{1}\sqrt{p_{2}}}\left(\sum_{t=1}^{T}\Eb_{t}^{\top}\hat{\Rb}^{(s+1)}\hat{\Rb}^{(s)\top}\Rb\Fb_{t}\Cb^{\top}\hat{\Cb}^{(s)}\right)\left(\Yb^{(s+1)}\right)^{-1/2}\\
			&+\dfrac{1}{p_{1}\sqrt{p_{2}}}\left(\sum_{t=1}^{T}\Cb\Fb_{t}^{\top}\Rb^{\top}\hat{\Rb}^{(s+1)}\hat{\Rb}^{(s)\top}\Eb_{t}\hat{\Cb}^{(s)}\right)\left(\Yb^{(s+1)}\right)^{-1/2}\\
			&+\dfrac{1}{p_{1}\sqrt{p_{2}}}\left(\sum_{t=1}^{T}\Eb_{t}^{\top}\hat{\Rb}^{(s+1)}\hat{\Rb}^{(s)\top}\Eb_{t}\hat{\Cb}^{(s)}\right)\left(\Yb^{(s+1)}\right)^{-1/2}.\\
		\end{aligned}$$
		Denote $\hat{\Hb}_{c}^{(s+1)}=\dfrac{1}{p_{1}\sqrt{p_{2}}}\left(\sum_{t=1}^{T}\Fb_{t}^{\top}\Rb^{\top}\hat{\Rb}^{(s+1)}\hat{\Rb}^{(s)\top}\Rb\Fb_{t}\Cb^{\top}\hat{\Cb}^{(s)}\right)\left(\Yb^{(s+1)}\right)^{-1/2}$, then
		$$\begin{aligned}
			\hat{\Cb}^{(s+1)}-\Cb\hat{\Hb}_{c}^{(s+1)}&=\dfrac{1}{p_{1}\sqrt{p_{2}}}\left(\sum_{t=1}^{T}\Eb_{t}^{\top}\hat{\Rb}^{(s+1)}\hat{\Rb}^{(s)\top}\Rb\Fb_{t}\Cb^{\top}\hat{\Cb}^{(s)}\right)\left(\Yb^{(s+1)}\right)^{-1/2}\\
			&+\dfrac{1}{p_{1}\sqrt{p_{2}}}\left(\sum_{t=1}^{T}\Cb\Fb_{t}^{\top}\Rb^{\top}\hat{\Rb}^{(s+1)}\hat{\Rb}^{(s)\top}\Eb_{t}\hat{\Cb}^{(s)}\right)\left(\Yb^{(s+1)}\right)^{-1/2}\\
			&+\dfrac{1}{p_{1}\sqrt{p_{2}}}\left(\sum_{t=1}^{T}\Eb_{t}^{\top}\hat{\Rb}^{(s+1)}\hat{\Rb}^{(s)\top}\Eb_{t}\hat{\Cb}^{(s)}\right)\left(\Yb^{(s+1)}\right)^{-1/2}\\
			&=\sqrt{p_{2}}(\bDelta_{2}^{(s+1)}+\bDelta_{3}^{(s+1)}+\bDelta_{4}^{(s+1)})(\Yb^{(s+1)})^{-1/2}.
		\end{aligned}$$
		thus
		$$
		\dfrac{1}{p_{2}}\|\hat{\Cb}^{(s+1)}-\Cb\hat{\Hb}_{c}^{(s+1)}\|_{F}^{2}=O_{p}\left(\dfrac{1}{Tp_{1}}+\dfrac{1}{p_{1}^{2}p_{2}^{2}}+\gamma_{c}^{(s+1)}\right),$$
		where $\gamma_{c}^{(s+1)}=\dfrac{1}{Tp_{2}}w_{r}^{(s)}+\dfrac{1}{T}w_{r}^{(s+1)}+\dfrac{1}{p_{1}^{2}}w_{c}^{(s)}+\dfrac{1}{p_{1}}w_{r}^{(s)}w_{c}^{(s)}+\dfrac{1}{T}w_{r}^{(s)}w_{c}^{(s)}+w_{r}^{(s)}w_{r}^{(s+1)}w_{c}^{(s)}+\dfrac{1}{p_{2}^{2}}w_{r}^{(s)}w_{r}^{(s+1)}$.
		
		Similar to the proof the row-consistency of $\hat{\Rb}^{(s+1)}$, for any $j\in[p_{2}]$, we can obtain  $$\|\hat{\bC}_{j\cdot}^{(s+1)}-\hat{\Hb}_{c}^{(s+1)\top}\bC_{j\cdot}\|_{2}^{2}=O_{p}\left(\dfrac{1}{Tp_{1}}+\dfrac{1}{p_{1}^{2}p_{2}^{2}}+\gamma_{c}^{(s+1)}\right),$$
		where $\gamma_{c}^{(s+1)}=\dfrac{1}{Tp_{2}}w_{r}^{(s)}+\dfrac{1}{T}w_{r}^{(s+1)}+\dfrac{1}{p_{1}^{2}}w_{c}^{(s)}+\dfrac{1}{p_{1}}w_{r}^{(s)}w_{c}^{(s)}+\dfrac{1}{T}w_{r}^{(s)}w_{c}^{(s)}+w_{r}^{(s)}w_{r}^{(s+1)}w_{c}^{(s)}+\dfrac{1}{p_{2}^{2}}w_{r}^{(s)}w_{r}^{(s+1)}$.
		
	\end{proof}

\subsection{Proof of Theorem \ref{the iteration results2}}

    \begin{proof}
    By definition,
    $$\begin{aligned}
   	 \hat{\Fb}_{t}^{(s+1)}&=\dfrac{1}{p_{1}p_{2}}\hat{\Rb}^{(s)\top}\Xb_{t}\hat{\Cb}^{(s)}=\dfrac{1}{p_{1}p_{2}}\hat{\Rb}^{(s)\top}\Rb\Fb_{t}\Cb^{\top}\hat{\Cb}^{(s)}+\dfrac{1}{p_{1}p_{2}}\hat{\Rb}^{(s)\top}\Eb_{t}\hat{\Cb}^{(s)}\\
   	 &=\dfrac{1}{p_{1}p_{2}}\hat{\Rb}^{(s)\top}\left(\Rb-\hat{\Rb}^{(s)}(\hat{\Hb}_{r}^{(s)})^{-1}+\hat{\Rb}^{(s)}(\hat{\Hb}_{r}^{(s)})^{-1}\right)\Fb_{t}\left(\Cb-\hat{\Cb}^{(s)}(\hat{\Hb}_{c}^{(s)})^{-1}+\hat{\Cb}^{(s)}(\hat{\Hb}_{c}^{(s)})^{-1}\right)^{\top}\hat{\Cb}^{(s)}\\
   	 &+\dfrac{1}{p_{1}p_{2}}\left(\hat{\Rb}^{(s)}-\Rb\hat{\Hb}_{r}^{(s)}+\Rb\hat{\Hb}_{r}^{(s)}\right)^{\top}\Eb_{t}\left(\hat{\Cb}^{(s)}-\Cb\hat{\Hb}_{c}^{(s)}+\Cb\hat{\Hb}_{c}^{(s)}\right).
    \end{aligned}$$

    Note that $\hat{\Rb}^{(s)\top}\hat{\Rb}^{(s)}=p_{1}\Ib_{k_{1}}$, $\hat{\Cb}^{(s)\top}\hat{\Cb}^{(s)}=p_{2}\Ib_{k_{2}}$, then
    $$\begin{aligned}
   	 \hat{\Fb}_{t}^{(s+1)}-(\hat{\Hb}_{r}^{(s)})^{-1}\Fb_{t}\left((\hat{\Hb}_{c}^{(s)})^{-1}\right)^{\top}&=\dfrac{1}{p_{1}p_{2}}\hat{\Rb}^{(s)\top}\left(\Rb-\hat{\Rb}^{(s)}(\hat{\Hb}_{r}^{(s)})^{-1}\right)\Fb_{t}\left(\Cb-\hat{\Cb}(\hat{\Hb}_{c}^{(s)})^{-1}\right)^{\top}\hat{\Cb}^{(s)}\\
   	 &+\dfrac{1}{p_{1}}\hat{\Rb}^{(s)\top}\left(\Rb-\hat{\Rb}^{(s)}(\hat{\Hb}_{r}^{(s)})^{-1}\right)\Fb_{t}\left((\hat{\Hb}_{c}^{(s)})^{-1}\right)^{\top}\\
   	 &+\dfrac{1}{p_{2}}(\hat{\Hb}_{r}^{(s)})^{-1}\Fb_{t}\left(\Cb-\hat{\Cb}(\hat{\Hb}_{c}^{(s)})^{-1}\right)^{\top}\hat{\Cb}^{(s)}\\
   	 &+\dfrac{1}{p_{1}p_{2}}\left(\hat{\Rb}^{(s)}-\Rb\hat{\Hb}_{r}^{(s)}\right)^{\top}\Eb_{t}\left(\hat{\Cb}^{(s)}-\Cb\hat{\Hb}_{c}^{(s)}\right)\\
   	 &+\dfrac{1}{p_{1}p_{2}}\left(\hat{\Rb}^{(s)}-\Rb\hat{\Hb}_{r}^{(s)}\right)^{\top}\Eb_{t}\Cb\hat{\Hb}_{c}^{(s)}+\dfrac{1}{p_{1}p_{2}}\hat{\Hb}_{r}^{(s)\top}\Rb^{\top}\Eb_{t}\left(\hat{\Cb}^{(s)}-\Cb\hat{\Hb}_{c}^{(s)}\right)\\
   	 &+\dfrac{1}{p_{1}p_{2}}\hat{\Hb}_{r}^{(s)\top}\Rb^{\top}\Eb_{t}\Cb\hat{\Hb}_{c}^{(s)}.
    \end{aligned}$$

    Therefore, by Cauchy-Schwartz inequality, Lemma \ref{loading product of (s+1)-th} and Theorem \ref{the iteration results}, we have
    $$\left\|\hat{\Fb}_{t}^{(s+1)}-(\hat{\Hb}_{r}^{(s)})^{-1}\Fb_{t}\left((\hat{\Hb}_{c}^{(s)})^{-1}\right)^{\top}\right\|_{F}=O_{p}\left(\dfrac{\sqrt{w_{r}^{(s)}}}{\sqrt{p_{2}}}+\dfrac{\sqrt{w_{r}^{(s-1)}}}{\sqrt{Tp_{2}}}+\dfrac{\sqrt{w_{c}^{(s)}}}{\sqrt{p_{1}}}+\dfrac{\sqrt{w_{c}^{(s-1)}}}{\sqrt{Tp_{1}}}+\dfrac{1}{\sqrt{p_{1}p_{2}}}+\gamma_{f}^{(s+1)}\right),$$
    where $\gamma_{f}^{(s+1)}=\sqrt{w_{r}^{(s)}w_{c}^{(s)}}+\dfrac{\sqrt{w_{r}^{(s-1)}w_{c}^{(s-1)}}}{\sqrt{T}}+\dfrac{\sqrt{w_{r}^{(s-1)}}w_{c}^{(s-1)}}{p_{1}}$.

    Next, for any $t,i,j$
    $$\begin{aligned}
    	\hat{S}_{t,ij}^{(s+1)}-S_{t,ij}&=\hat{\bR}_{i\cdot}^{(s)\top}\hat{\Fb}_{t}^{(s+1)}\hat{\bC}_{j\cdot}^{(s)}-\bR_{i\cdot}^{\top}\Fb_{t}\bC_{j\cdot}\\
    	&=(\hat{\bR}_{i\cdot}^{(s)}-\hat{\Hb}_{r}^{(s)\top}\bR_{i\cdot}+\hat{\Hb}_{r}^{(s)\top}\bR_{i\cdot})^{\top}\hat{\Fb}_{t}^{(s+1)}(\hat{\bC}_{j\cdot}^{(s)}-\hat{\Hb}_{c}^{(s)\top}\bC_{j\cdot}+\hat{\Hb}_{c}^{(s)\top}\bC_{j\cdot})-\bR_{i\cdot}^{\top}\Fb_{t}\bC_{j\cdot}\\
    	&=(\hat{\bR}_{i\cdot}^{(s)}-\hat{\Hb}_{r}^{(s)\top}\bR_{i\cdot})^{\top}\hat{\Fb}_{t}^{(s+1)}(\hat{\bC}_{j\cdot}^{(s)}-\hat{\Hb}_{c}^{(s)\top}\bC_{j\cdot})+\bR_{i\cdot}^{\top}\hat{\Hb}_{r}^{(s)}\hat{\Fb}_{t}^{(s+1)}(\hat{\bC}_{j\cdot}^{(s)}-\hat{\Hb}_{c}^{(s)\top}\bC_{j\cdot})\\
    	&+(\hat{\bR}_{i\cdot}^{(s)}-\hat{\Hb}_{r}^{(s)\top}\bR_{i\cdot})^{\top}\hat{\Fb}_{t}^{(s+1)}\hat{\Hb}_{c}^{(s)\top}\bC_{j\cdot}+\bR_{i\cdot}^{\top}(\hat{\Hb}_{r}^{(s)}\hat{\Fb}_{t}^{(s+1)}\hat{\Hb}_{c}^{(s)\top}-\Fb_{t})\bC_{j\cdot}.
    \end{aligned}$$
    Then, by Cauchy-Schwartz inequality, Theorem \ref{no_iter_convergence}, Theorem \ref{convergence of an iteration} and the consistency of the estimators for the factor matrices, we have
    $$|\hat{S}_{t,ij}^{(s+1)}-S_{t,ij}|=O_{p}\left(\sqrt{w_{r}^{(s)}}+\dfrac{\sqrt{w_{r}^{(s-1)}}}{\sqrt{Tp_{2}}}+\sqrt{w_{c}^{(s)}}+\dfrac{\sqrt{w_{c}^{(s-1)}}}{\sqrt{Tp_{1}}}+\dfrac{1}{\sqrt{p_{1}p_{2}}}+\gamma^{(s+1)}\right),$$
    where $\gamma^{(s+1)}=\dfrac{\sqrt{w_{r}^{(s-1)}w_{c}^{(s-1)}}}{\sqrt{T}}+\dfrac{\sqrt{w_{r}^{(s-1)}}w_{c}^{(s-1)}}{p_{1}}$.

    \end{proof}

\subsection{Proof of Theorem \ref{ER}}
\begin{proof}
	Similar to the proof of Theorem \ref{Convergence}, for any bounded $m_{1} \geq k_{1}, m_{2} \geq k_{2}$, we have $$\|\tilde{\Fb}_{t}-\tilde{\Hb}_{1}\Fb_{t}\tilde{\Hb}_{2}^{\top}\|_{F}=\|\tilde{\bm{\cE}}_{t}\|_{F}=O_{p}\left(\dfrac{1}{\sqrt{p_{1}}}\right),$$
	where $\tilde{\Hb}_{1}=\hat{\Rb}^{(S-1)\top}\Rb/p_1,    \tilde{\Hb}_{2}=\hat{\Cb}^{(S-1)\top}\Cb/p_2, \tilde{\bm{\cE}}_{t}=\hat{\Rb}^{(S-1)\top}\Eb_{t}\hat{\Cb}^{(S-1)}/(p_1p_2)$, $S$ is the number of iterations until convergence or reach a prefixed maximum number of iteration. thus $$\left\|\dfrac{1}{T}\sum_{t=1}^{T}\tilde{\Fb}_{t}\tilde{\Fb}_{t}^{\top}-\dfrac{1}{T}\sum_{t=1}^{T}\tilde{\Hb}_{1}\Fb_{t}\tilde{\Hb}_{2}^{\top}\tilde{\Hb}_{2}\Fb_{t}^{\top}\tilde{\Hb}_{1}^{\top}\right\|_{F}=O_{p}\left(\dfrac{1}{\sqrt{p_{1}}}\right).$$
	
	We first show that $\text{rank}(\tilde{\Hb}_{1})=k_1, \text{rank}(\tilde{\Hb}_{2})=k_2$.
	
	Denote $\hat{\Rb}^{(s)}_{k_1}$ as the first $k_1$ columns of $\hat{\Rb}^{(s)}$, $\hat{\Cb}^{(s)}_{k_2}$ is defined similarly. Then $$\hat{\Rb}^{(s)}_{k_1}=\sqrt{p_1}\left(\sum_{t=1}^{T}\Xb_{t}\hat{\Cb}^{(s-1)}\hat{\Fb}_{t}^{(s)\top}\right)\Zb^{(s)}_{k_1},$$
	where $\Zb_{k_1}^{(s)}$ is a matrix which is consisted of the first $k_1$ columns of $$\left(\Zb^{(s)}\right)^{-1/2}=\left[\left(\sum_{t=1}^{T}\hat{\Fb}_{t}^{(s)}\hat{\Cb}^{(s-1)\top}\Xb_{t}^{\top}\right)\left(\sum_{t=1}^{T}\Xb_{t}\hat{\Cb}^{(s-1)}\hat{\Fb}_{t}^{(s)\top}\right)\right]^{-1/2},$$
	hence, $\|\Zb_{k_1}^{(s)}\|_{2} \leq 1/\lambda_{\min}\left((\Zb^{(s)})^{1/2}\right)$.
	
	Assume that $k_1=k_2=1, m_1 > k_1, m_2 > k_2$, take $s=1$ for example. Similar to the proof of Theorem \ref{no_iter_convergence}, we have
	$$\begin{aligned}
	\nu_{\min}(\bdelta_{1})&=\nu_{\min}\left(\dfrac{\sqrt{p_1}p_2}{T}\sum_{t=1}^{T}\Fb_{t}\Hb_{2}^{\top}\Hb_{2}\Fb_{t}^{\top}\Hb_{1}\right) \gtrsim \sqrt{p_1}p_2\lambda_{\min}\left(\dfrac{1}{T}\sum_{t=1}^{T}\Fb_{t}\Hb_{2}^{\top}\Hb_{2}\Fb_{t}^{\top}\right)\nu_{\min}(\Hb_{1})\\
	&\gtrsim \sqrt{p_1}p_2\lambda_{\min}\left(\dfrac{1}{T}\sum_{t=1}^{T}\Fb_{t}\Fb_{t}^{\top}\right)\nu_{\min}(\Hb_{2}^{\top}\Hb_{2})\nu_{\min}(\Hb_{1})\gtrsim \sqrt{p_1}p_2\nu_{\min}^{2}(\Hb_{2})\nu_{\min}(\Hb_{1}).
	\end{aligned}$$
$\|\bdelta_{2}\|_{F}=O_{p}\left(\sqrt{p_2/T}\nu_{\min}(\Hb_{2})\right), \|\bdelta_{3}\|_{F}=O_{p}\left(\sqrt{p_1p_2/T}\nu_{\min}(\Hb_{1})\nu_{\min}(\Hb_{2})\right), \|\bdelta_{4}\|_{F}=O_{p}\left(\sqrt{p_2/T}+1/\sqrt{p_1}\right)$, thus $$\lambda_{\min}\left((\Zb^{(s)})^{1/2}\right) \gtrsim \sqrt{p_1}p_2\nu_{\min}(\Hb_{1})\nu_{\min}^{2}(\Hb_{2}).$$

Let $$\hat{\Hb}_{r,k_1}^{(1)}=\left(\dfrac{1}{T\sqrt{p_1}p_2}\sum_{t=1}^{T}\Fb_{t}\Cb^{\top}\Wb_{2}\Wb_{2}^{\top}\Cb\Fb_{t}^{\top}\Rb^{\top}\Wb_{1}\right)\Zb_{k_1}^{(1)},$$
hence, we can prove $$\dfrac{1}{p_1}\|\hat{\Rb}_{k_1}^{{(1)}}-\Rb\hat{\Hb}_{r,k_1}^{(1)}\|_{F}^{2}=o_{p}(1).$$
Finally, we have $\Ib_{k_{1}}=p_1^{-1}\hat{\Rb}_{k_1}^{(1)\top}\Rb\hat{\Hb}_{r,k_1}^{(1)}+o_{p}(1)=\hat{\Hb}_{r,k_1}^{(1)\top}\hat{\Hb}_{r,k_1}^{(1)}+o_{p}(1)$. That is to say,  $k_1=\text{rank}(p_1^{-1}\hat{\Rb}_{k_1}^{(1)\top}\Rb) \leq \text{rank}(p_1^{-1}\hat{\Rb}^{(1)\top}\Rb) \leq k_1$. Thus, $\text{rank}(p_1^{-1}\hat{\Rb}^{(1)\top}\Rb)=k_1$. Similarly, $\text{rank}(p_2^{-1}\hat{\Cb}^{(1)\top}\Cb)=k_2$. Follow a similar proof, we claim that $\text{rank}(\tilde{\Hb}_{1})=k_1, \text{rank}(\tilde{\Hb}_{2})=k_2$.
	
Due to $\text{rank}(\tilde{\Hb}_{1})=k_1, \text{rank}(\tilde{\Hb}_{2})=k_2$, hence, $\text{rank}(\tilde{\Hb}_{1}^{\top}\tilde{\Hb}_{1})=k_1, \text{rank}(\tilde{\Hb}_{2}^{\top}\tilde{\Hb}_{2})=k_2$, $$\text{rank}\left(\dfrac{1}{T}\sum_{t=1}^{T}\tilde{\Hb}_1\Fb_{t}\tilde{\Hb}_{2}^{\top}\tilde{\Hb}_{2}\Fb_{t}^{\top}\tilde{\Hb}_{1}^{\top}\right)=\text{rank}\left(\dfrac{1}{T}\sum_{t=1}^{T}\Fb_{t}\Fb_{t}^{\top}\right)=k_1,$$
thus, $T^{-1}\sum_{t=1}^{T}\tilde{\Hb}_{1}\Fb_{t}\tilde{\Hb}_{2}^{\top}\tilde{\Hb}_{2}\Fb_{t}^{\top}\tilde{\Hb}_{1}^{\top}$ has exactly $m_1-k_1$ multiple roots of $\lambda=0$.

For $i=1,\dots, k_1$, by Weyl theorem, we have
$$\begin{aligned}
	\left|\lambda_{i}\left(\dfrac{1}{T}\sum_{t=1}^{T}\tilde{\Fb}_{t}\tilde{\Fb}_{t}\right)-\lambda_{i}\left(\dfrac{1}{T}\sum_{t=1}^{T}\tilde{\Hb}_{1}\Fb_{t}\tilde{\Hb}_{2}^{\top}\tilde{\Hb}_{2}\Fb_{t}^{\top}\tilde{\Hb}_{1}^{\top}\right)\right|& \leq \left\|\dfrac{1}{T}\sum_{t=1}^{T}\tilde{\Fb}_{t}\tilde{\Fb}_{t}^{\top}-\dfrac{1}{T}\sum_{t=1}^{T}\tilde{\Hb}_{1}\Fb_{t}\tilde{\Hb}_{2}^{\top}\tilde{\Hb}_{2}\Fb_{t}^{\top}\tilde{\Hb}_{1}^{\top}\right\|_{F}\\
	&=O_{p}\left(\dfrac{1}{\sqrt{p_{1}}}\right).
\end{aligned}$$
	For $i > k_{1}$, we have $$\left|\lambda_{i}\left(\dfrac{1}{T}\sum_{t=1}^{T}\tilde{\Fb}_{t}\tilde{\Fb}_{t}\right)\right|=O_{p}\left(\dfrac{1}{\sqrt{p_1}}\right).$$
	
	As a result, $$
		\lambda_{i}\left(\dfrac{1}{T}\sum_{t=1}^{T}\tilde{\Fb}_{t}\tilde{\Fb}_{t}^{\top}\right)=\begin{cases}
			\lambda_{i}\left(\dfrac{1}{T}\sum_{t=1}^{T}\tilde{\Hb}_{1}\Fb_{t}\tilde{\Hb}_{2}^{\top}\tilde{\Hb}_{2}\Fb_{t}^{\top}\tilde{\Hb}_{1}^{\top}\right)+O_{p}\left(\dfrac{1}{\sqrt{p_1}}\right), \quad i\leq k_1, \\
			O_{p}\left(\dfrac{1}{\sqrt{p_{1}}}\right), \quad i > k_1.
	\end{cases}$$
Note that $\|T^{-1}\sum_{t=1}^{T}\tilde{\Hb}_{1}\Fb_{t}\tilde{\Hb}_{2}^{\top}\tilde{\Hb}_{2}\Fb_{t}^{\top}\tilde{\Hb}_{1}^{\top}\|_{F}=O_{p}(1)$ and $\text{rank}(T^{-1}\sum_{t=1}^{T}\tilde{\Hb}_{1}\Fb_{t}\tilde{\Hb}_{2}^{\top}\tilde{\Hb}_{2}\Fb_{t}^{\top}\tilde{\Hb}_{1}^{\top})=k_1$, as $T,p_1,p_2 \to \infty$, we can calculate the eigenvalue ratios.

$$\max_{i \leq k_{1}-1}\dfrac{\lambda_{i}\left(\dfrac{1}{T}\sum_{t=1}^{T}\tilde{\Fb}_{t}\tilde{\Fb}_{t}^{\top}\right)}{\lambda_{i+1}\left(\dfrac{1}{T}\sum_{t=1}^{T}\tilde{\Fb}_{t}\tilde{\Fb}_{t}^{\top}\right)}=O_{p}(1),$$
$$\max_{i \geq k_{1}+1}\dfrac{\lambda_{i}\left(\dfrac{1}{T}\sum_{t=1}^{T}\tilde{\Fb}_{t}\tilde{\Fb}_{t}^{\top}\right)}{\lambda_{i+1}\left(\dfrac{1}{T}\sum_{t=1}^{T}\tilde{\Fb}_{t}\tilde{\Fb}_{t}^{\top}\right)} \leq O_{p}(1),$$
$$\dfrac{\lambda_{k_1}\left(\dfrac{1}{T}\sum_{t=1}^{T}\tilde{\Fb}_{t}\tilde{\Fb}_{t}^{\top}\right)}{\lambda_{k_1+1}\left(\dfrac{1}{T}\sum_{t=1}^{T}\tilde{\Fb}_{t}\tilde{\Fb}_{t}^{\top}\right)} \to \infty ,$$
which concludes the consistency.

Similarly, we can get $\PP(\hat{k}_2=k_2) \to 1$, as $T,p_1,p_2 \to \infty$.

\end{proof}

\section{Propositions and Lemmas}

In this section, we first give some propositions and lemmas which are essential for the proofs of main theorems.
	\begin{proposition}\label{Pro}
	Suppose $m_{1} \geq k_{1}, m_{2} \geq k_{2}$ and $T, p_{1}, p_{2} \to \infty$.Under Assumptions \ref{Assumption 1}-\ref{Assumption 3}, we have
		
	\begin{itemize}
   \item  $ \lambda_{\min}(\dfrac{1}{T}\hat{\Fb}^{\top}\hat{\Fb}) \geq c(p_{1}p_{2})^{-1}$ with probability approaching one for some $c >0$.

   \end{itemize}

	\end{proposition}
	
    \begin{proof}
    It is easy to get that
	$$\hat{\Fb}=\Fb\Hb^{\top}+\bm{\cE},$$ where $\bm{\cE}=\Eb\Wb/({p_{1}p_{2}}), \Wb=\Wb_{2} \otimes \Wb_{1}.$
	Let
	$$\Delta:=\dfrac{1}{T}\Hb\Fb^{\top}\bm{\cE}+\dfrac{1}{T}\bm{\cE}^{\top}\Fb\Hb^{\top}+\dfrac{1}{T}\EE\bm{\cE}^{\top}\bm{\cE}+\dfrac{1}{T}(\bm{\cE}^{\top}\bm{\cE}-\EE\bm{\cE}^{\top}\bm{\cE}).$$
	Then we have
	$$\dfrac{1}{T}\hat{\Fb}^{\top}\hat{\Fb}=\dfrac{1}{T}\Hb\Fb^{\top}\Fb\Hb^{\top}+\Delta.$$
	By the assumption $$\lambda_{\min}\left(\dfrac{1}{T}\EE\Eb^{\top}\Eb \right)=\lambda_{\min}\left(\dfrac{1}{T}\sum_{t=1}^{T}\EE\text{Vec}(\Eb_{t})\text{Vec}(\Eb_{t})^{\top}\right) \geq c_{0},$$ and the  property of the Kronecker product, we can get  $$\lambda_{\min}\left(\dfrac{1}{T}\EE\bm{\cE}^{\top}\bm{\cE}\right)
	\geq \lambda_{\min}\left(\dfrac{1}{T}\EE\Eb^{\top}\Eb\right)\lambda_{\min}\left(\dfrac{1}{p_{1}^{2}p_{2}^2}\Wb^{\top}\Wb\right) \geq c_{0}(p_{1}p_{2})^{-1}$$ for some $c_{0} >0$. In addition, in the following Lemma \ref{lemma1} (4), we show that $$\|\dfrac{1}{T}(\bm{\cE}^{\top}\bm{\cE}-\EE\bm{\cE}^{\top}\bm{\cE})\|_{2}=O_{p}\left(\dfrac{1}{p_{1}p_{2}\sqrt{T}}\right).$$ Therefore, the inequality $\|\dfrac{1}{T}(\bm{\cE}^{\top}\bm{\cE}-\EE\bm{\cE}^{\top}\bm{\cE})\|_{2} \leq \dfrac{1}{2}\lambda_{\min}(\dfrac{1}{T}\EE\bm{\cE}^{\top}\bm{\cE})$ hold with large probability. We now continue the argument conditioning on this event.
		
	Let $\bv$ be the unit vector such that $\bv\top \dfrac{1}{T}\hat{\Fb}^{\top}\hat{\Fb}\bv=\lambda_{\min}(\dfrac{1}{T}\hat{\Fb}^{\top}\hat{\Fb})$, as $\bv\top \dfrac{1}{T}\hat{\Fb}^{\top}\hat{\Fb}\bv=\bv\top\dfrac{1}{T}\Hb\Fb^{\top}\Fb\Hb^{\top}\bv+\bv \top \Delta\bv$, then we have $$\lambda_{\min}(\dfrac{1}{T}\hat{\Fb}^{\top}\hat{\Fb}) \geq \dfrac{1}{T}\bv\top\Hb\Fb^{\top}\Fb\Hb^{\top}\bv+\dfrac{2}{T}\bv\top\Hb\Fb^{\top}\bm{\cE}\bv+\dfrac{c_{0}}{2p_{1}p_{2}}.$$
	If $\bv\top\Hb=0$, then $\lambda_{\min}(\dfrac{1}{T}\hat{\Fb}^{\top}\hat{\Fb}) \geq \dfrac{c_{0}}{2p_{1}p_{2}}$. If $\bv\top \Hb \neq 0$, we have $\dfrac{1}{T}\bv\top\Hb\Fb^{\top}\Fb\Hb^{\top}\bv>0$ with large probability because $\dfrac{1}{T}\Fb^{\top}\Fb$ is positive definite. Let
	$$\alpha_{\bv}^{2}=\dfrac{1}{T}\bv\top\Hb\Fb^{\top}\Fb\Hb^{\top}\bv, \ \ X=(\dfrac{\alpha_{\bv}^{2}}{Tp_{1}p_{2}})^{-1/2}2\bv\top\dfrac{1}{T}\Hb\Fb^{\top}\bm{\cE}\bv, \ \ 2\bv\top\dfrac{1}{T}\Hb\Fb^{\top}\bm{\cE}\bv=X\sqrt{\dfrac{\alpha_{\bv}^{2}}{Tp_{1}p_{2}}}.$$
	Then $$\lambda_{\min}(\dfrac{1}{T}\hat{\Fb}^{\top}\hat{\Fb}) \geq \alpha_{\bv}^{2}+X\sqrt{\dfrac{\alpha_{\bv}^{2}}{Tp_{1}p_{2}}}+\dfrac{c_{0}}{2p_{1}p_{2}}.$$
		
	First, we prove that $X=O_{p}(1)$. By the Assumption \ref{Assumption 3}, it holds that $\lambda_{\min}(\dfrac{1}{T}\Fb^{\top}\Fb) >c>0 $, then we have
	$$\alpha_{\bv}^{2} \geq \lambda_{\min}(\dfrac{1}{T}\Fb^{\top}\Fb)\bv\top\Hb\Hb^{\top}\bv > c\|\bv\top \Hb\|_{2}^{2}.$$
	By  Lemma \ref{lemma1} (3) in the following, we have $\|\dfrac{1}{T}\Fb^{\top}\bm{\cE}\|_{2}^{2}=O_{p}\left(\dfrac{1}{Tp_{1}p_{2}}\right)$ and as a result,
	$$|X|^{2} \leq 4Tp_{1}p_{2}\alpha_{\bv}^{-2}\|\bv\top\Hb\|_{2}^{2}\|\dfrac{1}{T}\Fb^{\top}\bm{\cE}\|_{2}^{2} \leq O_{p}(1)\alpha_{\bv}^{-2}\|\bv\top\Hb\|_{2}^{2} \leq O_{p}(1)c^{-1}\|\bv\top\Hb\|_{2}^{-2}\|\bv\top\Hb\|_{2}^{2}=O_{p}\left(1\right).$$
		
	Then we consider two cases. \vspace{1em}
		
	\textbf{Case one}: $\alpha_{\bv}^{2} \leq 4|X|\sqrt{\dfrac{\alpha_{\bv}^{2}}{Tp_{1}p_{2}}}.$ Then $|\alpha_{\bv}| \leq 4|X|\dfrac{1}{\sqrt{Tp_{1}p_{2}}}$ and
	$$\lambda_{\min}(\dfrac{1}{T}\hat{\Fb}^{\top}\hat{\Fb}) \geq \dfrac{c_{0}}{2p_{1}p_{2}}-|X||\alpha_{\bv}|\dfrac{1}{\sqrt{Tp_{1}p_{2}}} \geq \dfrac{c_{0}}{2p_{1}p_{2}}-4|X|^{2}\dfrac{1}{Tp_{1}p_{2}} \geq \dfrac{c_{0}}{4p_{1}p_{2}},$$
where the last inequality holds  with probability  approaching to 1 by the fact that $X=O_{p}(1)$ and $T \to \infty $.

\vspace{1em}

	\textbf{Case two}: $\alpha_{\bv}^{2} >4|X|\sqrt{\dfrac{\alpha_{\bv}^{2}}{Tp_{1}p_{2}}},$ then
	$$\lambda_{\min}(\dfrac{1}{T}\hat{\Fb}^{\top}\hat{\Fb})\geq \alpha_{\bv}^{2}-|X|\sqrt{\dfrac{\alpha_{\bv}^{2}}{Tp_{1}p_{2}}}+\dfrac{c_{0}}{2p_{1}p_{2}} \geq \dfrac{3}{4}\alpha_{\bv}^{2}+\dfrac{c_{0}}{2p_{1}p_{2}} \geq \dfrac{c_{0}}{2p_{1}p_{2}}.$$
		
	From what has been discussed above, we conclude that  with probability approaching one, we have
$$\lambda_{\min}(\dfrac{1}{T}\hat{\Fb}^{\top}\hat{\Fb}) > c_{0}/(p_{1}p_{2}).$$
		
	\end{proof}
	
	\begin{lemma}\label{lemma1}
	For any $m_{1} \geq 1$ and $m_{2} \geq 1$, (note that $m_{1},m_{2}$ can be either smaller, equal to or larger than $k_{1},k_{2}$),
		
	(1) $ \|\EE(\bm{\cE}_{t}^{\top}\bm{\cE}_{t})\|_{2}=O\left(\dfrac{1}{p_{1}}\right) \, \text{and} \, \|\bm{\cE}_{t}\|_{2}=O_{p}\left(\dfrac{1}{\sqrt{p_{1}}}\right), t \in [T],$
		
	(2) $\|\bm{\cE}\|_{2}=O_{p}\left(\sqrt{\dfrac{T}{p_{1}p_{2}}}\right),$
		
	(3) $\EE\|\dfrac{1}{T}\Fb^{\top}\bm{\cE}\|_{2}^{2} \leq O\left(\dfrac{1}{Tp_{1}p_{2}}\right),$
		
	(4) $\left\|\dfrac{1}{T}(\bm{\cE}^{\top}\bm{\cE}-\EE\bm{\cE}^{\top}\bm{\cE})\right\|_{2} \leq O_{p}\left(\dfrac{1}{p_{1}p_{2}\sqrt{T}}\right).$
		
	\end{lemma}
	
	In this lemma, all $\EE(\cdot), \EE(\cdot|\cdot)$ and $\text{\text{Cov}}(\cdot)$ are calculated conditionally on $\Wb_{1}, \Wb_{2}$. The proof of the lemma is given below.
	\begin{proof}
	(1) By Assumption \ref{Assumption 3}, we have $\|\EE(\Eb_{t}\Eb_{t}^{\top})\|_{2} \leq cp_{2}$, thus,
	$$
	\begin{aligned}
	&\|\EE(\Eb_{t}^{\top}\Wb_{1}\Wb_{1}^{\top}\Eb_{t})\|_{2} \leq \text{tr}(\EE\Eb_{t}^{\top}\Wb_{1}\Wb_{1}^{\top}\Eb_{t})\leq \EE(\text{tr}(\Eb_{t}^{\top}\Wb_{1}\Wb_{1}^{\top}\Eb_{t}))\\
	&=\text{tr}(\Wb_{1}^{\top}\EE(\Eb_{t}\Eb_{t}^{\top})\Wb_{1}) \leq m_{1}\|\Wb_{1}\|_{2}^{2}\|\EE(\Eb_{t}\Eb_{t}^{\top})\|_{2} \leq cm_{1}p_{1}p_{2},
	\end{aligned}$$
where the penultimate inequality is derived as follows: let $\bv_{i}$ be the $i$-th eigenvector of $\Wb_{1}^{\top}\EE(\Eb_{t}\Eb_{t}^{\top})\Wb_{1}$, then
	$$\text{tr}(\Wb_{1}^{\top}\EE(\Eb_{t}\Eb_{t}^{\top})\Wb_{1})=\sum_{i=1}^{m_{1}}\bv_{i}^{\top}\Wb_{1}^{\top}\EE(\Eb_{t}\Eb_{t}^{\top})\Wb_{1}\bv_{i} \leq \|\EE(\Eb_{t}\Eb_{t}^{\top})\|_{2}\sum_{i=1}^{m_{1}}\|\Wb_{1}\bv_{i}\|_{2}^{2} \leq m_{1}\|\EE(\Eb_{t}\Eb_{t}^{\top})\|_{2}\|\Wb_{1}\|_{2}^{2}.$$
		
	$$
\begin{aligned}
\|\EE(\bm{\cE}_{t}^{\top}\bm{\cE}_{t})\|_{2}&=\dfrac{1}{p_{1}^{2}p_{2}^{2}}\|\EE(\Wb_{2}^{\top}\Eb_{t}^{\top}\Wb_{1}\Wb_{1}^{\top}\Eb_{t}\Wb_{2})\|_{2}=\dfrac{1}{p_{1}^{2}p_{2}^{2}}\|\Wb_{2}^{\top}\EE(\Eb_{t}^{\top}\Wb_{1}\Wb_{1}^{\top}\Eb_{t})\Wb_{2}\|_{2}\\
& \leq \dfrac{1}{p_{1}^{2}p_{2}^{2}}\|\Wb_{2}\|_{2}^{2}\|\EE(\Eb_{t}^{\top}\Wb_{1}\Wb_{1}^{\top}\Eb_{t})\|_{2},
\end{aligned}$$
	which imply $$\|\EE(\bm{\cE}_{t}^{\top}\bm{\cE}_{t})\|_{2} \leq \dfrac{c}{p_{1}}.$$	
	Thus, we have $\EE\|\bm{\cE}_{t}\|_{2}^{2} \leq \text{tr}(\EE\bm{\cE}_{t}^{\top}\bm{\cE}_{t}) \leq m_{2}\|\EE\bm{\cE}_{t}^{\top}\bm{\cE}_{t}\|_{2} \leq \dfrac{cm_{1}m_{2}}{p_{1}}.$

		\vspace{2em}

	(2) It holds that $\bm{\cE}=\dfrac{1}{p_{1}p_{2}}\Eb\Wb, \Wb=\Wb_{2} \otimes \Wb_{1}.$ By the Assumption \ref{Assumption 3},  $$\begin{aligned}
\Big\|\dfrac{1}{T}\EE(\Eb^{\top}\Eb)\Big\|_{2}&=\Big\|\dfrac{1}{T}\sum_{t=1}^{T}\EE\text{Vec}(\Eb_{t})\text{Vec}(\Eb_{t})^{\top}\Big\|_{2}\\
&=\|\EE\text{Vec}(\Eb_{t})\text{Vec}(\Eb_{t})^{\top}\|_{2} \\
&\leq \EE\|\EE\left(\text{Vec}(\Eb_{t})\text{Vec}(\Eb_{t})^{\top}|\Fb\right)\|_{2} \leq c.
\end{aligned}$$ Thus
	$$\EE\Big\|\dfrac{1}{p_{1}p_{2}}\Eb\Wb\Big\|_{2}^{2} \leq \dfrac{1}{p_{1}^{2}p_{2}^{2}}\text{tr}\EE(\Wb^{\top}\Eb^{\top}\Eb\Wb) \leq \dfrac{1}{p_{1}^{2}p_{2}^{2}}\|\Wb\|_{F}^{2}\|\EE\Eb^{\top}\Eb\|_{2} \leq \dfrac{cT}{p_{1}p_{2}}.$$
	
	\vspace{2em}
	
	(3) By the Assumption \ref{Assumption 3}, we have
	$$\dfrac{1}{T}\sum_{s=1}^{T}\sum_{T=1}^{T}\EE\|\Fb_{t}\|_{F}\|\Fb_{s}\|_{F}\|\EE(\text{Vec}(\Eb_{s})\text{Vec}(\Eb_{t})^{\top}|\Fb)\|_{2}<c, $$
    then we can obtain
	$$
	\begin{aligned}
	&\EE\|\dfrac{1}{T}\Fb^{\top}\bm{\cE}\|_{2}^{2}=\dfrac{1}{T^{2}p_{1}^{2}p_{2}^{2}}\EE\|\Fb^{\top}\Eb\Wb\|_{2}^{2}=\dfrac{1}{T^{2}p_{1}^{2}p_{2}^{2}}\EE\|\sum_{t=1}^{T}\Wb^{\top}\text{Vec}(\Eb_{t})\text{Vec}(\Fb_{t})^{\top}\|_{2}^{2}\\
	&=\dfrac{1}{T^{2}p_{1}^{2}p_{2}^{2}}\EE\left[\EE\left(\|\sum_{t=1}^{T}\Wb^{\top}\text{Vec}(\Eb_{t})\text{Vec}(\Fb_{t})^{\top}\|_{2}^{2}\right)\Bigg|\Fb\right]\\
	&\leq \dfrac{1}{T^{2}p_{1}^{2}p_{2}^{2}}\EE\left[\EE\left(\text{tr}(\sum_{t=1}^{T}\sum_{s=1}^{T}\text{Vec}(\Fb_{t})\text{Vec}(\Eb_{t})^{\top}\Wb\Wb^{\top}\text{Vec}(\Eb_{s})\text{Vec}(\Fb_{s})^{\top})\right)\Bigg|\Fb\right]\\
	&=\dfrac{1}{T^{2}p_{1}^{2}p_{2}^{2}}\EE\left\{\text{tr}\left[\sum_{t=1}^{T}\sum_{s=1}^{T}\text{Vec}(\Fb_{t})\EE\left(\text{Vec}(\Eb_{t})^{\top}\Wb\Wb^{\top}\text{Vec}(\Eb_{s})\Big|\Fb\right)\text{Vec}(\Fb_{s})^{\top}\right]\right\}\\
	&=\dfrac{1}{T^{2}p_{1}^{2}p_{2}^{2}}\EE\left\{\text{tr}\left[\sum_{t=1}^{T}\sum_{s=1}^{T}\text{Vec}(\Fb_{t})\text{Vec}(\Fb_{s})^{\top}\EE\left(\text{Vec}(\Eb_{t})^{\top}\Wb\Wb^{\top}\text{Vec}(\Eb_{s})\Big|\Fb\right)\right]\right\}\\
	&\leq\dfrac{1}{T^{2}p_{1}^{2}p_{2}^{2}}\sum_{i=1}^{k_{1}}\sum_{j=1}^{k_{2}}\sum_{t=1}^{T}\sum_{s=1}^{T}\EE\left[ f_{t,ij}f_{s,ij}\EE\left(\text{Vec}(\Eb_{t})^{\top}\Wb\Wb^{\top}\text{Vec}(\Eb_{s})\Big|\Fb\right)\right]\\
	&=\dfrac{1}{T^{2}p_{1}^{2}p_{2}^{2}}\sum_{i=1}^{k_{1}}\sum_{j=1}^{k_{2}}\sum_{t=1}^{T}\sum_{s=1}^{T}\EE\left\{ f_{t,ij}f_{s,ij}\text{tr}\left[\Wb^{\top}\EE\left(\text{Vec}(\Eb_{s})\text{Vec}(\Eb_{t})^{\top}\Big|\Fb\right)\Wb\right]\right\}\\
	&\leq\dfrac{1}{T^{2}p_{1}^{2}p_{2}^{2}}\sum_{i=1}^{k_{1}}\sum_{j=1}^{k_{2}}\sum_{t=1}^{T}\sum_{s=1}^{T}\EE\left[| f_{t,ij}f_{s,ij}|\|\Wb\|_{F}^{2}\left\|\EE\left(\text{Vec}(\Eb_{s})\text{Vec}(\Eb_{t})^{\top}\Big|\Fb\right)\right\|_{2}\right]\\
	&\leq\dfrac{c}{T^{2}p_{1}p_{2}}\sum_{t=1}^{T}\sum_{s=1}^{T}\EE\|\Fb_{t}\|_{F}\|\Fb_{s}\|_{F}\left\|\EE\left(\text{Vec}(\Eb_{s})\text{Vec}(\Eb_{t})^{\top}\Big|\Fb\right)\right\|_{2} \leq \dfrac{c}{Tp_{1}p_{2}}.
	\end{aligned}$$
	Thus, $\EE\|\dfrac{1}{T}\Fb^{\top}\bm{\cE}\|_{2}^{2} \leq O\left(\dfrac{1}{Tp_{1}p_{2}}\right).$
		
	(4) By the Assumption \ref{Assumption 3} (4), we have
	$$
	\begin{aligned}
	&\EE\|\dfrac{1}{T}(\bm{\cE}^{\top}\bm{\cE}-\EE\bm{\cE}^{\top}\bm{\cE})\|_{2}^{2}\leq \EE\|\dfrac{1}{T}(\bm{\cE}^{\top}\bm{\cE}-\EE\bm{\cE}^{\top}\bm{\cE})\|_{F}^{2}\\
	&\leq \sum_{k,q=1}^{m_{1}m_{2}}\EE\left(\dfrac{1}{Tp_{1}^{2}p_{2}^{2}}\sum_{t=1}^{T}\sum_{i,j=1}^{p_{1}p_{2}}w_{ik}w_{jq}(e_{ti}e_{tj}-\EE e_{ti}e_{tj})\right)^{2}\\
	&= \sum_{k,q=1}^{m_{1}m_{2}}\EE\left(\dfrac{1}{Tp_{1}^{2}p_{2}^{2}}\sum_{t=1}^{T}\sum_{i,j=1}^{p_{1}p_{2}}w_{ik}w_{jq}\left(\text{Vec}(\Eb_{t})_{i}\text{Vec}(\Eb_{t})_{j}-\EE\text{Vec}(\Eb_{t})_{i}\text{Vec}(\Eb_{t})_{j}\right)\right)^{2}\\
	&\leq \dfrac{1}{Tp_{1}^{2}p_{2}^{2}}\sum_{k,q=1}^{m_{1}m_{2}}\dfrac{1}{Tp_{1}^{2}p_{2}^{2}}\sum_{t,s=1}^{T}\sum_{i,j,u,v=1}^{p_{1}p_{2}}\big|w_{ik}w_{jq}w_{uk}w_{vq}\big|\big|\text{Cov}\left(\text{Vec}(\Eb_{t})_{i}\text{Vec}(\Eb_{t})_{j},\text{Vec}(\Eb_{s})_{u}\text{Vec}(\Eb_{s})_{v}\right)\big|\\
	&\leq \dfrac{c}{Tp_{1}^{2}p_{2}^{2}}\dfrac{1}{T p_{1}^2 p_{2}^{2}}\sum_{t,s=1}^{T}\sum_{i_{1},j_{1},u_{1},v_{1}=1}^{p_{1}}\sum_{i_{2},j_{2},u_{2},v_{2}=1}^{p_{2}}|\text{Cov}(e_{t,i_{1}i_{2}}e_{t,j_{1}j_{2}},e_{s,u_{1}u_{2}}e_{s,v_{1}v_{2}})| \leq \dfrac{c}{Tp_{1}^{2}p_{2}^{2}}.
	\end{aligned}$$
		
	\end{proof}

     \begin{lemma}\label{basic-lemma}
    	Under Assumption \ref{Assumption 1} (1), Assumption \ref{Assumption 5}-\ref{Assumption 7}, as $\min\{T,p_{1},p_{2}\} \to \infty$, we have
    	
    	(1) $\sum_{t=1}^{T}\EE\|\Eb_{t}^{\top}\Rb\|_{F}^{2}=O\left(Tp_{1}p_{2}\right), \sum_{t=1}^{T}\EE\|\Eb_{t}\Cb\|_{F}^{2}=O\left(Tp_{1}p_{2}\right),$ $\EE\|\sum_{t=1}^{T}\Eb_{t}^{\top}\Fb_{t}\|_{F}^{2}=O\left(Tp_{1}p_{2}\right)$,
    	
    \vspace{1em}

    	(2) $\EE\|\sum_{t=1}^{T}\Fb_{t}\Cb^{\top}\Eb_{t}^{\top}\|_{F}^{2}=O\left(Tp_{1}p_{2}\right), \EE\|\sum_{t=1}^{T}\Fb_{t}^{\top}\Rb^{\top}\Eb_{t}\|_{F}^{2}=O\left(Tp_{1}p_{2}\right)$,
    	
            \vspace{1em}

    	$\EE\|\sum_{t=1}^{T}\Fb_{t}\Cb^{\top}\Eb_{t}^{\top}\Rb\|_{F}^{2}=O\left(Tp_{1}p_{2}\right), \EE\|\sum_{t=1}^{T}\Fb_{t}^{\top}\Rb^{\top}\Eb_{t}\Cb\|_{F}^{2}=O\left(Tp_{1}p_{2}\right)$,
    	
        \vspace{1em}

    	(3) $	\EE\|\sum_{t=1}^{T}\Fb_{t}\Wb_{2}^{\top}\Eb_{t}^{\top}\Wb_{1}\|_{F}^{2}=O\left(Tp_{1}p_{2}\right), \EE\|\sum_{t=1}^{T}\Eb_{t}^{\top}\Wb_{1}\Fb_{t}\|_{F}^{2}=O\left(Tp_{1}p_{2}\right)$,
    	
        \vspace{1em}

    	$\EE\|\sum_{t=1}^{T}\Fb_{t}\Wb_{2}^{\top}\Eb_{t}^{\top}\|_{F}^{2}=O\left(Tp_{1}p_{2}\right), \EE\|\sum_{t=1}^{T}\Wb_{1}^{\top}\Eb_{t}\Wb_{2}\|_{F}^{2}=O\left(Tp_{1}p_{2}\right)$,
    	
        \vspace{1em}

    	(4) $\EE\|\sum_{t=1}^{T}\Rb^{\top}\Eb_{t}\Wb_{2}\Fb_{t}^{\top}\|_{F}^{2}=O\left(Tp_{1}p_{2}\right), \EE\|\sum_{t=1}^{T}\Rb^{\top}\Eb_{t}\Wb_{2}\|_{F}^{2}=O\left(Tp_{1}p_{2}\right)$
    	
        \vspace{1em}

    	(5)
    	$\EE\|\sum_{t=1}^{T}\Eb_{t}\Wb_{2}\Wb_{2}^{\top}\Eb_{t}^{\top}\Wb_{1}\|_{F}^{2}=O\left(Tp_{1}^{2}p_{2}^{3}+T^{2}p_{1}p_{2}^{2}\right)$.
    \end{lemma}

    \begin{proof}
    	Assume $m_{1}=k_{1}=1, m_{2}=k_{2}=1$.
    	
    	(1) By Assumption \ref{Assumption 5}, we can get $$\EE\|\Eb_{t}^{\top}\Rb\|_{F}^{2}=\sum_{j=1}^{p_{2}}\EE(\sum_{i=1}^{p_{1}}e_{t,ij}r_{i})^{2}=\sum_{j=1}^{p_{2}}\sum_{i_{1},i_{2}=1}^{p_{1}}\EE(r_{i_{1}}r_{i_{2}}e_{t,i_{1}j}e_{t,i_{2}j})=O\left(p_{1}p_{2}\right).$$
    	Similarly, $\sum_{t=1}^{T}\EE\|\Eb_{t}\Cb\|_{F}^{2}=O\left(Tp_{1}p_{2}\right)$.
    	By Assumption \ref{Assumption 7}(2),
    	$$\EE\|\sum_{t=1}^{T}\Eb_{t}^{\top}\Fb_{t}\|_{F}^{2}=\sum_{i=1}^{p_{1}}\sum_{j=1}^{p_{2}}\EE(\sum_{t=1}^{T}\Fb_{t}e_{t,ij})^{2}=T\sum_{i=1}^{p_{1}}\sum_{j=1}^{p_{2}}\EE(\bxi_{i,j})^{2}=O\left(Tp_{1}p_{2}\right).$$
    	
    	(2) The results hold directly by Assumption \ref{Assumption 7} (1).
    	
    	(3) By Assumption \ref{Assumption 1} (1) and Assumption \ref{Assumption 7} (2), we have
    	
    	$$\begin{aligned}
    		\EE\|\sum_{t=1}^{T}\Fb_{t}\Wb_{2}^{\top}\Eb_{t}^{\top}\Wb_{1}\|_{F}^{2}&=\EE(\sum_{t=1}^{T}\sum_{i=1}^{p_{1}}\sum_{j=1}^{p_{2}}\Fb_{t}w_{1,i}e_{t,ij}w_{2,j})^{2}=T\EE(\sum_{i=1}^{p_{1}}\sum_{j=1}^{p_{2}}\bxi_{i,j}w_{1,i}w_{2,j})^{2}\\
    		&=T\sum_{i_{1},i_{2}=1}^{p_{1}}\sum_{j_{1},j_{2}=1}^{p_{2}}\EE(\bxi_{i_{1},j_{1}}\bxi_{i_{2},j_{2}})w_{1,i_{1}}w_{1,i_{2}}w_{2,j_{1}}w_{2,j_{2}}=O\left(Tp_{1}p_{2}\right).
    	\end{aligned}$$
    	
    	$$\begin{aligned}
    		\EE\|\sum_{t=1}^{T}\Eb_{t}^{\top}\Wb_{1}\Fb_{t}\|_{F}^{2}&=\sum_{j=1}^{p_{2}}\EE(\sum_{t=1}^{T}\sum_{i=1}^{p_{1}}e_{t,ij}w_{1,j}\Fb_{t})^{2}=T\sum_{j=1}^{p_{2}}\EE(\sum_{i=1}^{p_{1}}\bxi_{i,j}w_{1,j})^{2}\\
    		&=T\sum_{j=1}^{p_{2}}\sum_{i,i^{\prime}=1}^{p_{1}}\EE(\bxi_{i,j}\bxi_{i^{\prime},j})w_{1,i}w_{1,i^{\prime}}=O\left(Tp_{1}p_{2}\right).
    	\end{aligned}$$
    	
    	Analogously,  $\EE\|\sum_{t=1}^{T}\Fb_{t}\Wb_{2}^{\top}\Eb_{t}^{\top}\|_{F}^{2}=O(Tp_{1}p_{2})$. By Assumption \ref{Assumption 5} (2), we can also obtain
    	
    	$$
    	\begin{aligned}
    		\EE\|\sum_{t=1}^{T}\Wb_{1}^{\top}\Eb_{t}\Wb_{2}\|_{F}^{2}=\EE(\sum_{t=1}^{T}\sum_{i=1}^{p_{1}}\sum_{j=1}^{p_{2}}w_{1,i}w_{2,j}e_{t,ij})^{2}&=\sum_{s,t=1}^{T}\sum_{i_{1},i_{2}=1}^{p_{1}}\sum_{j_{1},j_{2}=1}^{p_{2}}w_{1,i_{1}}w_{1,i_{2}}w_{2,j_{1}}w_{2,j_{2}}\EE(e_{t,i_{1}j_{1}}e_{s,i_{2}j_{2}})^{2}\\
    		&=O\left(Tp_{1}p_{2}\right).
    	\end{aligned}$$
    	
    	(4) $$\begin{aligned}
    		\EE\|\sum_{t=1}^{T}\Rb^{\top}\Eb_{t}\Wb_{2}\Fb_{t}\|_{F}^{2}&=\EE(\sum_{t=1}^{T}\sum_{i=1}^{p_{1}}\sum_{j=1}^{p_{2}}r_{i}w_{2,j}e_{t,ij}\Fb_{t})^{2}=T\EE(\sum_{i=1}^{p_{1}}\sum_{j=1}^{p_{2}}r_{i}w_{2,j}\bxi_{i,j})^{2}\\
    		&=T\sum_{i_{1},i_{2}=1}^{p_{1}}\sum_{j_{1},j_{2}=1}^{p_{2}}r_{i_{1}}r_{i_{2}}w_{2,j_{1}}w_{2,j_{2}}\EE(\bxi_{i_{1},j_{1}}\bxi_{i_{2},j_{2}})=O\left(Tp_{1}p_{2}\right).
    	\end{aligned}$$
    	
    	$$\begin{aligned}
    		\EE\|\sum_{t=1}^{T}\Rb^{\top}\Eb_{t}\Wb_{2}\|_{F}^{2}=\EE(\sum_{t=1}^{T}\sum_{i=1}^{p_{1}}\sum_{j=1}^{p_{2}}r_{i}w_{2,j}e_{t,ij})^{2}&=\sum_{s,t=1}^{T}\sum_{i_{1},i_{2}=1}^{p_{1}}\sum_{j_{1},j_{2}=1}^{p_{2}}r_{i_{1}}r_{i_{2}}w_{2,j_{1}}w_{2,j_{2}}\EE(e_{t,i_{1}j_{1}}e_{s,i_{2}j_{2}})^{2}\\
    		&=O\left(Tp_{1}p_{2}\right).
    	\end{aligned}$$
    	
    	(5) Note that $$\EE\|\sum_{t=1}^{T}\Eb_{t}\Wb_{2}\Wb_{2}^{\top}\Eb_{t}^{\top}\Wb_{1}\|_{F}^{2}\leq \|\Wb_{2}\|_{F}^{2}\|\sum_{t=1}^{T}\Wb_{1}^{\top}\Eb_{t}\Wb_{2}\Eb_{t}\|_{F}^{2}\leq \|\Wb_{2}\|_{F}^{2}\sum_{i=1}^{p_{1}}\sum_{j=1}^{p_{2}}\EE\|\sum_{t=1}^{T}\Wb_{1}^{\top}\Eb_{t}\Wb_{2}e_{t,ij}\|^{2},$$
    	while for any $i,j$, by Assumption \ref{Assumption 5},
    	$$\begin{aligned}
    		&\EE\|\sum_{t=1}^{T}\bW_{1}^{\top}\Eb_{t}\Wb_{2}e_{t,ij}\|^{2}=\EE(\sum_{t=1}^{T}\sum_{i^{\prime}=1}^{p_{1}}\sum_{j^{\prime}=1}^{p_{2}}w_{1,i^{\prime}}e_{t,i^{\prime}j^{\prime}}w_{2,j^{\prime}}e_{t,ij})^{2}\\
    		&\leq  c\sum_{t,s=1}^{T}\sum_{i_{1},i_{2}=1}^{p_{1}}\sum_{j_{1},j_{2}=1}^{p_{2}}|\text{Cov}(e_{t,ij}e_{t,i_{1}j_{1}},e_{s,ij}e_{s,i_{2}j_{2}})|+ c(\sum_{t=1}^{T}\sum_{i^{\prime}=1}^{p_{1}}\sum_{j^{\prime}=1}^{p_{2}}|\EE e_{t,i^{\prime}j^{\prime}}e_{t,ij}|)^{2}\\
    		&=O\left(Tp_{1}p_{2}+T^{2}\right).
    	\end{aligned}$$
    	
    \end{proof}

    \begin{lemma}\label{need-in-proof}
    	Under Assumption \ref{Assumption 1} (1), Assumption \ref{Assumption 2} and Assumption \ref{Assumption 4}-\ref{Assumption 7}, as $\min\{T, p_{1}, p_{2}\} \to \infty$, we have
    	$$\|\sum_{s=1}^{T}\Eb_{s}^{\top}(\hat{\Rb}^{(1)}-\Rb\hat{\Hb}_{r}^{(1)})\Fb_{s}\|_{F}^{2}=O_{p}\left(\dfrac{p_{1}^{2}}{p_{2}\nu_{\min}^{2}(\Hb_{2})}+\dfrac{p_{1}}{\nu_{\min}^{2}(\Hb_{1})\nu_{\min}^{4}(\Hb_{2})}+\dfrac{T}{p_{2}\nu_{\min}^{2}(\Hb_{1})\nu_{\min}^{4}(\Hb_{2})}\right).$$
    \end{lemma}

    \begin{proof}
    	By the proof in Theorem \ref{no_iter_convergence}, $\hat{\Rb}^{(1)}-\Rb\hat{\Hb}_{r}^{(1)}=\cI\cI+\cI\cI\cI+\cI\cV$, then $$\sum_{s=1}^{T}\Eb_{s}^{\top}(\hat{\Rb}^{(1)}-\Rb\hat{\Hb}_{r}^{(1)})\Fb_{s}=\sum_{s=1}^{T}\Eb_{s}^{\top}(\cI\cI+\cI\cI\cI+\cI\cV)\Fb_{s}.$$
    	For the first term, by Lemma\ref{basic-lemma} (2) and (3), we get
    	$$\begin{aligned}
    		\|\sum_{t=1}^{T}\Eb_{s}^{\top}\cI\cI\Fb_{s}\|_{F}^{2}&=\dfrac{1}{T^{2}p_{1}p_{2}^{2}}\|\sum_{s=1}^{T}\Eb_{s}^{\top}\Rb[(\sum_{t=1}^{T}\Fb_{t}\Cb^{\top}\Wb_{2}\Wb_{2}^{\top}\Eb_{t}^{\top}\Wb_{1})\Zb^{-1/2}]\Fb_{s}\|_{F}^{2}\\
    		& \leq \dfrac{1}{T^{2}p_{1}p_{2}^{2}}\|\sum_{s=1}^{T}\Eb_{s}^{\top}\Rb\Fb_{s}\|_{F}^{2}\|(\sum_{t=1}^{T}\Fb_{t}\Cb^{\top}\Wb_{2}\Wb_{2}^{\top}\Eb_{t}^{\top}\Wb_{1})\Zb^{-1/2}\|_{F}^{2}\\
    		& \leq \dfrac{1}{T^{2}p_{1}p_{2}^{2}}\|\sum_{s=1}^{T}\Eb_{s}^{\top}\Rb\Fb_{s}\|_{F}^{2}\|\sum_{t=1}^{T}\Fb_{t}\Wb_{2}^{\top}\Eb_{t}^{\top}\Wb_{1}\|_{F}^{2}\|\Zb^{-1/2}\|_{F}^{2}\|\Cb^{\top}\Wb_{2}\|_{2}^{2}\\
    		&=O_{p}\left(\dfrac{1}{\nu_{\min}^{2}(\Hb_{1})\nu_{\min}^{2}(\Hb_{2})}\right).
    	\end{aligned}$$
    	For the second term,
    	$$\begin{aligned}
    		\|\sum_{s=1}^{T}\Eb_{s}^{\top}\cI\cI\cI\Fb_{s}\|_{F}^{2}&=\dfrac{1}{T^{2}p_{1}p_{2}^{2}}\|\sum_{s=1}^{T}\Eb_{s}^{\top}[(\sum_{t=1}^{T}\Eb_{t}\Wb_{2}\Wb_{2}^{\top}\Cb\Fb_{t}^{\top}\Rb^{\top}\Wb_{1})\Zb^{-1/2}]\Fb_{s}\|_{F}^{2}\\
    		& \leq \dfrac{1}{T^{2}p_{1}p_{2}^{2}}\|\sum_{s=1}^{T}\Eb_{s}^{\top}(\sum_{t=1}^{T}\Eb_{t}\Wb_{2}\Fb_{t}^{\top})\Fb_{s}\|_{F}^{2}\|\Wb_{2}^{\top}\Cb\|_{F}^{2}\|\Rb^{\top}\Wb_{1}\|_{F}^{2}\|\Zb^{-1/2}\|_{F}^{2}\\
    		&=O_{p}\left(\dfrac{p_{1}^{2}}{p_{2}\nu_{\min}^{2}(\Hb_{2})}+\dfrac{p_{1}}{\nu_{\min}^{2}(\Hb_{2})}\right),
    	\end{aligned}$$
    	where $$\begin{aligned}
    		\EE\|\sum_{s=1}^{T}\Eb_{s}^{\top}(\sum_{t=1}^{T}\Eb_{t}\Wb_{2}\Fb_{t}^{\top})\Fb_{s}\|_{F}^{2}&=\sum_{j=1}^{p_{2}}\EE(\sum_{s,t=1}^{T}\sum_{i=1}^{p_{1}}\sum_{j_{1}=1}^{p_{2}}w_{2,j_{1}}\Fb_{s}\Fb_{t}e_{s,ij}e_{t,ij_{1}})^{2}= T^{2}\sum_{j=1}^{p_{2}}\EE(\sum_{i=1}^{p_{1}}\sum_{j_{1}=1}^{p_{2}}w_{2,j_{1}}\bxi_{i,j}\bxi_{i,j_{1}})^{2}\\
    		&=T^{2}\sum_{j=1}^{p_{2}}\left[(\sum_{i=1}^{p_{1}}\sum_{j_{1}=1}^{p_{2}}\EE\bxi_{i,j}\bxi_{i,j_{1}})^{2}+\sum_{i_{1},i_{2}=1}^{p_{1}}\sum_{j_{1}j_{2}=1}^{p_{2}}\text{Cov}(\bxi_{i_{1},j}\bxi_{i_{1},j_{1}},\bxi_{i_{2},j}\bxi_{i_{2},j_{2}})\right]\\
    		&=O\left(T^{2}p_{1}^{2}p_{2}+T^{2}p_{1}p_{2}^{2}\right).
    	\end{aligned}$$
    	For the third term, by lemma \ref{basic-lemma} (1) and (5), we have
    	$$\begin{aligned}
    		\|\sum_{s=1}^{T}\Eb_{s}^{\top}\cI\cV\Fb_{s}\|_{F}^{2}&=\dfrac{1}{T^{2}p_{1}p_{2}^{2}}\|\sum_{s=1}^{T}\Eb_{s}[(\sum_{t=1}^{T}\Eb_{t}\Wb_{2}\Wb_{2}^{\top}\Eb_{t}^{\top}\Wb_{1})\Zb^{-1/2}]\Fb_{s}\|_{F}^{2}\\
    		& \leq \dfrac{1}{T^{2}p_{1}p_{2}^{2}}\|\sum_{s=1}^{T}\Eb_{s}^{\top}\Fb_{s}\|_{F}^{2}\|(\sum_{t=1}^{T}\Eb_{t}\Wb_{2}\Wb_{2}^{\top}\Eb_{t}^{\top}\Wb_{1})\Zb^{-1/2}\|_{F}^{2}\\
    		&=O_{p}\left(\dfrac{p_{1}}{\nu_{\min}^{2}(\Hb_{1})\nu_{\min}^{4}(\Hb_{2})}+\dfrac{T}{p_{2}\nu_{\min}^{2}(\Hb_{1})\nu_{\min}^{4}(\Hb_{2})}\right).
    	\end{aligned}$$
    	As a result, $$\|\sum_{s=1}^{T}\Eb_{s}^{\top}(\hat{\Rb}^{(1)}-\Rb\hat{\Hb}_{r}^{(1)})\Fb_{s}\|_{F}^{2}=O_{p}\left(\dfrac{p_{1}^{2}}{p_{2}\nu_{\min}^{2}(\Hb_{2})}+\dfrac{p_{1}}{\nu_{\min}^{2}(\Hb_{1})\nu_{\min}^{4}(\Hb_{2})}+\dfrac{T}{p_{2}\nu_{\min}^{2}(\Hb_{1})\nu_{\min}^{4}(\Hb_{2})}\right).$$
    	
    \end{proof}

    \begin{lemma}\label{Product of loadings}
    	Under Assumption \ref{Assumption 1} (1), Assumption \ref{Assumption 2} and Assumption \ref{Assumption 4}-\ref{Assumption 7}, take $\hat{\Rb}^{(1)}, \hat{\Cb}^{(1)}$ as the resultant estimators from the one-step iteration, as $\min\{T, p_{1}, p_{2}\} \to \infty$, then we have
    	$$\left\|\Rb^{\top}(\hat{\Rb}^{(1)}-\Rb\hat{\Hb}_{r}^{(1)})\right\|_{F}=O_{p}\left(\dfrac{\sqrt{p_{1}}}{\sqrt{Tp_{2}}\nu_{\min}(\Hb_{1})\nu_{\min}(\Hb_{2})}+\dfrac{1}{p_{2}\nu_{\min}(\Hb_{1})\nu_{\min}^{2}(\Hb_{2})}\right),$$
    	$$\left\|\Cb^{\top}(\hat{\Cb}^{(1)}-\Cb\hat{\Hb}_{c}^{(1)})\right\|_{F}=O_{p}\left(\dfrac{1}{T\nu_{\min}(\Hb_{2})}+\dfrac{\sqrt{p_{2}}}{\sqrt{Tp_{1}}\nu_{\min}(\Hb_{1})\nu_{\min}(\Hb_{2})}+\dfrac{1}{p_{1}\nu_{\min}(\Hb_{1})\nu_{\min}(\Hb_{2})}\right).$$
    \end{lemma}

    \begin{proof}
    	$$\begin{aligned}
    		\left\|\Rb^{\top}\cI\cI\right\|_{F}&=\dfrac{1}{T\sqrt{p_{1}}p_{2}}\left\|\Rb^{\top}\Rb(\sum_{t=1}^{T}\Fb_{t}\Cb^{\top}\Wb_{2}\Wb_{2}^{\top}\Eb_{t}^{\top}\Wb_{1})\Zb^{-1/2}\right\|_{F}=\dfrac{\sqrt{p_{1}}}{Tp_{2}}\left\|(\sum_{t=1}^{T}\Fb_{t}\Cb^{\top}\Wb_{2}\Wb_{2}^{\top}\Eb_{t}^{\top}\Wb_{1})\Zb^{-1/2}\right\|_{F}\\
    		&\leq \dfrac{\sqrt{p_{1}}}{Tp_{2}}\left\|\sum_{t=1}^{T}\Fb_{t}\Wb_{2}^{\top}\Eb_{t}^{\top}\Wb_{1}\right\|_{F}\left\|\Zb^{-1/2}\right\|_{F}\left\|\Cb^{\top}\Wb_{2}\right\|_{F}=O_{p}\left(\dfrac{\sqrt{p_{1}}}{\sqrt{Tp_{2}}\nu_{\min}(\Hb_{1})\nu_{\min}(\Hb_{2})}\right),
    	\end{aligned}$$
    	
    	$$\begin{aligned}
    		\left\|\Rb^{\top}\cI\cI\cI\right\|_{F}&=\dfrac{1}{T\sqrt{p_{1}}p_{2}}\left\|(\sum_{t=1}^{T}\Rb^{\top}\Eb_{t}\Wb_{2}\Wb_{2}^{\top}\Cb\Fb_{t}^{\top}\Rb^{\top}\Wb_{1})\Zb^{-1/2}\right\|_{F}\\
    		& \leq \dfrac{1}{T\sqrt{p_{1}}p_{2}}\left\|\sum_{t=1}^{T}\Rb^{\top}\Eb_{t}\Wb_{2}\Fb_{t}^{\top}\right\|_{F}\left\|\Wb_{2}^{\top}\Cb\right\|_{F}\left\|\Rb^{\top}\Wb_{1}\right\|_{F}\left\|\Zb^{-1/2}\right\|_{F}\\
    		&=O_{p}\left(\dfrac{\sqrt{p_{1}}}{\sqrt{Tp_{2}}\nu_{\min}(\Hb_{2})}\right),
    	\end{aligned}$$
    	
    	$$\begin{aligned}
    		\left\|\Rb^{\top}\cI\cV\right\|_{F}&=\dfrac{1}{T\sqrt{p_{1}}p_{2}}\left\|(\sum_{t=1}^{T}\Rb^{\top}\Eb_{t}\Wb_{2}\Wb_{2}^{\top}\Eb_{t}^{\top}\Wb_{1})\Zb^{-1/2}\right\|_{F}\\
    		& \leq \dfrac{1}{T\sqrt{p_{1}}p_{2}}\sqrt{\left\|\sum_{t=1}^{T}\Rb^{\top}\Eb_{t}\Wb_{2}\right\|_{F}^{2} \times \left\|\sum_{t=1}^{T}\Wb_{2}^{\top}\Eb_{t}^{\top}\Wb_{1}\right\|_{F}^{2}}\left\|\Zb^{-1/2}\right\|_{F}\\
    		&=O_{p}\left(\dfrac{1}{p_{2}\nu_{\min}(\Hb_{1})\nu_{\min}^{2}(\Hb_{2})}\right).
    	\end{aligned}$$
    	As $\Rb^{\top}(\hat{\Rb}^{(1)}-\Rb\hat{\Hb}_{r}^{(1)})=\Rb^{\top}(\cI\cI+\cI\cI\cI+\cI\cV)$, combining the above three items, we get
    	$$\left\|\Rb^{\top}(\hat{\Rb}^{(1)}-\Rb\hat{\Hb}_{r}^{(1)})\right\|_{F}=O_{p}\left(\dfrac{\sqrt{p_{1}}}{\sqrt{Tp_{2}}\nu_{\min}(\Hb_{1})\nu_{\min}(\Hb_{2})}+\dfrac{1}{p_{2}\nu_{\min}(\Hb_{1})\nu_{\min}^{2}(\Hb_{2})}\right).$$
    	As for the second formula, by the proof of Theorem \ref{convergence of an iteration}, 	we have
    	$$\hat{\Cb}^{(1)}-\Cb\hat{\Hb}_{c}^{(1)}=\sqrt{p_{2}}\left(\bdelta_{2}^{(1)}(\Zb^{(1)})^{-1/2}+\bdelta_{3}^{(1)}(\Zb^{(1)})^{-1/2}+\bdelta_{4}^{(1)}(\Zb^{(1)})^{-1/2}\right).$$
    	For the first term,
    	$$\begin{aligned}
    		&\sqrt{p_{2}}\left\|\Cb^{\top}\bdelta_{2}^{(1)}(\Zb^{(1)})^{-1/2}\right\|_{F}=\dfrac{1}{p_{1}\sqrt{p_{2}}}\left\|(\sum_{t=1}^{T}\Cb^{\top}\Eb_{t}^{\top}\hat{\Rb}^{(1)}\Wb_{1}^{\top}\Rb\Fb_{t}\Cb^{\top}\Wb_{2})(\Zb^{(1)})^{-1/2}\right\|_{F}\\
    		& \leq \dfrac{1}{p_{1}\sqrt{p_{2}}}\left\|\sum_{t=1}^{T}\Cb^{\top}\Eb_{t}^{\top}\hat{\Rb}^{(1)}\Fb_{t}\right\|_{F}\left\|\Wb_{1}^{\top}\Rb\right\|_{2}\left\|\Cb^{\top}\Wb_{2}\right\|_{2}\left\|(\Zb_{t}^{(1)})^{-1/2}\right\|_{F}\\
    		&=O_{p}\left(\dfrac{1}{T\nu_{\min}(\Hb_{2})}+\dfrac{1}{T\sqrt{p_{1}}\nu_{\min}(\Hb_{1})\nu_{\min}^{2}(\Hb_{2})}+\dfrac{\sqrt{p_{2}}}{\sqrt{Tp_{1}}}\right),
    	\end{aligned}$$
    	where $$\begin{aligned}
    		\left\|\sum_{t=1}^{T}\Cb^{\top}\Eb_{t}^{\top}\hat{\Rb}^{(1)}\Fb_{t}\right\|_{F} &\leq \left\|\sum_{t=1}^{T}\Cb^{\top}\Eb_{t}^{\top}(\hat{\Rb}^{(1)}-\Rb\hat{\Hb}_{r}^{(1)})\Fb_{t}\right\|_{F}+\left\|\sum_{t=1}^{T}\Cb^{\top}\Eb_{t}^{\top}\Rb\hat{\Hb}_{r}^{(1)}\Fb_{t}\right\|_{F}\\
    		& \leq \left\|\sum_{t=1}^{T}\Cb^{\top}\Eb_{t}^{\top}\Fb_{t}\right\|_{F}\left\|\hat{\Rb}^{(1)}-\Rb\hat{\Hb}_{r}^{(1)}\right\|_{F}+\left\|\sum_{t=1}^{T}\Cb^{\top}\Eb_{t}^{\top}\Rb\Fb_{t}\right\|_{F}\left\|\hat{\Hb}_{r}^{(1)}\right\|_{F}\\
    		&=O_{p}\left(\dfrac{p_{1}}{\nu_{\min}(\Hb_{2})}+\dfrac{\sqrt{p_{1}}}{\nu_{\min}(\Hb_{1})\nu_{\min}^{2}(\Hb_{2})}+\sqrt{Tp_{1}p_{2}}\right).
    	\end{aligned}$$
    	For the second term,
    	$$\begin{aligned}
    		\sqrt{p_{2}}\left\|\Cb^{\top}\bdelta_{3}^{(1)}(\Zb^{(1)})^{-1/2}\right\|_{F}&=\dfrac{\sqrt{p_{2}}}{p_{1}}\left\|(\sum_{t=1}^{T}\Fb_{t}^{\top}\Rb^{\top}\hat{\Rb}^{(1)}\Wb_{1}^{\top}\Eb_{t}\Wb_{2})(\Zb^{(1)})^{-1/2}\right\|_{F}\\
    		& \asymp \sqrt{p_{2}}\left\|(\sum_{t=1}^{T}\Fb_{t}^{\top}\Wb_{1}^{\top}\Eb_{t}\Wb_{2})(\Zb^{(1)})^{-1/2}\right\|_{F}\\
    		& \leq \sqrt{p_{2}}\left\|\sum_{t=1}^{T}\Fb_{t}^{\top}\Wb_{1}^{\top}\Eb_{t}\Wb_{2}\right\|_{F}\left\|(\Zb^{(1)})^{-1/2}\right\|_{F}\\
    		&=O_{p}\left(\dfrac{\sqrt{p_{2}}}{\sqrt{Tp_{1}}\nu_{\min}(\Hb_{1})\nu_{\min}(\Hb_{2})}\right).
    	\end{aligned}$$
    	For the third term,
    	$$\begin{aligned}
    		&\sqrt{p_{2}}\left\|\Cb^{\top}\bdelta_{4}^{(1)}(\Zb^{(1)})^{-1/2}\right\|_{F}=\dfrac{1}{p_{1}\sqrt{p_{2}}}\left\|(\sum_{t=1}^{T}\Cb^{\top}\Eb_{t}^{\top}\hat{\Rb}^{(1)}\Wb_{1}^{\top}\Eb_{t}\Wb_{2})(\Zb^{(1)})^{-1/2}\right\|_{F}\\
    		& \leq \dfrac{1}{p_{1}\sqrt{p_{2}}}\sqrt{\left\|\sum_{t=1}^{T}\Cb^{\top}\Eb_{t}^{\top}\hat{\Rb}^{(1)}\right\|_{F}^{2}\times\left\|\sum_{t=1}^{T}\Wb_{1}^{\top}\Eb_{t}\Wb_{2}\right\|_{F}^{2}}\left\|(\Zb^{(1)})^{-1/2}\right\|_{F}\\
    		&=O_{p}\left(\dfrac{1}{\sqrt{Tp_{1}p_{2}}\nu_{\min}(\Hb_{1})\nu_{\min}^{2}(\Hb_{2})}+\dfrac{1}{\sqrt{Tp_{2}}p_{1}\nu_{\min}^{2}(\Hb_{1})\nu_{\min}^{3}(\Hb_{2})}+\dfrac{1}{p_{1}\nu_{\min}(\Hb_{1})\nu_{\min}(\Hb_{2})}\right),
    	\end{aligned}$$
    where	the last equation holds due to
    	$$\begin{aligned}
    		&\left\|\sum_{t=1}^{T}\Cb^{\top}\Eb_{t}^{\top}\hat{\Rb}^{(1)}\right\|_{F}^{2}=\left\|\sum_{t=1}^{T}\Cb^{\top}\Eb_{t}^{\top}(\hat{\Rb}^{(1)}-\Rb\hat{\Hb}_{r}^{(1)}+\Rb\hat{\Hb}_{r}^{(1)})\right\|_{F}^{2}\\
    		& \leq \left\|\sum_{t=1}^{T}\Cb^{\top}\Eb_{t}^{\top}(\hat{\Rb}^{(1)}-\Rb\hat{\Hb}_{r}^{(1)})\right\|_{F}^{2}+\left\|\sum_{t=1}^{T}\Cb^{\top}\Eb_{t}^{\top}\Rb\hat{\Hb}_{r}^{(1)}\right\|_{F}^{2}\\
    		&=O_{p}\left(\dfrac{p_{1}^{2}}{\nu_{\min}^{2}(\Hb_{2})}+\dfrac{p_{1}}{\nu_{\min}^{2}(\Hb_{1})\nu_{\min}^{4}(\Hb_{2})}+Tp_{1}p_{2}\right).
    	\end{aligned}$$
    	Hence, $$
    	\left\|\Cb^{\top}(\hat{\Cb}^{(1)}-\Cb\hat{\Hb}_{c}^{(1)})\right\|_{F}=O_{p}\left(\dfrac{1}{T\nu_{\min}(\Hb_{2})}+\dfrac{\sqrt{p_{2}}}{\sqrt{Tp_{1}}\nu_{\min}(\Hb_{1})\nu_{\min}(\Hb_{2})}+\dfrac{1}{p_{1}\nu_{\min}(\Hb_{1})\nu_{\min}(\Hb_{2})}\right).$$
    	
    \end{proof}

    \begin{lemma}\label{loading product of (s+1)-th}
    	Under Assumption \ref{Assumption 1} (1), Assumption \ref{Assumption 2} and Assumption \ref{Assumption 4}-\ref{Assumption 7}, take $\hat{\Rb}^{(s+1)}, \hat{\Cb}^{(s+1)}$ as the result of a $(s+1)$th step iteration, as $\min\{T, p_{1}, p_{2}\} \to \infty$, then we have
    	$$\|\Rb^{\top}(\hat{\Rb}^{(s+1)}-\Rb\hat{\Hb}_{r}^{(s+1)})\|_{F}=O_{p}\left(\dfrac{\sqrt{p_{1}}}{\sqrt{Tp_{2}}}+\dfrac{p_{1}\sqrt{w_{r}^{(s)}}}{\sqrt{Tp_{2}}}+\dfrac{\sqrt{w_{r}^{(s)}}}{p_{2}}+\dfrac{\sqrt{p_{1}w_{c}^{(s)}}}{\sqrt{T}}+\dfrac{p_{1}\sqrt{w_{r}^{(s)}w_{c}^{(s)}}}{\sqrt{T}}+\sqrt{w_{r}^{(s)}}w_{c}^{(s)}\right),$$
    	$$\|\Cb^{\top}(\hat{\Cb}^{(s+1)}-\Cb\hat{\Hb}_{c}^{(s+1)})\|_{F}=O_{p}\left(\dfrac{\sqrt{p_{2}}}{\sqrt{Tp_{1}}}+\dfrac{1}{p_{1}}+\dfrac{\sqrt{w_{r}^{(s)}}}{\sqrt{p_{1}}}+\dfrac{\sqrt{p_{2}w_{r}^{(s)}}}{\sqrt{T}}+\dfrac{p_{2}\sqrt{w_{c}^{(s)}}}{\sqrt{Tp_{1}}}+\dfrac{p_{2}\sqrt{w_{r}^{(s)}w_{c}^{(s)}}}{\sqrt{T}}+\sqrt{w_{r}^{(s)}w_{r}^{(s+1)}}\right).$$
    \end{lemma}

    \begin{proof}
    	By the proof of Theorem \ref{the iteration results},
    	$$\hat{\Rb}^{(s+1)}-\Rb\hat{\Hb}_{r}^{(s+1)}=\sqrt{p_{1}}(\bdelta_{2}^{(s+1)}+\bdelta_{3}^{(s+1)}+\bdelta_{4}^{(s+1)})(\Zb^{(s+1)})^{-1/2}.$$
    	By  Lemma \ref{lemma11} (1) and (2), we have
    	$$\begin{aligned}
    	\left\|\sqrt{p_{1}}\Rb^{\top}\bdelta_{2}^{(s+1)}(\Zb^{(s+1)})^{-1/2}\right\|_{F}&\asymp\left\|\sqrt{p_{1}}\left(\sum_{t=1}^{T}\Fb_{t}\hat{\Cb}^{(s)\top}\Eb_{t}^{\top}\hat{\Rb}^{(s)}\right)\left(\Zb^{(s+1)}\right)^{-1/2}\right\|_{F}\\
    	&\leq \sqrt{p_{1}}\left\|\sum_{t=1}^{T}\Fb_{t}\hat{\Cb}^{(s)\top}\Eb_{t}^{\top}\hat{\Rb}^{(s)}\right\|_{F}\left\|\left(\Zb^{(s+1)}\right)^{-1/2}\right\|_{F}\\
    	&=O_{p}\left(\dfrac{p_{1}\sqrt{w_{r}^{(s)}w_{c}^{(s)}}}{\sqrt{T}}+\dfrac{p_{1}\sqrt{w_{r}^{(s)}}}{\sqrt{Tp_{2}}}+\dfrac{\sqrt{p_{1}w_{c}^{(s)}}}{\sqrt{T}}+\dfrac{\sqrt{p_{1}}}{\sqrt{Tp_{2}}}\right),
    	\end{aligned}$$
    	where the last equation holds due to
    	$$\left\|\sum_{t=1}^{T}\Fb_{t}\hat{\Cb}^{(s)\top}\Eb_{t}^{\top}\hat{\Rb}^{(s)}\right\|_{F} \leq \left\|\sum_{t=1}^{T}\Fb_{t}\hat{\Cb}^{(s)\top}\Eb_{t}^{\top}\right\|_{F}\left\|\hat{\Rb}^{(s)}-\Rb\hat{\Hb}_{r}^{(s)}\right\|_{F}+\left\|\sum_{t=1}^{T}\Fb_{t}\hat{\Cb}^{(s)\top}\Eb_{t}^{\top}\Rb\right\|_{F}\left\|\hat{\Hb}_{r}^{(s)}\right\|_{F}.$$
    	
    	$$\begin{aligned}
    	\left\|\sqrt{p_{1}}\Rb^{\top}\bdelta_{3}^{(s+1)}(\Zb^{(s+1)})^{-1/2}\right\|_{F}&\asymp \sqrt{p_{1}}\left\|\left(\sum_{t=1}^{T}\Rb^{\top}\Eb_{t}\hat{\Cb}^{(s)}\Fb_{t}^{\top}\right)\left(\Zb^{(s+1)}\right)^{-1/2}\right\|_{F}\\
    	&\leq \sqrt{p_{1}}\left\|\sum_{t=1}^{T}\Rb^{\top}\Eb_{t}\hat{\Cb}^{(s)}\Fb_{t}^{\top}\right\|_{F}\left\|\left(\Zb^{(s+1)}\right)^{-1/2}\right\|_{F}\\
    	&=O_{p}\left(\dfrac{\sqrt{p_{1}w_{c}^{(s)}}}{\sqrt{T}}+\dfrac{\sqrt{p_{1}}}{\sqrt{Tp_{2}}}\right).
    	\end{aligned}$$
    	By the fact that $$\begin{aligned}
    	\left\|\sum_{t=1}^{T}\Rb^{\top}\Eb_{t}\hat{\Cb}^{(s)}\hat{\Cb}^{(s)\top}\Eb_{t}^{\top}\hat{\Rb}^{(s)}\right\|_{F}&\leq \left\|\sum_{t=1}^{T}\Rb^{\top}\Eb_{t}\hat{\Cb}^{(s)}\hat{\Cb}^{(s)\top}\Eb_{t}^{\top}(\hat{\Rb}^{(s)}-\Rb\hat{\Hb}_{r}^{(s)})\right\|_{F}+\left\|\sum_{t=1}^{T}\Rb^{\top}\Eb_{t}\hat{\Cb}^{(s)}\hat{\Cb}^{(s)\top}\Eb_{t}^{\top}\Rb\hat{\Hb}_{r}^{(s)}\right\|_{F}\\
    	&\leq \left\|\sum_{t=1}^{T}\Rb^{\top}\Eb_{t}\hat{\Cb}^{(s)}\hat{\Cb}^{(s)\top}\Eb_{t}^{\top}\right\|_{F}\left\|\hat{\Rb}^{(s)}-\Rb\hat{\Hb}_{r}^{(s)}\right\|_{F}+\left\|\sum_{t=1}^{T}\Rb^{\top}\Eb_{t}\hat{\Cb}^{(s)}\right\|_{F}^{2}\left\|\hat{\Hb}_{r}^{(s)}\right\|_{F},
    	\end{aligned}$$
    and Lemma \ref{lemma10} (1), we can get
    	$$\begin{aligned}
    	\left\|\sqrt{p_{1}}\Rb^{\top}\bdelta_{4}^{(s+1)}(\Zb^{(s+1)})^{-1/2}\right\|_{F}&=\left\|\dfrac{1}{\sqrt{p_{1}}p_{2}}\left(\sum_{t=1}^{T}\Rb^{\top}\Eb_{t}\hat{\Cb}^{(s)}\hat{\Cb}^{(s)\top}\Eb_{t}^{\top}\hat{\Rb}^{(s)}\right)\left(\Zb^{(s+1)}\right)^{-1/2}\right\|_{F}\\
    	&\leq \dfrac{1}{\sqrt{p_{1}}p_{2}}\left\|\sum_{t=1}^{T}\Rb^{\top}\Eb_{t}\hat{\Cb}^{(s)}\hat{\Cb}^{(s)\top}\Eb_{t}^{\top}\hat{\Rb}^{(s)}\right\|_{F}\left\|\left(\Zb^{(s+1)}\right)^{-1/2}\right\|_{F}\\
    	&=O_{p}\left(\dfrac{\sqrt{p_{1}w_{r}^{(s)}}}{\sqrt{Tp_{2}}}+\dfrac{\sqrt{w_{r}^{(s)}}}{p_{2}}+\sqrt{w_{r}^{(s)}}w_{c}^{(s)}+\dfrac{\sqrt{p_{1}w_{r}^{(s)}}w_{c}^{(s)}}{\sqrt{T}}+\dfrac{\sqrt{w_{c}^{(s)}}}{\sqrt{Tp_{1}}p_{2}}+\dfrac{1}{\sqrt{Tp_{1}p_{2}}p_{2}}\right).
    	\end{aligned}$$
    	Hence, $$\|\Rb^{\top}(\hat{\Rb}^{(s+1)}-\Rb\hat{\Hb}_{r}^{(s+1)})\|_{F}=O_{p}\left(\dfrac{\sqrt{p_{1}}}{\sqrt{Tp_{2}}}+\dfrac{p_{1}\sqrt{w_{r}^{(s)}}}{\sqrt{Tp_{2}}}+\dfrac{\sqrt{w_{r}^{(s)}}}{p_{2}}+\dfrac{\sqrt{p_{1}w_{c}^{(s)}}}{\sqrt{T}}+\dfrac{p_{1}\sqrt{w_{r}^{(s)}w_{c}^{(s)}}}{\sqrt{T}}+\sqrt{w_{r}^{(s)}}w_{c}^{(s)}\right).$$
    	And by a similar proof, we can get
$$\|\Cb^{\top}(\hat{\Cb}^{(s+1)}-\Cb\hat{\Hb}_{c}^{(s+1)})\|_{F}=O_{p}\left(\dfrac{\sqrt{p_{2}}}{\sqrt{Tp_{1}}}+\dfrac{1}{p_{1}}+\dfrac{\sqrt{w_{r}^{(s)}}}{\sqrt{p_{1}}}+\dfrac{\sqrt{p_{2}w_{r}^{(s)}}}{\sqrt{T}}+\dfrac{p_{2}\sqrt{w_{c}^{(s)}}}{\sqrt{Tp_{1}}}+\dfrac{p_{2}\sqrt{w_{r}^{(s)}w_{c}^{(s)}}}{\sqrt{T}}+\sqrt{w_{r}^{(s)}w_{r}^{(s+1)}}\right).$$

    \end{proof}

	\begin{lemma}\label{lemma9}
		Suppose that $T, p_{1}, p_{2}$ tend to infinity, and $m_{1}=k_{1}, m_{2}=k_{2}$ are fixed. If  Assumption \ref{Assumption 1} (1), Assumption \ref{Assumption 2} and Assumption \ref{Assumption 4}-\ref{Assumption 7} hold, then we have
		$$(1) \ \sum_{t=1}^{T}\|\Eb_{t}\hat{\Cb}^{(s)}\|_{F}^{2}=O_{p}\left(Tp_{1}p_{2}+Tp_{1}p_{2}^{2}w_{c}^{(s)}\right);$$
		$$(2) \ \sum_{t=1}^{T}\|\Eb_{t}^{\top}\hat{\Rb}^{(s)}\|_{F}^{2}=O_{p}\left(Tp_{1}p_{2}+Tp_{1}^{2}p_{2}w_{r}^{(s)}\right).$$
	\end{lemma}
	
	\begin{proof}
		(1) As $\sum_{t=1}^{T}\|\Eb_{t}\Cb\|_{F}^{2}=\sum_{t=1}^{T}\sum_{i=1}^{p_{1}}(\sum_{k=1}^{p_{2}}e_{t,ik}c_{k})^{2} \leq c\sum_{t=1}^{T}\sum_{i=1}^{p_{1}}\sum_{k,k^{\prime}=1}^{p_{2}}e_{t,ik}e_{t,ik^{\prime}}=O_{p}\left(Tp_{1}p_{2}\right),$ $\sum_{t=1}^{T}\|\Eb_{t}\|_{F}^{2}=\sum_{t=1}^{T}\sum_{i=1}^{p_{1}}\sum_{j=1}^{p_{2}}e_{t,ij}^{2}=O_{p}\left(Tp_{1}p_{2}\right)$, then
		$$\sum_{t=1}^{T}\|\Eb_{t}\hat{\Cb}^{(s)}\|_{F}^{2} \leq \sum_{t=1}^{T}\|\Eb_{t}\Cb\hat{\Hb}_{c}^{(s)}\|_{F}^{2}+\sum_{t=1}^{T}\|\Eb_{t}(\hat{\Cb}^{(s)}-\Cb\hat{\Hb}_{c}^{(s)})\|_{F}^{2}=O_{p}\left(Tp_{1}p_{2}+Tp_{1}p_{2}^{2}w_{c}^{(s)}\right).$$
		
		(2) $\sum_{t=1}^{T}\|\Rb^{\top}\Eb_{t}\|_{F}^{2}=\sum_{t=1}^{T}\sum_{j=1}^{p_{2}}(\sum_{i=1}^{p_{1}}r_{i}e_{t,ij})^{2}\leq c\sum_{t=1}^{T}\sum_{j=1}^{p_{2}}\sum_{i_{1},i_{2}=1}^{p_{1}}e_{t,i_{1}j}e_{t,i_{2}j}=O_{p}\left(Tp_{1}p_{2}\right)$,
		$$\sum_{t=1}^{T}\|\Eb_{t}^{\top}\hat{\Rb}^{(s)}\|_{F}^{2}\leq \sum_{t=1}^{T}\|\Eb_{t}^{\top}(\hat{\Rb}^{(s)}-\Rb\hat{\Hb}_{r}^{(s)})\|_{F}^{2}+\sum_{t=1}^{T}\|\Eb_{t}^{\top}\Rb\hat{\Hb}_{r}^{(s)}\|_{F}^{2}=O_{p}\left(Tp_{1}^{2}p_{2}w_{r}^{(s)}+Tp_{1}p_{2}\right).$$
	\end{proof}
	
	\begin{lemma}\label{lemma10}
		Suppose that $T, p_{1}, p_{2}$ tend to infinity, and $m_{1}=k_{1}, m_{2}=k_{2}$ are fixed. If  Assumption \ref{Assumption 1} (1), Assumption \ref{Assumption 2} and Assumption \ref{Assumption 4}-\ref{Assumption 7} hold, then we have
\begin{itemize}
  \item (1)	$$ \ \|\sum_{t=1}^{T}\Eb_{t}\hat{\Cb}^{(s)}\hat{\Cb}^{(s)\top}\Eb_{t}^{\top}\Rb\|_{F}^{2}=O_{p}\left(Tp_{1}^{2}p_{2}^{3}+T^{2}p_{1}p_{2}^{2}+(Tp_{1}^{2}p_{2}^{4}+T^{2}p_{1}p_{2}^{4})w_{c}^{(s)2}\right);$$
  \item (2) 	$$\ \|\sum_{t=1}^{T}\Eb_{t}^{\top}\hat{\Rb}^{(s+1)}\hat{\Rb}^{(s)\top}\Eb_{t}\Cb\|_{F}^{2}=O_{p}\left(Tp_{1}^{3}p_{2}^{2}+T^{2}p_{1}^{2}p_{2}+(Tp_{1}^{4}p_{2}^{2}+T^{2}p_{1}^{4}p_{2})w_{r}^{(s)}w_{r}^{(s+1)}\right).$$
\end{itemize}

	\end{lemma}
	
	\begin{proof}
		(1) For simplicity, we fix $k_{1}=k_{2}=m_{1}=m_{2}=1$. Note that
		$$
		\begin{aligned}
			\|\sum_{t=1}^{T}\Eb_{t}\hat{\Cb}^{(s)}\hat{\Cb}^{(s)\top}\Eb_{t}^{\top}\Rb\|_{F}^{2}& \leq \|\sum_{t=1}^{T}\Eb_{t}\Cb\hat{\Hb}_{c}^{(s)}\hat{\Hb}_{c}^{(s)\top}\Cb^{\top}\Eb_{t}^{\top}\Rb\|_{F}^{2}+\|\sum_{t=1}^{T}\Eb_{t}(\hat{\Cb}^{(s)}-\Cb\hat{\Hb}_{c}^{(s)})\hat{\Hb}_{c}^{(s)\top}\Cb^{\top}\Eb_{t}^{\top}\Rb\|_{F}^{2}\\
			&+\|\sum_{t=1}^{T}\Eb_{t}\Cb\hat{\Hb}_{c}^{(s)}(\hat{\Cb}^{(s)}-\Cb\hat{\Hb}_{c}^{(s)})^{\top}\Eb_{t}^{\top}\Rb\|_{F}^{2}+\|\sum_{t=1}^{T}\Eb_{t}(\hat{\Cb}^{(s)}-\Cb\hat{\Hb}_{c}^{(s)})(\hat{\Cb}^{(s)}-\Cb\hat{\Hb}_{c}^{(s)})^{\top}\Eb_{t}^{\top}\Rb\|_{F}^{2}.
		\end{aligned}
		$$
		
		For the first term,
		$$
		\begin{aligned}
			\|\sum_{t=1}^{T}\Eb_{t}\Cb\hat{\Hb}_{c}^{(s)}\hat{\Hb}_{c}^{(s)\top}\Cb^{\top}\Eb_{t}^{\top}\Rb\|_{F}^{2}& \leq \|\sum_{t=1}^{T}\Eb_{t}\Cb^{\top}\Eb_{t}^{\top}\Rb\|_{F}^{2}\|\Cb\|_{F}^{2}\|\hat{\Hb}_{c}^{(s)}\|_{F}^{4}\\
			& \leq p_{2} \sum_{i^{\prime}=1}^{p_{1}}\sum_{j^{\prime}=1}^{p_{2}}(\sum_{t=1}^{T}\sum_{i=1}^{p_{1}}\sum_{j=1}^{p_{2}}r_{i}c_{j}e_{t,ij}e_{t,i^{\prime}j^{\prime}})^{2}\\
			&=O_{p}\left(Tp_{1}^{2}p_{2}^{3}+T^{2}p_{1}p_{2}^{2}\right),	
		\end{aligned}
		$$
		where
		$$
		\begin{aligned}
			\EE(\sum_{t=1}^{T}\sum_{i=1}^{p_{1}}\sum_{j=1}^{p_{2}}r_{i}c_{j}e_{t,ij}e_{t,i^{\prime}j^{\prime}})^{2}& \leq c\sum_{t,s=1}^{T}\sum_{i_{1},i_{2}=1}^{p_{1}}\sum_{j_{1},j_{2}=1}^{p_{2}}\text{Cov}(e_{t,i_{1}j_{2}}e_{t,i^{\prime}j^{\prime}},e_{s,i_{2}j_{2}}e_{s,i^{\prime}j^{\prime}})+c(\sum_{t=1}^{T}\sum_{i=1}^{p_{1}}\sum_{j=1}^{p_{2}}\EE e_{t,ij}e_{t,i^{\prime}j^{\prime}})^{2}\\
			&=O\left(Tp_{1}p_{2}+T^{2}\right).
		\end{aligned}
		$$
		
		For the second term,
		$$\begin{aligned}
			\|\sum_{t=1}^{T}\Eb_{t}(\hat{\Cb}^{(s)}-\Cb\hat{\Hb}_{c}^{(s)})\hat{\Hb}_{c}^{(s)\top}\Cb^{\top}\Eb_{t}^{\top}\Rb\|_{F}^{2}& \leq \|\sum_{t=1}^{T}\Eb_{t}\Cb^{\top}\Eb_{t}^{\top}\bR\|_{F}^{2}\|\hat{\Cb}^{(s)}-\Cb\hat{\Hb}_{c}^{(s)}\|_{F}^{2}\|\hat{\Hb}_{c}^{(s)}\|_{F}^{2}\\
			&=O_{p}\left((Tp_{1}^{2}p_{2}^{3}+T^{2}p_{1}p_{2}^{2})w_{c}^{(s)}\right).
		\end{aligned}$$
		
		For the third term,
		$$\begin{aligned}
			\|\sum_{t=1}^{T}\Eb_{t}\Cb\hat{\Hb}_{c}^{(s)}(\hat{\Cb}^{(s)}-\Cb\hat{\Hb}_{c}^{(s)})^{\top}\Eb_{t}^{\top}\Rb\|_{F}^{2}& \leq \|\sum_{t=1}^{T}\Eb_{t}\Cb\Rb^{\top}\Eb_{t}\|_{F}^{2}\|\hat{\Cb}^{(s)}-\Cb\hat{\Hb}_{c}^{(s)}\|_{F}^{2}\|\hat{\Hb}_{c}^{(s)}\|_{F}^{2}\\
			&=\sum_{i^{\prime}=1}^{p_{1}}\sum_{j^{\prime}=1}^{p_{2}}(\sum_{t=1}^{T}\sum_{i=1}^{p_{1}}\sum_{j=1}^{p_{2}}r_{i}c_{j}e_{t,i^{\prime}j}e_{t,ij^{\prime}})^{2}\|\hat{\Cb}^{(s)}-\Cb\hat{\Hb}_{c}^{(s)}\|_{F}^{2}\|\hat{\Hb}_{c}^{(s)}\|_{F}^{2}\\
			&=O_{p}\left((Tp_{1}^{2}p_{2}^{3}+T^{2}p_{1}p_{2}^{2})w_{c}^{(s)}\right),
		\end{aligned}$$
		where $$
		\begin{aligned}
			\EE(\sum_{t=1}^{T}\sum_{i=1}^{p_{1}}\sum_{j=1}^{p_{2}}r_{i}c_{j}e_{t,i^{\prime}j}e_{t,ij^{\prime}})^{2} &\leq c\sum_{t,s=1}^{T}\sum_{i_{1}i_{2}=1}^{p_{1}}\sum_{j_{1}j_{2}=1}^{p_{2}}\text{Cov}(e_{t,i_{\prime}j_{1}}e_{t,i_{1}j_{\prime}},e_{s,i^{\prime}j_{2}}e_{s,i_{2}j_{\prime}})+c(\sum_{t=1}^{T}\sum_{i=1}^{p_{1}}\sum_{j=1}^{p_{2}}\EE e_{t,i^{\prime}j}e_{t,ij^{\prime}})^{2}\\
			&=O\left(Tp_{1}p_{2}+T^{2}\right).
		\end{aligned}$$
		
		For the last term,
		$$
		\begin{aligned}
			\|\sum_{t=1}^{T}\Eb_{t}(\hat{\Cb}^{(s)}-\Cb\hat{\Hb}_{c}^{(s)})(\hat{\Cb}^{(s)}-\Cb\hat{\Hb}_{c}^{(s)})^{\top}\Eb_{t}^{\top}\Rb\|_{F}^{2}& \leq \|\hat{\Cb}^{(s)}-\Cb\hat{\Hb}_{c}^{(s)}\|_{F}^{2}\|\sum_{t=1}^{T}\Eb_{t}(\hat{\Cb}^{(s)}-\Cb\hat{\Hb}_{c}^{(s)})^{\top}\Eb_{t}^{\top}\Rb\|_{F}^{2}\\
			& \leq \|\hat{\Cb}^{(s)}-\Cb\hat{\Hb}_{c}^{(s)}\|_{F}^{4}\sum_{i=1}^{p_{1}}\sum_{j=1}^{p_{2}}\|\sum_{t=1}^{T}\Eb_{t}^{\top}\Rb e_{t,ij}\|_{F}^{2}\\
			& \leq \|\hat{\Cb}^{(s)}-\Cb\hat{\Hb}_{c}^{(s)}\|_{F}^{4}\sum_{i=1}^{p_{1}}\sum_{j=1}^{p_{2}}\sum_{k=1}^{p_{2}}(\sum_{t=1}^{T}\sum_{i^{\prime}=1}^{p_{1}}r_{i^{\prime}}e_{t,i^{\prime}k}e_{t,ij})^{2}\\
			&=O_{p}\left((T^{2}p_{1}p_{2}^{4}+Tp_{1}^{2}p_{2}^{4})w_{c}^{(s)2}\right),
		\end{aligned}
		$$
		where $$
		\begin{aligned}
			\sum_{i=1}^{p_{1}}\sum_{j=1}^{p_{2}}\sum_{k=1}^{p_{2}}\EE(\sum_{t=1}^{T}\sum_{i^{\prime}=1}^{p_{1}}r_{i^{\prime}}e_{t,i^{\prime}k}e_{t,ij})^{2} &\leq c\sum_{i=1}^{p_{1}}\sum_{j=1}^{p_{2}}\sum_{k=1}^{p_{2}}(\sum_{t=1}^{T}\sum_{i^{\prime}=1}^{p_{1}}\EE e_{t,i^{\prime}k}e_{t,ij})^{2}\\
			&+c\sum_{i=1}^{p_{1}}\sum_{j=1}^{p_{2}}\sum_{k=1}^{p_{2}}\sum_{t,s=1}^{T}\sum_{i_{1}i_{2}=1}^{p_{1}}\text{Cov}(e_{t,i_{1}k}e_{t,ij},e_{s,i_{2}k}e_{s,ij})\\
			&=O_{p}\left(T^{2}p_{1}p_{2}^{2}+Tp_{1}^{2}p_{2}^{2}\right).
		\end{aligned}
		$$
		
		(2) Note that
		$$\begin{aligned}
			\|\sum_{t=1}^{T}\Eb_{t}^{\top}\hat{\Rb}^{(s+1)}\hat{\Rb}^{(s)\top}\Eb_{t}\Cb\|_{F}^{2}&=\|\sum_{t=1}^{T}\Eb_{t}^{\top}(\hat{\Rb}^{(s+1)}-\Rb\hat{\Hb}_{r}^{(s+1)}+\Rb\hat{\Hb}_{r}^{(s+1)})(\hat{\Rb}^{(s)}-\Rb\hat{\Hb}_{r}^{(s)}+\Rb\hat{\Hb}_{r}^{(s)})^{\top}\Eb_{t}\Cb\|_{F}^{2}\\
			&\leq \|\sum_{t=1}^{T}\Eb_{t}^{\top}\Rb\hat{\Hb}_{r}^{(s+1)}\hat{\Hb}_{r}^{(s)\top}\Rb^{\top}\Eb_{t}\Cb\|_{F}^{2}+\|\sum_{t=1}^{T}\Eb_{t}^{\top}\Rb\hat{\Hb}_{r}^{(s+1)}(\hat{\Rb}^{(s)}-\Rb\hat{\Hb}_{r}^{(s)})^{\top}\Eb_{t}\Cb\|_{F}^{2}\\
			&+\|\sum_{t=1}^{T}\Eb_{t}^{\top}(\hat{\Rb}^{(s+1)}-\Rb\hat{\Hb}_{r}^{(s+1)})\hat{\Hb}_{r}^{(s)\top}\Rb^{\top}\Eb_{t}\Cb\|_{F}^{2}\\
			&+\|\sum_{t=1}^{T}\Eb_{t}^{\top}(\hat{\Rb}^{(s+1)}-\Rb\hat{\Hb}_{r}^{(s+1)})(\hat{\Rb}^{(s)}-\Rb\hat{\Hb}_{r}^{(s)})^{\top}\Eb_{t}\Cb\|_{F}^{2}.
		\end{aligned}$$
		
		For the first term,
		$$
		\|\sum_{t=1}^{T}\Eb_{t}^{\top}\Rb\hat{\Hb}_{r}^{(s+1)}\hat{\Hb}_{r}^{(s)\top}\Rb^{\top}\Eb_{t}\Cb\|_{F}^{2}\leq \|\sum_{t=1}^{T}\Eb_{t}^{\top}\Rb^{\top}\Eb_{t}\Cb\|_{F}^{2}\|\Rb\|_{F}^{2}\|\hat{\Hb}_{r}^{(s+1)}\|_{F}^{2}\|\hat{\Hb}_{r}^{(s)}\|_{F}^{2}=O_{p}\left(Tp_{1}^{3}p_{2}^{2}+T^{2}p_{1}^{2}p_{2}\right).$$
		
		For the second term,
		$$\begin{aligned}
			\|\sum_{t=1}^{T}\Eb_{t}^{\top}\Rb\hat{\Hb}_{r}^{(s+1)}(\hat{\Rb}^{(s)}-\Rb\hat{\Hb}_{r}^{(s)})^{\top}\Eb_{t}\Cb\|_{F}^{2}&\leq \|\sum_{t=1}^{T}\Eb_{t}^{\top}\Rb\Cb^{\top}\Eb_{t}^{\top}\|_{F}^{2}\|\hat{\Rb}^{(s)}-\Rb\hat{\Hb}_{r}^{(s)}\|_{F}^{2}\|\hat{\Hb}_{r}^{(s+1)}\|_{F}^{2}\\
			&=O_{p}\left((Tp_{1}^{3}p_{2}^{2}+T^{2}p_{1}^{2}p_{2})w_{r}^{(s)}\right).
		\end{aligned}$$
		
		For the third term,
		$$\begin{aligned}
			\|\sum_{t=1}^{T}\Eb_{t}^{\top}(\hat{\Rb}^{(s+1)}-\Rb\hat{\Hb}_{r}^{(s+1)})\hat{\Hb}_{r}^{(s)\top}\Rb^{\top}\Eb_{t}\Cb\|_{F}^{2}&\leq \|\sum_{t=1}^{T}\Eb_{t}^{\top}\Rb^{\top}\Eb_{t}\Cb\|_{F}^{2}\|\hat{\Rb}^{(s+1)}-\Rb\hat{\Hb}_{r}^{(s+1)}\|_{F}^{2}\|\hat{\Hb}_{r}^{(s)}\|_{F}^{2}\\
			&=O_{p}\left((Tp_{1}^{3}p_{2}^{2}+T^{2}p_{1}^{2}p_{2})w_{r}^{(s+1)}\right).
		\end{aligned}$$
		
		For the last term,
		$$\begin{aligned}
			&\|\sum_{t=1}^{T}\Eb_{t}^{\top}(\hat{\Rb}^{(s+1)}-\Rb\hat{\Hb}_{r}^{(s+1)})(\hat{\Rb}^{(s)}-\Rb\hat{\Hb}_{r}^{(s)})^{\top}\Eb_{t}\Cb\|_{F}^{2}\\
			&\leq \|\sum_{t=1}^{T}(\hat{\Rb}^{(s)}-\Rb\hat{\Hb}_{r}^{(s)})^{\top}\Eb_{t}\Cb\Eb_{t}^{\top}\|_{F}^{2}\|\hat{\Rb}^{(s+1)}-\Rb\hat{\Hb}_{r}^{(s+1)}\|_{F}^{2}\\
			&\leq \|\hat{\Rb}^{(s)}-\Rb\hat{\Hb}_{r}^{(s)}\|_{F}^{2}\|\hat{\Rb}^{(s+1)}-\Rb\hat{\Hb}_{r}^{(s+1)}\|_{F}^{2}\sum_{i^{\prime}=1}^{p_{1}}\sum_{j^{\prime}=1}^{p_{2}}\|\sum_{t=1}^{T}\Eb_{t}\Cb e_{t,i^{\prime}j^{\prime}}\|_{F}^{2}\\
			&\leq \|\hat{\Rb}^{(s)}-\Rb\hat{\Hb}_{r}^{(s)}\|_{F}^{2}\|\hat{\Rb}^{(s+1)}-\Rb\hat{\Hb}_{r}^{(s+1)}\|_{F}^{2}\sum_{i^{\prime}=1}^{p_{1}}\sum_{j^{\prime}=1}^{p_{2}}\sum_{i=1}^{p_{1}}(\sum_{t=1}^{T}\sum_{j=1}^{p_{2}}c_{j}e_{t,ij}e_{t,i^{\prime}j^{\prime}})^{2}\\
			&=O_{p}\left((Tp_{1}^{4}p_{2}^{2}+T^{2}p_{1}^{4}p_{2})w_{r}^{(s)}w_{r}^{(s+1)}\right),
		\end{aligned}$$
		where$$\begin{aligned}
			\EE(\sum_{t=1}^{T}\sum_{j=1}^{p_{2}}c_{j}e_{t,ij}e_{t,i^{\prime}j^{\prime}})^{2}&\leq c\sum_{t,s=1}^{T}\sum_{j_{1},j_{2}=1}^{p_{2}}\text{Cov}(e_{t,ij_{1}}e_{t,i^{\prime}j^{\prime}},e_{s,ij_{2}}e_{s,i^{\prime}j^{\prime}})+c(\sum_{t=1}^{T}\sum_{j=1}^{p_{2}}\EE e_{t,ij}e_{t,i^{\prime}j^{\prime}})^{2}\\
			=O\left(Tp_{2}+T^{2}\right).
		\end{aligned}$$
		
	\end{proof}
	
	\begin{lemma} \label{lemma11}
		Suppose that $T, p_{1}, p_{2}$ tend to infinity, and $m_{1}=k_{1}, m_{2}=k_{2}$ are fixed. If  Assumption \ref{Assumption 1} (1), Assumption \ref{Assumption 2} and Assumption \ref{Assumption 4}-\ref{Assumption 7} hold, then we have
		$$(1) \ \|\sum_{t=1}^{T}\Eb_{t}\hat{\Cb}^{(s)}\Fb_{t}^{\top}\|_{F}^{2}=O_{p}\left(Tp_{1}p_{2}^{2}w_{c}^{(s)}+Tp_{1}p_{2}\right).$$
		$$(2) \ \|\sum_{t=1}^{T}\Rb^{\top}\Eb_{t}\hat{\Cb}^{(s)}\Fb_{t}^{\top}\|_{F}^{2}=O_{p}\left(Tp_{1}p_{2}^{2}w_{c}^{(s)}+Tp_{1}p_{2}\right).$$
		$$(3) \ \|\sum_{t=1}^{T}\Eb_{t}^{\top}\hat{\Rb}^{(s)}\Fb_{t}\|_{F}^{2}=O_{p}\left(Tp_{1}^{2}p_{2}w_{r}^{(s)}+Tp_{1}p_{2}\right).$$
		$$(4) \ \|\sum_{t=1}^{T}\Cb^{\top}\Eb_{t}^{\top}\hat{\Rb}^{(s)}\Fb_{t}\|_{F}^{2}=O_{p}\left(Tp_{1}^{2}p_{2}w_{r}^{(s)}+Tp_{1}p_{2}\right)$$
	\end{lemma}
	
	\begin{proof}
		(1)
		$$\begin{aligned}
			\|\sum_{t=1}^{T}\Eb_{t}\hat{\Cb}^{(s)}\Fb_{t}^{\top}\|_{F}^{2}&\leq \|
			\sum_{t=1}^{T}\Eb_{t}(\hat{\Cb}^{(s)}-\Cb\hat{\Hb}_{c}^{(s)})\Fb_{t}^{\top}\|_{F}^{2}+\|\sum_{t=1}^{T}\Eb_{t}\Cb\hat{\Hb}_{c}^{(s)}\Fb_{t}^{\top}\|_{F}^{2}\\
			&\leq \|\sum_{t=1}^{T}\Eb_{t}\Fb_{t}^{\top}\|_{F}^{2}\|\hat{\Cb}^{(s)}-\Cb\hat{\Hb}_{c}^{(s)}\|_{F}^{2}+\|\sum_{t=1}^{T}\Eb_{t}\Cb\Fb_{t}^{\top}\|_{F}^{2}\|\hat{\Hb}_{c}^{(s)}\|_{F}^{2}\\
			&=O_{p}\left(Tp_{1}p_{2}^{2}w_{c}^{(s)}+Tp_{1}p_{2}\right).
		\end{aligned}$$
		
		(2)
		$$\begin{aligned}
			\|\sum_{t=1}^{T}\Rb^{\top}\Eb_{t}\hat{\Cb}^{(s)}\Fb_{t}^{\top}\|_{F}^{2}&\leq \|\sum_{t=1}^{T}\Rb^{\top}\Eb_{t}(\hat{\Cb}^{(s)}-\Cb\hat{\Hb}_{c}^{(s)})\Fb_{t}^{\top}\|_{F}^{2}+\|\sum_{t=1}^{T}\Rb^{\top}\Eb_{t}\Cb\hat{\Hb}_{c}^{(s)}\Fb_{t}^{\top}\|_{F}^{2}\\
			&\leq \|\sum_{t=1}^{T}\Fb_{t}^{\top}\Rb^{\top}\Eb_{t}\|_{F}^{2}\|\hat{\Cb}^{(s)}-\Cb\hat{\Hb}_{c}^{(s)}\|_{F}^{2}+\|\sum_{t=1}^{T}\Rb^{\top}\Eb_{t}\Cb\Fb_{t}^{\top}\|_{F}^{2}\|\hat{\Hb}_{c}^{(s)}\|_{F}^{2}\\
			&=O_{p}\left(Tp_{1}p_{2}^{2}w_{c}^{(s)}+Tp_{1}p_{2}\right).
		\end{aligned}$$
		
		(3) $$\begin{aligned}
			\|\sum_{t=1}^{T}\Eb_{t}^{\top}\hat{\Rb}^{(s)}\Fb_{t}\|_{F}^{2}&\leq \|\sum_{t=1}^{T}\Eb_{t}^{\top}(\hat{\Rb}^{(s)}-\Rb\hat{\Hb}_{r}^{(s)})\Fb_{t}\|_{F}^{2}+\|\sum_{t=1}^{T}\Eb_{t}^{\top}\Rb\hat{\Hb}_{r}^{(s)}\Fb_{t}\|_{F}^{2}\\
			&\leq \|\sum_{t=1}^{T}\Eb_{t}^{\top}\Fb_{t}\|_{F}^{2}\|\hat{\Rb}^{(s)}-\Rb\hat{\Hb}_{r}^{(s)}\|_{F}^{2}+\|\sum_{t=1}^{T}\Eb_{t}^{\top}\Rb\Fb_{t}\|_{F}^{2}\|\hat{\Hb}_{r}^{(s)}\|_{F}^{2}\\
			&=O_{p}\left(Tp_{1}^{2}p_{2}w_{r}^{(s)}+Tp_{1}p_{2}\right).
		\end{aligned}$$
		
		(4) $$\begin{aligned}
			\|\sum_{t=1}^{T}\Cb^{\top}\Eb_{t}^{\top}\hat{\Rb}^{(s)}\Fb_{t}\|_{F}^{2}&\leq \|\sum_{t=1}^{T}\Cb^{\top}\Eb_{t}^{\top}(\hat{\Rb}^{(s)}-\Rb\hat{\Hb}_{r}^{(s)})\Fb_{t}\|_{F}^{2}+\|\sum_{t=1}^{T}\Cb^{\top}\Eb_{t}^{\top}\Rb\hat{\Hb}_{r}^{(s)}\Fb_{t}\|_{F}^{2}\\
			&\leq \|\sum_{t=1}^{T}\Cb^{\top}\Eb_{t}^{\top}\Fb_{t}\|_{F}^{2}\|\hat{\Rb}^{(s)}-\Rb\hat{\Hb}_{r}^{(s)}\|_{F}^{2}+\|\sum_{t=1}^{T}\Cb^{\top}\Eb_{t}^{\top}\Rb\Fb_{t}\|_{F}^{2}\|\hat{\Hb}_{r}^{(s)}\|_{F}^{2}\\
			&=O_{p}\left(Tp_{1}^{2}p_{2}w_{r}^{(s)}+Tp_{1}p_{2}\right).
		\end{aligned}$$
		
	\end{proof}

\end{document}